\documentclass[10pt,reqno]{amsart}
\usepackage[pdftex]{graphicx}
\graphicspath{ {./figures/} }
\usepackage[left=3cm,right=3cm,top=3cm,bottom=3cm]{geometry}
\usepackage{amsmath, amssymb,amsthm}
\usepackage{mathrsfs}
\usepackage[shortlabels]{enumitem}
\usepackage{mlmodern}
\usepackage{bbm}
\usepackage{stmaryrd} 
\usepackage{esint} 
\usepackage{tikz}
\usetikzlibrary{arrows,arrows.meta,patterns,decorations.pathreplacing,external}
\usepackage{mlmodern}
\usepackage{upref} 
\usepackage{trimclip}
\usepackage{subcaption}
\usepackage{hyperref} 
\hypersetup{colorlinks,citecolor=blue,filecolor=blue,linkcolor=blue,urlcolor=navyblue}
\definecolor{navyblue}{rgb}{0.0, 0.0, 0.5}

\newtheorem{thm}{Theorem}[section]
\newtheorem{prop}[thm]{Proposition}
\newtheorem{lem}[thm]{Lemma}
\newtheorem{cor}[thm]{Corollary}
\theoremstyle{definition}
\newtheorem{defn}{Definition}
\newtheorem{rmk}{Remark}[section]

\newtheorem*{claim*}{Claim}

\newcommand{\N}{\mathbb{N}}
\newcommand{\Z}{\mathbb{Z}}
\newcommand{\R}{\mathbb{R}}
\newcommand{\C}{\mathbb{C}}

\newcommand{\T}{\mathbb{T}}
\renewcommand{\Im}{\operatorname{\mathrm{Im}}}
\renewcommand{\Re}{\operatorname{\mathrm{Re}}}
\DeclareRobustCommand{\Chi}{{\mathpalette\irchi\relax}}
\newcommand{\irchi}[2]{\raisebox{\depth}{$#1\chi$}}
\renewcommand{\epsilon}{\varepsilon}
\renewcommand{\hat}{\widehat}
\renewcommand{\tilde}{\widetilde}

\newcommand{\intbrr}[1]{\llbracket#1\rrbracket}

\newcommand{\E}{\mathbb{E}}
\newcommand{\oneb}{\mathbbm{1}}
\renewcommand{\P}{\mathbb{P}}

\newcommand{\floor}[1]{\left\lfloor#1\right\rfloor}

\newcommand\numberthis{\stepcounter{equation}\tag{\theequation}}
\numberwithin{equation}{section}

\newcommand{\x}{\mathbf{x}}
\newcommand{\z}{\mathbf{z}}
\newcommand{\W}{\mathrm{W}}
\newcommand{\da}{DA}
\newcommand{\jj}{\mathfrak{j}}
\newcommand{\Wa}{\mathrm{Wa}}

\begin{document}
\title[Eigenstates of quantized baker's map]{Eigenstates and spectral projection for quantized baker's map}
\author{Laura Shou}
\address{School of Mathematics, University of Minnesota, 206 Church St SE, Minneapolis, MN 55455 USA}
\curraddr{Joint Quantum Institute, Department of Physics, University of Maryland, College Park, MD 20742}
\email{lshou@umd.edu}
\begin{abstract}
We extend the approach from \cite{pw} to prove windowed spectral projection estimates and a generalized Weyl law for the (Weyl) quantized baker's map on the torus. The spectral window is allowed to shrink in the semiclassical (large dimension) limit. 
As a consequence, we obtain a strengthening of the quantum ergodic theorem from \cite{DNW} to hold in shrinking spectral windows, a Weyl law on uniform spreading of eigenvalues, and statistics of random quasimodes. 
Using similar techniques, we also investigate random eigenbases of a different (non-Weyl) quantization, the Walsh-quantized baker's map, which has high degeneracies in its spectrum. For such random eigenbases, we prove that Gaussian eigenstate statistics and QUE hold with high probability in the semiclassical limit.
\end{abstract}

\maketitle

\section{Introduction}

Quantum chaos seeks to understand the relationship between classically chaotic dynamical systems and their quantum counterparts. The correspondence principle from quantum mechanics suggests the classical behavior should manifest itself in the associated quantum system in the semiclassical limit.
It is conjectured that spectra and eigenfunctions of a quantum system associated with a classically chaotic system should reflect such behavior in the following ways: the eigenvalues should generally exhibit random matrix ensemble spectral statistics (BGS conjecture \cite{bgs}), and the eigenfunctions should behave like random waves (Berry random wave conjecture \cite{berry}). 

Mathematically, one manifestation of the classical-quantum correspondence principle is given by the quantum ergodic theorem of Shnirelman--Zelditch--Colin de Verdi\`ere \cite{shnirelman,deverdiere,zelditch}. In this setting, classical dynamics are described by geodesic flow $\varphi_t$ on the unit tangent bundle of a manifold $M$, and the associated quantum Hamiltonian is the Laplacian.  For $\varphi_t$ ergodic, the quantum ergodic theorem guarantees, in the large eigenvalue limit, a density 1 subsequence of Laplace eigenfunctions on $M$ that equidistribute in all of phase space.  Similar quantum ergodic properties for many other models have also been investigated, including for Hamiltonian flows \cite{HMR}, torus maps \cite{bdb,KurlbergRudnick,scar-cat,DBDE,MOK,DNW}, and graphs \cite{qgraphs, alm,as,anantharaman}.
The large variety of models on the torus and on graphs have been popular as simpler models for studying quantum chaos, as they tend to have minimal technical complications, yet can still exhibit fundamental or generic quantum chaotic behavior.

In most of this paper, we study the Balazs--Voros quantization \cite{BV}  of the baker's map on the two-torus $\T^2=\R^2/\Z^2$. We will also refer to this quantization as a Weyl-quantized baker map, or sometimes for convenience simply as quantized baker's map, since it was the first and is generally one of the most well-known quantizations of the baker's map. Along with cat maps, baker maps have been among the most well-studied models for quantum chaos on the torus. This quantization and other quantizations of the baker map have been of interest in both the physics and mathematics literature, e.g. \cite{BV,Saraceno,SaracenoVoros,RubinSalwen,DBDE,DNW,AN,bakergap}, as well as in quantum computing \cite{exp2,exp,ScottCaves,modmul}.
As noted in \cite{BV,Saraceno}, this type of baker's map quantization appears to suffer less from number theoretical degeneracies than the quantum cat maps, for which the eigenspaces can be highly degenerate. In particular, the spectrum for the Balazs--Voros quantization appears to be non-degenerate, and in the symmetrized Saraceno quantization \cite{Saraceno}, for most dimensions one even numerically recovers random matrix ensemble (COE) level spacing statistics in agreement with the BGS conjecture, after separating the spectra by symmetry class \cite{Saraceno}.

To introduce quantized baker's map, we start by briefly describing classical functions and maps on the torus.
A more detailed overview of the quantization will be given in Section~\ref{subsec:quant}.
The torus $\T^2$ is viewed as a classical phase space, with 1D coordinates of position $q$ and momentum $p$, which are required to behave periodically. These dimensions correspond to having a (2D) surface as phase space. 
Classical observables are smooth functions on the torus, $a\in C^\infty(\T^2)$. 
The classical baker's map $B:\T^2\to\T^2$ is the ergodic measure-preserving transformation defined by
\begin{align}\label{eqn:baker}
B(q,p) &= \begin{cases}
(2q,\frac{p}{2}), & 0\le q<\frac{1}{2}\\
(2q-1,\frac{p}{2}+\frac{1}{2}), &\frac{1}{2}\le q<1.
\end{cases}.
\end{align}
The position coordinate $q$ always is sent to $2q\mod 1$, while the momentum coordinate $p$ is sent to either $\frac{p}{2}$ or $\frac{p}{2}+\frac{1}{2}$, depending on whether $q$ is in the left or right half of $\T^2$. 
The map is named for its similarity to the process of kneading dough; the unit square is first stretch lengthwise, then cut in half and stacked back into the unit square (Figure~\ref{fig:baker}). 
By encoding the coordinates $(q,p)\in\T^2$ as an infinite binary sequence, one can also see the baker's map is equivalent to a two-sided Bernoulli shift.

\begin{figure}[!ht]
\centering
\includegraphics{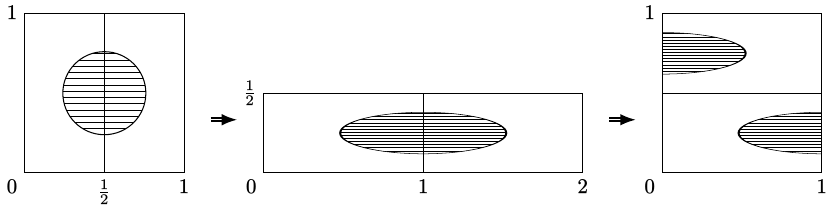}
\caption{Classical baker's map operation.}
\label{fig:baker}
\end{figure}

Quantization of the torus $\T^2$ associates for each $N\in\N$ an $N$-dimensional Hilbert space $\mathcal{H}_N$ of quantum states. Operators on $\mathcal{H}_N$ can thus be represented as $N\times N$ matrices. In terms of the semiclassical parameter $\hbar$, we have $N=(2\pi\hbar)^{-1}$, and so the semiclassical limit $\hbar\to0$ is the limit $N\to\infty$.
Given a classical observable $a\in C^\infty(\T^2)$, one has the Weyl quantization (derived from the usual Weyl quantization on $\R^2$), which for each $N$ identifies the classical observable $a$ with a quantum operator $\operatorname{Op}_N^\W(a):\mathcal{H}_N\to\mathcal{H}_N$.

For the baker's map $B:\T^2\to\T^2$, a quantization was first proposed by Balazs and Voros in \cite{BV}.
For $N\in2\N$, this is a unitary $N\times N$ matrix $\hat{B}_N$ 
which recovers\footnote{This classical-quantum correspondence can be stated precisely as an Egorov theorem, which states that the quantum evolution $\hat{B}_N^{-1}\operatorname{Op}_N^\W(a)\hat{B}_N$ should be well-approximated by the quantization of the classical evolution, $\operatorname{Op}_N^\W(a\circ B)$, as $N\to\infty$ for suitable observables $a$ supported away from discontinuities of $B$. 
The precise Egorov's theorem for $\hat{B}_N$ as proved in \cite{DNW} is stated in Theorem~\ref{thm:egorov}.} the classical map $B$ in the semiclassical limit $N\to\infty$. 
In the position basis, this quantized baker's map is defined as the $N\times N$ matrix
\begin{align}\label{eqn:bv}
\hat{B}_N &= \hat{F}_N^{-1}\begin{pmatrix}\hat{F}_{N/2}&0\\0&\hat{F}_{N/2}\end{pmatrix},
\end{align} 
where $(\hat{F}_N)_{jk}=\frac{1}{\sqrt{N}}e^{-2\pi i jk/N}$ for $j,k\in\intbrr{0:N-1}$ is the discrete Fourier transform (DFT) matrix. 
We note that \eqref{eqn:bv} is not the only possible quantization of the baker map; various other quantizations with certain properties have also been proposed and studied, e.g. \cite{Saraceno,RubinSalwen,SchackCaves}. 
In Section~\ref{sec:walshintro} and Section~\ref{sec:walsh}, we will also look at an alternative (non-Weyl) quantization method for the torus and baker map that was studied in \cite{SchackCaves,TracyScott,NonnenmacherZworski,AN}. 
While most of this article will focus on the Balazs--Voros quantization \eqref{eqn:bv}, similar techniques also apply to other related Weyl quantizations such as the symmetrized Saraceno quantization \cite{Saraceno}.

Due to discontinuities of the classical baker's map, which is discontinuous at the unit square boundaries as well as along the vertical line $q=1/2$, there are some complications with the classical-quantum correspondence (Egorov theorem). The correspondence can only hold for observables supported away from the discontinuities, otherwise there can be diffraction effects \cite{DNW}.
This type of correspondence was proved and was still sufficient to prove a quantum ergodic theorem \cite{DNW}.
The quantum ergodic theorem for this model states that a limiting density one subset of the eigenvectors in an orthonormal basis of $\hat{B}_N$ equidistribute in phase space, in the following sense: Selecting, for each $N$, an orthonormal eigenbasis $\{\varphi^{(j,N)}\}_j$ of $\hat{B}_N$, there are sets $\Lambda_N\subseteq\intbrr{0:N-1}$ with $\frac{\#\Lambda_N}{N}\to1$ as $N\to\infty$ so that for any sequence $(j_N\in\Lambda_N)_{N\in2\N}$ and $a\in C^\infty(\T^2)$,
\begin{align}\label{eqn:qe}
\lim_{N\to\infty}\langle \varphi^{(j_N,N)}|\operatorname{Op}_N^\W(a)|\varphi^{(j_N,N)}\rangle &= \int_{\T^2}a(\x)\,d\x.
\end{align}
This was proved in \cite{DNW} (and earlier for observables of position only, $a=a(q)$, in \cite{DBDE}).
As usual, this quantum ergodic theorem allows for a limiting density zero set of exceptional eigenvectors that may fail to equidistribute. Such exceptional eigenvectors are often expected to concentrate, or scar, near periodic orbits of the classical map. The property of having \emph{all} eigenvectors equidistribute in the sense of satisfying \eqref{eqn:qe} was termed ``quantum unique ergodicity'' (QUE) in \cite{RudnickSarnak}, and is still a major question for many quantum chaotic systems.

\vspace{2.5mm}
In this article, we study the eigenvectors of the quantized baker map $\hat{B}_N$.
We first prove spectral estimates for $\hat{B}_N$, including spectral projection estimates and a generalized Weyl law for eigenvectors corresponding to eigenvalues in a shrinking arc on the unit circle.
By taking spectral projections in shrinking arcs as $N\to\infty$, we will obtain a strengthening of the quantum ergodicity statement \eqref{eqn:qe}, which will hold for a limiting density one set of eigenvectors within a shrinking spectral window. This guarantees phase space equidistribution for a limiting density one set of eigenvectors \emph{within} a limiting density zero set. 
This puts a limit on how many of the possible exceptional non-equidistributing eigenstates can accumulate in a small eigenvalue range: if one is searching for such exceptional eigenstates, then there cannot be too many all belonging to too small a spectral window. 
As another application of the windowed spectral estimates, we will also prove properties on the statistics of random states (quasimodes) within these spectral arcs. It is expected due to random matrix statistics and the random wave conjecture that actual eigenvectors should have random Gaussian statistics. Due to the difficulty of studying  statistics of non-random eigenstates, one often considers random linear combinations of eigenstates with nearby eigenvalues as a proxy average for the actual eigenstates, e.g. see discussion in \cite{HejhalRackner,NonnenmacherAnatomy,Canzani}.
We will show such random wave analogues for $\hat{B}_N$ have many of the desired statistics and properties, such as Gaussian value statistics, equidistribution in phase space, and uniform expected sign change distribution for its real and imaginary parts, with high probability as $N\to\infty$.

Secondly, we will look at an alternate Walsh (non-Weyl) quantization of the baker map studied in \cite{SchackCaves,TracyScott,NonnenmacherZworski,AN} which has highly degenerate eigenspaces. Due to high spectral degeneracies, we can study statistics of actual (randomly chosen) eigenbases, rather than just those of random quasimodes. In particular, we will prove Gaussian value statistics and QUE for actual eigenbases with high probability (according to Haar measure in each eigenspace) as $N\to\infty$. We note that explicit non-equidistributing eigenstates were constructed in \cite{AN}, so this result will show that such choices of non-QUE eigenstates must be rare.

\subsection{Outline}
In Section~\ref{sec:main} we present the main results of the paper.
In Section~\ref{sec:setup}, we introduce background on torus quantization, define several useful subsets of coordinate pairs in $\intbrr{0:N-1}^2$, and give the proof outlines for Theorems~\ref{thm:Pmat} and \ref{thm:lweyl} on windowed spectral estimates and generalized Weyl law.

Sections~\ref{sec:Bpowers-proof}--\ref{sec:rw} contain the proofs of the main results concerning the Balazs--Voros baker map quantization \eqref{eqn:bv}. More specifically, the sections are divided as follows.
\begin{itemize}
\item Sections~\ref{sec:Bpowers-proof}--\ref{sec:p-sf-proof}: Proof of Theorem~\ref{thm:Pmat} for the spectral projection. 
\item Section~\ref{sec:lweyl-proof}: Proof of Theorem~\ref{thm:lweyl}, the windowed generalized Weyl law. The application to proving windowed quantum ergodicity is given in Section~\ref{subsec:qe} and Appendix~\ref{sec:wqe}.
\item Section~\ref{sec:rw}: Proof of Theorem~\ref{thm:rw} on random quasimode properties.
\end{itemize}
Finally, in Section~\ref{sec:walsh}, we discuss the Walsh-quantized baker map and prove Theorem~\ref{thm:walsh-main} on statistics of randomly chosen eigenbases.

\section{Main results}\label{sec:main}
Our main results are as follows. We first have two spectral estimates: pointwise estimates on spectral projection and spectral functions (Theorem~\ref{thm:Pmat}), and a windowed generalized Weyl law (Theorem~\ref{thm:lweyl}).
As a consequence we then obtain several applications: a Weyl law for uniform spreading of eigenvalues (Corollary~\ref{cor:weyl}), a windowed quantum ergodic theorem (Theorem~\ref{thm:qv}, Corollary~\ref{cor:qe}),
and properties of random quasimodes (Theorem~\ref{thm:rw}). All of the above is for the quantization \eqref{eqn:bv}, although similar arguments also apply to related Weyl quantizations such as the Saraceno quantization \cite{Saraceno}.
In the last subsection here, we introduce the Walsh quantization and Walsh quantized baker map, and state Theorem~\ref{thm:walsh-main} on statistics of random eigenbases for the Walsh quantized baker map.

Several of the methods we use are related to \cite{pw}, where we studied eigenvectors of quantum graph models associated with 1D ergodic interval maps. 
The baker's map quantizations here however require more involved analysis and some different methods to prove the necessary spectral estimates and time evolution behavior.

\subsection{Windowed spectral function} Here we consider eigenvectors with eigenvalues in a spectral window $I(N)\subset\R/(2\pi\Z)$, which is allowed to shrink as $N\to\infty$. We will always assume that eigenvectors are orthonormal, and that $N$ is even so \eqref{eqn:bv} is defined.

\begin{thm}[windowed spectral projection]\label{thm:Pmat} 
Let $N\in2\N$, and let $(e^{i\theta^{(j,N)}},\varphi^{(j,N)})_{j}$ be  eigenvalue-eigenvector pairs corresponding to an orthonormal eigenbasis $(\varphi^{(j,N)})_j$ of $\hat{B}_N$. Suppose $(I(N))_{N\in2\N}$ is a sequence of intervals in $\R/(2\pi \Z)$ such that
$|I(N)|\log N\to\infty$ as $N\to\infty$.
Let $P_{I(N)}$ be the spectral projection matrix of $\hat{B}_N$ on the interval $I(N)$,
\begin{align*}
P_{I(N)}=\sum_{j:\theta^{(j,N)}\in I(N)}|\varphi^{(j,N)}\rangle\langle\varphi^{(j,N)}|,
\end{align*}
where $|\varphi^{(j,N)}\rangle\langle\varphi^{(j,N)}|$ is the orthogonal projection onto the eigenstate $\varphi^{(j,N)}$.
Then for at least $N(1-o(|I(N)|))$ coordinates $x\in\intbrr{0:N-1}$, we have the pointwise estimate
\begin{align}\label{eqn:P-diag}
(P_{I(N)})_{xx}&=\frac{|I(N)|}{2\pi}(1+o(1)),
\end{align}
and for at least $N^2(1-o(|I(N)|))$ pairs $(x,y)\in\intbrr{0:N-1}^2$, we have the bound,
\begin{align}\label{eqn:P-offdiag}
(P_{I(N)})_{xy}&=o(|I(N)|),\; x\ne y,
\end{align}
with asymptotic decay rates uniform over the allowable $x,y$, and the location of $I(N)$.
The points $x$ and pairs $(x,y)$ to avoid are taken independent of the location of $I(N)$, and will correspond to coordinates near short-time forward or backward iterations of the graph of the classical baker's map (cf. \S\ref{subsec:sets}, Fig.~\ref{fig:avoid}), or near classical discontinuities. More precisely, in terms of sets to be defined in Section~\ref{subsec:sets} and parameters to be defined in \eqref{eqn:param2}, Eq.~\eqref{eqn:P-diag} holds for $x\not\in\da_{J,\delta,\gamma,N}^W$, and \eqref{eqn:P-offdiag} holds for $(x,y)\not\in\tilde{A}_{J,\delta,\gamma,N}^W$.
\end{thm}

Taking the trace of $P_{I(N)}$ using \eqref{eqn:P-diag} will then give the eigenvalue counting function for the window $I(N)$, which produces a more common Weyl law (Corollary~\ref{cor:weyl}). The estimates \eqref{eqn:P-diag} and \eqref{eqn:P-offdiag} are a pointwise (or local) Weyl law, the name coming from the literature on asymptotics of the spectral projection kernel for the Laplacian on Riemannian manifolds, e.g. see \cite{Zelditch-book}.
To avoid confusion with the generalized Weyl law and Weyl law in the next sections, we will generally refer to \eqref{eqn:P-diag} and \eqref{eqn:P-offdiag} just as spectral projection estimates.

\begin{rmk}
\begin{enumerate}[(i)]

\item For $q_N:\R/(2\pi\Z)\to\C$ a sequence of $C^2$ functions with $\|q_N''\|_\infty=o(\log N)$, one can extend the pointwise estimates of Theorem~\ref{thm:Pmat} to the operator
\begin{align*}
Q_{N,I(N)}:=(q_N\Chi_{I(N)})(\hat{B}_N)=\sum_{\theta^{(j,N)}\in I(N)}q_N(\theta^{(j,N)})|\varphi^{(j,N)}\rangle\langle\varphi^{(j,N)}|,
\end{align*}
where $\Chi_{I(N)}$ is the indicator function of the interval $I(N)$.
This is stated and proved in Appendix~\ref{sec:qN}.

\item Position vs momentum basis:
Theorem~\ref{thm:Pmat} is written for coordinates in the position basis $\{|x\rangle\}_{x=0}^{N-1}$, but the same results hold in momentum basis  $\{|p\rangle\}_{p=0}^{N-1}$ (defined in Section~\ref{subsec:quant}), as the matrix $\hat{B}_N$ in momentum basis is just the transpose, $\hat{B}_{N,\text{momentum}}=\hat{F}_N\hat{B}_N\hat{F}_N^{-1}=\hat{B}_N^T$, and so Theorem~\ref{prop:mpowers} applies in this case with just $x$ and $y$ swapped.

\item We show in Proposition~\ref{lem:notall} that in general the estimate \eqref{eqn:P-diag} cannot hold for all $x$; one does indeed need to exclude some coordinates. We expect the same to hold for off-diagonal estimates as well (cf. Figure~\ref{fig:P}).

\end{enumerate}
\end{rmk}

The main points of Theorem~\ref{thm:Pmat} are that the window $I(N)$ is allowed to shrink, and that the estimates are pointwise rather than a trace. Both of these more precise results will be necessary for the applications of a stronger quantum ergodic theorem and statistics of random quasimodes.
The estimates where one considers the trace $\operatorname{Tr}P_I$ with a fixed (non-shrinking) spectral window $I$ were proved for a general class of operators on the torus in \cite{MOK}. However, for the shrinking window $I(N)$ and the pointwise estimates that are required here, we will need more precise information on the time evolution of our specific operator $\hat{B}_N$.
Additionally, since the cut-off function $\Chi_{I(N)}$ is not continuous, its approximation by smooth functions is a bit delicate, in that the choice will determine the condition $|I(N)|\log N\to\infty$, which requires that $I(N)$ does not shrink too rapidly.

The proof of Theorem~\ref{thm:Pmat} will start by using the approximation used in \cite{pw}; we approximate $\Chi_{I(N)}$ by certain trigonometric polynomials to write,
\[
(P_{I(N)})_{xy}\approx \delta_{xy}\frac{|I(N)|}{2\pi}(1+o(1))+\sum_{|k|\ge 1} a_k(\hat{B}_N^k)_{xy}.
\]
We then need to understand the matrix entries $(\hat{B}_N^k)_{xy}$, or time evolution of $\hat{B}_N$, for short times $|k|$ up to (just below) $\sim\log N$.
The proof method for diagonal entries in \cite{pw} relied on knowing the matrix elements for powers $k$ up to the Ehrenfest time $\sim\log N$, and that the large matrix elements only occurred near coordinates corresponding to periodic orbits of the classical map. As pictured in Figure~\ref{fig:Bpowers}, a related, though more dispersed behavior, occurs for these matrices $\hat{B}_N$. Small powers trace out the classical map of the transformation $q\mapsto 2q\mod 1$ on position $q$, up to what we will consider an Ehrenfest-like time $J$, where $J$ will be of order just below $\log N$. We will prove the general behavior seen in the specific case of Figure~\ref{fig:Bpowers} in  Theorem~\ref{prop:mpowers}, using the evolution of coherent states under the matrices $\hat{B}_N$. This behavior of powers of $\hat{B}_N$ will explain the structure of the spectral projection matrices $P_{I(N)}$ shown in Figure~\ref{fig:P} and proved as Proposition~\ref{prop:mp-sf}.

We note two comparisons to the graph model we studied in \cite{pw}: First, the matrices $\hat{B}_N$ here are more complicated than those considered in the graph model, resulting in more complicated time iterates $\hat{B}_N^k$. Second, the windowed local Weyl law Theorem~\ref{thm:lweyl} requires more work than the diagonal equivalent in \cite{pw}, which in that case was a quick consequence of the quantization method from \cite{qgraphs}. The quantum observables  there were diagonal  since there was no momentum in the graph quantization, and so we only cared about how many {on-diagonal} entries of the spectral projection matrix were asymptotically $\frac{|I(N)|}{2\pi}$. In contrast, since here we have the standard Weyl quantization involving a usual position and momentum phase space, in order to use the projection estimates to prove a windowed local Weyl law, we will need off-diagonal estimates as well as the locations of the possible exceptional off-diagonal entries. 
The structure of the possible exceptional coordinates will allow us to prove Theorem~\ref{thm:lweyl}.

\begin{figure}[!ht]
\centering
\includegraphics[trim={2.1cm 3.2cm 2.1cm 3.5cm},clip,width=6in]{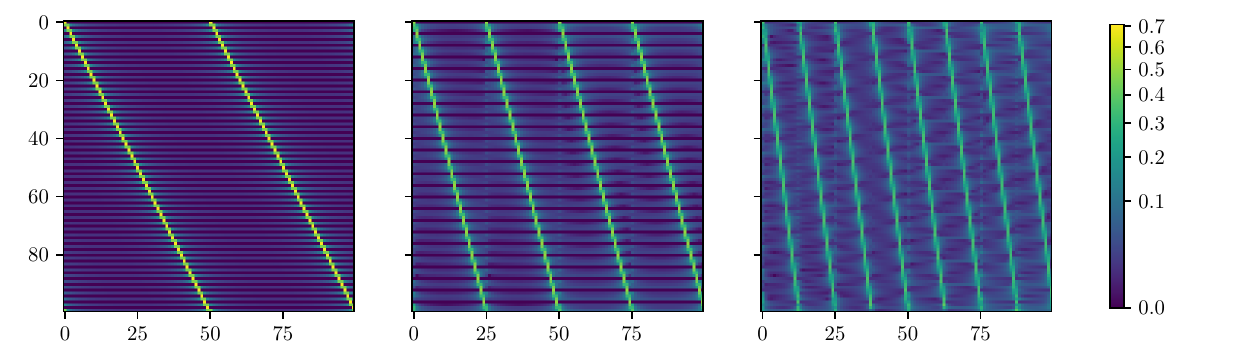}
\caption{The absolute value of the matrix entries of $\hat{B}_N^k$, for $N=100$ and $k=1,2,3$, plotted on a power scale. For these small powers $k$, the large matrix entries trace out the classical map $x\mapsto 2^kx\;\mathrm{mod}\;1$ (flipped vertically). However, as $k$ becomes larger and reaches the Ehrenfest time, the relation to the classical map begins to collapse, and the matrix entry patterns begin to look fairly random. 
(By the Ehrenfest time, the graph of $x\mapsto 2^kx\;\mathrm{mod}\;1$ fills up the entire grid.)
As $N$ increases, one can allow longer times $k$ before the collapse.
}
\label{fig:Bpowers}
\end{figure}

\subsection{Windowed generalized Weyl law} 

The pointwise estimates in Theorem~\ref{thm:Pmat} will be used to prove Theorem~\ref{thm:lweyl} below. This windowed local Weyl law will be one of the steps for proving windowed quantum ergodicity.

In what follows, the quantum operator $\operatorname{Op}_N^\W(f)$, expressible as an $N\times N$ matrix, is the Weyl quantization of the classical observable $f:\T^2\to\C$, and will be defined in Section~\ref{subsec:quant}.

\begin{thm}[windowed generalized Weyl law]\label{thm:lweyl}
Let $(I(N))_{N\in2\N}$ be a sequence of intervals in $\R/(2\pi \Z)$ such that
$
|I(N)|\log N\to\infty$, as $N\to\infty$.
Let $(e^{i\theta^{(j,N)}},\varphi^{(j,N)})_j$ be eigenvalue-eigenvector pairs corresponding to an orthonormal basis of the $N\times N$ unitary matrix $\hat{B}_N$. Then in the semiclassical limit $N\to\infty$, for any classical observable $f\in C^\infty(\T^2)$, there is the {windowed} local Weyl law,
\begin{align}
\frac{2\pi}{N|I(N)|}\sum_{\theta^{(j,N)}\in I(N)}\langle\varphi^{(j,N)}|\operatorname{Op}_N^\W(f)|\varphi^{(j,N)}\rangle &= \int_{\T^2}f(q,p)\,dq\,dp +o(1).
\end{align}
The $o(1)$ term depends only on $N$, $\|f\|_{C^3}$, and $|I(N)|$.
\end{thm}
More generally, using Theorem~\ref{thm:Qmat}, for any sequence of continuous $q_N:\R/(2\pi\Z)\to\C$ with $\|q_N''\|_\infty=o(\log N)$, we have the windowed generalized local Weyl law,
\begin{multline}\label{eqn:lweyl-q}
\frac{2\pi}{N|I(N)|}\sum_{\theta^{(j,N)}\in I(N)}q_N(\theta^{(j,N)})\langle\varphi^{(j,N)}|\operatorname{Op}_N^\W(f)|\varphi^{(j,N)}\rangle = \\
=\int_{\T^2}f(q,p)\,dq\,dp \fint_{I(N)}q_N(z)\,\frac{dz}{2\pi} +o(1)(1+\|q_N\|_\infty),
\end{multline}
where $\fint_{I(N)}q_N(z)\,dz:=\frac{1}{|I(N)|}\int_{I(N)}q_N(z)\,dz$.
(See Appendix~\ref{sec:qN}.)

\subsection{Uniform spreading of eigenvalues} 

As a consequence of the diagonal bound in Theorem~\ref{thm:Pmat} (and in Theorem~\ref{thm:Qmat}), we obtain,
\begin{cor}[Weyl law/eigenvalue counting]\label{cor:weyl}
For any sequence of intervals $(I(N))_{N\in2\N}$ in $\R/(2\pi\Z)$ with $|I(N)|\log N\to\infty$,
\begin{align}\label{eqn:weyl-ev}
\#\{j:\theta^{(j,N)}\in I(N)\} &=\frac{N|I(N)|}{2\pi}(1+o(1)).
\end{align}
Additionally, for any sequence of $C^2$ functions $q_N:\R/(2\pi\Z)\to\C$ with $\|q_N''\|_\infty =o(\log N)$,
\begin{align}\label{eqn:weyl-qn}
\sum_{\theta^{(j,N)}\in I(N)}q_N(\theta^{(j,N)}) &= \frac{N|I(N)|}{2\pi}\left(\fint_{I(N)}q_N(z)\,dz+o(1)(1+\|q_N\|_\infty)\right).
\end{align}
\end{cor}
Equation~\eqref{eqn:weyl-ev} shows that the eigenvalues of $\hat{B}_N$ are, down to the shrinking scale determined by $|I(N)|\to0$, uniformly distributed on the unit circle. 

\begin{rmk}
This type of eigenvalue counting Weyl law remainder and derivation using dynamics up to a logarithmic time is similar to the result of B\'erard \cite{Berard77} for the Laplacian on negatively curved manifolds.
\end{rmk}

\subsection{Windowed quantum ergodicity}
Next, we discuss the applications of the above spectral results to eigenstates, first quantum ergodicity in a shrinking spectral windows, and second, properties of random quasimodes.  As a consequence of Theorem~\ref{thm:lweyl} and Corollary~\ref{cor:weyl}, we will obtain,
\begin{thm}[windowed quantum variance]\label{thm:qv}
Let $(I(N))_{N\in2\N}$ be a sequence of intervals in $\R/(2\pi\Z)$ with $|I(N)|\log N\to\infty$, and let $a\in C^\infty(\T^2)$. Then for $N\in2\N$, 
\begin{align}
\lim_{N\to\infty}\frac{2\pi}{N|I(N)|}\sum_{\theta^{(j,N)}\in I(N)}\bigg|\langle\varphi^{(j,N)}|\operatorname{Op}_N^\W(a)|\varphi^{(j,N)}\rangle -\int_{\T^2}a(q,p)\,dq\,dp\bigg|^2 &=0.
\end{align}
\end{thm}
This strengthens the quantum ergodic theorem proved in \cite{DNW} to hold in a shrinking window $I(N)$. 

One can get an explicit decay rate  on the quantum variance above depending on $N$ and $C^M$ norms of $a$ (Section~\ref{subsec:qe-proof}), however the rate is very slow since we only guarantee rather slow convergence in Theorem~\ref{thm:lweyl} (see Proposition~\ref{thm:window}).

Decay of the windowed quantum variance then yields windowed quantum ergodicity,
\begin{cor}[windowed quantum ergodicity]\label{cor:qe}
Let $(I(N))_{N\in2\N}$ be a sequence of intervals in $\R/(2\pi\Z)$ with $|I(N)|\log N\to\infty$. For each $N\in2\N$, there is a subset of indices $\Lambda_N\subseteq\{j:\theta^{(j,N)}\in I(N)\}$ with $\frac{\#\Lambda_N}{\#\{j:\theta^{(j,N)}\in I(N)\}}\xrightarrow{N\to\infty}1$, such that for any $a\in C^\infty(\T^2)$ and any sequence $(j_N\in\Lambda_N)_{N\in2\N}$,
\begin{align}
\lim_{N\to\infty}\langle \varphi^{(j_N,N)}|\operatorname{Op}_N^\W(a)|\varphi^{(j_N,N)}\rangle &= \int_{\T^2}a(\x)\,d\x.
\end{align}
\end{cor}
As mentioned in the introduction, 
this is a statement about phase space equidistribution 
for eigenvectors in a sequence of sets that is limiting density \emph{zero} in the entire set of eigenvectors.
(By Corollary~\ref{cor:weyl} the set $\{\varphi^{(j,N)}:\theta^{(j,N)}\in I(N)\}$ is limiting density zero when $|I(N)|\to0$.) 
Thus one cannot have too many exceptional eigenvectors that fail to equidistribute all clustered in too small a spectral window.

\begin{rmk}
On manifolds, for Hamiltonian flows which are ergodic on energy shells in some range $[E_1,E_2]$, one considers eigenfunctions $u_j=u_j(\hbar)$ of the quantum operator $-\hbar^2\Delta+V$, for $\Delta$ the Laplace--Beltrami operator and $V$ a potential.
In the quantum ergodic theorem, one considers eigenvalues in $[E_1,E_2]$, or in a smaller shrinking window $[E,E+\hbar]$ as in \cite{HMR}, \cite[Appendix D]{DyGu}. 
In both of these cases (the former with an additional symbol averaging condition \cite{zworski}) the $\hbar\to0$ limit of the quantum expectation values $\langle u_{j}|\operatorname{Op}_\hbar(a)|u_{j}\rangle$ along a density one subsequence is an integral of the classical symbol $a$ over the corresponding subset of phase space, e.g. $\{(x,\xi):E_1\le|\xi|^2+V(x)\le E_2\}$ or $\{(x,\xi):|\xi|^2+V(x)=E\}$, 
and control of the quantum dynamics is only needed up to an $\hbar$-independent time.

In our case, we require control of the quantum dynamics up to a time close to logarithmic in $N$ to allow for smaller spectral intervals. 
To our knowledge, having any type of shrinking spectral window is new for discrete time unitary maps on the torus.
For such maps, quantum ergodic theorems are typically stated in the limit $N\to\infty$ over all $N$ eigenvalues on the unit circle, and with a single fixed limiting value $\int_{\T^2}a(q,p)\,dq\,dp$. We note that a quantum ergodic statement in just a \emph{fixed} spectral window $I(N)=[\alpha,\beta]$ in this case would not be a strengthening, as there are of order $\frac{\beta-\alpha}{2\pi}N=cN$ eigenvectors with eigenangles in $[\alpha,\beta]$,
and a limiting positive density set out of $cN$ states is also limiting positive density out of all $N$ states.
Thus allowing the spectral window $I(N)$ to shrink in Theorems~\ref{thm:Pmat} and \ref{thm:lweyl} is necessary to produce a stronger quantum ergodic theorem.

We also mention the work of \cite{Keeler}, which allows a similarly shrinking spectral window, for the Laplacian on manifolds without conjugate points, with applications to random band-limited waves.
\end{rmk}

\subsection{Random quasimodes}
For the second eigenvector application, we consider random quasimodes, here meaning random linear combinations of eigenvectors corresponding to eigenvalues in a small spectral window. 
For example, one can take a shrinking interval $[\theta,\theta+o(1)]$ which limits to the single point $\theta$, although in general the interval does not need to be fixed around a specific $\theta$. The following properties are consequences of the {pointwise} estimates in Theorem~\ref{thm:Pmat}. 
The spectral projection matrix $P_{I(N)}$ is the covariance matrix (up to scaling) of the random state, and the asymptotics in Theorem~\ref{thm:Pmat} are enough to conclude several statistical properties of these states. 

For use below, we recall a standard complex Gaussian random variable $g\sim N_\C(0,1)$ is one whose real and imaginary parts are described by independent real $N(0,1/2)$ Gaussian random variables.
We also recall that a sequence of probability measures $(\nu_N)_{N}$ on $\R$ or $\C$
is said to converge weakly to a limiting probability measure $\nu_\infty$ if $\E_{\nu_N}f\to \E_{\nu_\infty} f$ as $N\to\infty$ for all bounded continuous $f$.
For a sequence of \emph{random} probability measures $(\mu_N)_N$, the sequence $(\mu_N)_N$ is said to converge weakly in probability to a non-random
measure $\mu$ if for any bounded continuous $f$ and any $\varepsilon>0$,
\begin{align*}
\P\left[\left|\int f\,d\mu_N-\int f\,d\mu\right|>\varepsilon\right]\xrightarrow{N\to\infty}0.
\end{align*}

\begin{thm}[random quasimodes]\label{thm:rw}
Let $(I(N))_{N\in2\N}$ be a sequence of intervals in $\R/(2\pi\Z)$ with $|I(N)|\log N\to\infty$.
Define the subspace $S_{I(N)}=\operatorname{span}(\varphi^{(j,N)}:\theta^{(j,N)}\in I(N))$, and let $g_i$, for $i=1,\ldots,{\operatorname{dim}S_{I(N)}}$, be iid $N_\C(0,1)$ random variables. Define the random quasimode $\psi=\psi_N$ as
\begin{align}
\psi(x) &= \frac{1}{\sqrt{\dim S_{I(N)}}}\sum_{j:\theta^{(j,N)}\in I(N)} g_j \varphi^{(j,N)}(x),
\end{align}
and let $\Omega_N$ be the probability space from which $\psi$ is drawn.
Then as $N\to\infty$,
\begin{enumerate}[(i)]
\item The vector $\psi_N$ has approximately Gaussian value statistics with high probability: the empirical distribution $\mu_{N} =\frac{1}{N}\sum_{x=0}^{N-1}\delta_{\sqrt{N}\psi_N(x)}$ of the scaled coordinates $\sqrt{N}\{\psi_N(x)\}_{x=0}^{N-1}$ converges weakly in probability to the standard complex Gaussian measure as $N\to\infty$.

\item With high probability, the $(\psi_N)$ equidistribute in all of phase space; more precisely, there are sets $\Gamma_N\subseteq\Omega_N$ with $\P[\Gamma_N]=1-o(1)$ so that for any sequence  $(\psi_N\in\Gamma_N)_N$, then
\begin{align}\label{eqn:equidist}
\lim_{N\to\infty}\langle\psi_N|\operatorname{Op}_N^\W(a)|\psi_N\rangle=\int_{\T^2}a(\x)\,d\x,\quad\forall a\in C^\infty(\T^2).
\end{align}

\item The scaled coordinates $\sqrt{N}\{\psi_N(x)\}_{x=0}^{N-1}$ have moments $\E|\sqrt{N}\psi_N(x)|^{m}$ and autocorrelation functions $\E[N^2|\psi_N(x)|^2|\psi_N(y)|^2]$ that agree to leading order with those of the standard complex Gaussian vector $N_\C(0,I_N)$ for nearly all $x,y$. In general, however, this standard Gaussian behavior need not hold for all coordinates $x,y$.

\item $\ell^p$ norms: For $g\sim N_\C(0,1)$ and $p>0$,
\begin{align*}
\E\|\psi_N\|_{p}^p &= \frac{\E|g|^p}{N^{p/2-1}}\left[1+o(1)\right],
\quad \E\|\psi_N\|_\infty\le \frac{C\sqrt{\log N}}{\sqrt{N|I(N)|}}=\frac{o(\log N)}{\sqrt{N}}.
\end{align*}

\item Sign changes: Let $Z^r_{N,\psi}=\{x\in\Z_N:\Re\psi_N(x),\Re\psi_N(x+1)\text{ have opposite signs}\}$. Thus $Z^r_{N,\psi}$ counts the sign changes\footnote{For convenience we will not count $x$ where $\Re\psi_N(x)=0$ or $\Re\psi_N(x+1)=0$; we only consider $x$ where $\Re\psi_N(x)$ and $\Re\psi_N(x+1)$ have definite signs. However this distinction will not matter, as a zero value can only happen with nonzero probability if $(P_{I(N)})_{xx}=0$, which is ruled out for almost all $x$ by \eqref{eqn:P-diag}.} of $\Re\psi_N$. Similarly, define $Z^i_{N,\psi}$ for sign changes of $\Im\psi_N$. Let $|Z^{r/i}_{N,\psi}|$ be the associated random measures $|Z^{r/i}_{N,\psi}|=\sum_{x\in Z^{r/i}_{N,\psi}}\delta_{x/N}$, which are scaled so that the support is in $[0,1]$. Then the expected limit distributions $\frac{2}{N}\E|Z_{N,\psi}^{r}|$ and $\frac{2}{N}\E|Z_{N,\psi}^{i}|$ converge weakly to the uniform distribution $\operatorname{Unif}([0,1])$ as $N\to\infty$. In particular, the expected value of the number of sign changes of the real or imaginary part of $\psi_N$ is $\frac{N}{2}(1+o(1))$.
\end{enumerate}
\end{thm}

Parts (i), (iii), and (iv) all describe the Gaussian-like behavior of the random quasimode.
Parts (i) and (ii) as stated only concern a single vector $\psi_N$ for each $N$ (rather than for a full basis of $\C^N$), though one can also construct full bases of random quasimodes satisfying the properties in (i) and (ii) with high probability, using similar methods as for the full basis results in Theorem~\ref{thm:walsh-main} or \cite[Theorem 2.5]{pw}.)

Comparing the coordinates in (iii) to those of a random vector chosen according to Haar measure on the unit sphere in $\C^N$ (also called circular random wave model in \cite[\S3]{phys}), we see the leading order Gaussian behavior for these good coordinates is what would be expected for a Haar random vector. (Note that, for large $N$, a Haar random state $Z/\|Z\|_2$ for $Z\sim N_\C(0,I_N)$, is well approximated by the Gaussian vector $Z/\sqrt{N}$ due to concentration of measure, e.g. \cite[\S3.1]{Vershynin}.) 
The excluded $x$ or $(x,y)$ in (iii) are those corresponding to short-time forward or backward orbits of the classical map, or to classical discontinuities. In terms of the sets to be defined in Section~\ref{subsec:sets}, the $N_\C(0,1)$ Gaussian results will hold for $x\not\in\da_{J,\delta,\gamma,N}^W$ and $(x,y)\not\in\tilde{A}_{J,\delta,\gamma,N}^W$. 

Despite the non-standard-Gaussian coordinates in (iii), they are few enough that one can still compute $\ell^p$ norms in (iv). The leading order values for $p<\infty$ agree with those for a random unit vector from the sphere in $\C^N$ (chosen according to Haar measure) as $N\to\infty$. In the $\ell^\infty$ norm, the bound is off by a $o(\sqrt{\log N})$ factor (with a more precise rate depending on $|I(N)|$)  from the expected value for a random unit vector in $\C^N$.

Studying ``zeros'' (or in this case, sign changes) in part (v) is motivated by studies of zeros of 2D (continuous) Riemannian waves \cite{Zelditch2}. The result (v) here is a much simpler version than the continuous case, giving the limiting mean distribution for the \emph{sign changes} (``nodal points'') of these 1D discrete eigenstates. 
For real Haar-random vectors, the distribution of sign changes was determined in \cite{KeatingMezzadriMonstra} in the context of studying sign changes for quantum maps associated with chaotic torus maps.
As in the real case, the mean number of sign changes of the real and imaginary parts of a complex Haar random vector is $\frac{N}{2}$, which can be seen since $Z/\|Z\|_2$, for $Z\sim N_\C(0,I_N)$, is Haar distributed.

\begin{rmk}
Numerical evidence suggests the actual eigenspaces of $\hat B_N$ are non-degenerate (at least for even $N$ between $50$ and $10\,000$ \cite{bakernumerics}), and that the actual eigenvectors can exhibit scarring along periodic orbits \cite{BV,ML2005}, even leading to fractal-like eigenstates \cite{ML2005}.
\end{rmk}

\subsection{Eigenstates for the Walsh quantized baker map}\label{sec:walshintro}
In this section we define the \emph{Walsh} quantization of the baker map, which has been studied in \cite{SchackCaves,TracyScott,NonnenmacherZworski,AN}.
We will provide further details in Section~\ref{sec:walsh}. 
Fix $D\ge2$ and consider the classical $D$-baker map on $\T^2$,
\begin{align*}
B(q,p)&=(Dq\;\;\mathrm{mod}\; 1,\frac{p+\floor{Dq}}{D}).
\end{align*}
When $D=2$, this is just the standard baker map \eqref{eqn:baker}. For the quantum Hilbert spaces $\mathcal{H}_N$, consider dimensions $N=D^k$ for $k\in\N$, so that states in $\mathcal{H}_N$ can be represented using tensor products, with (position) basis states $|\varepsilon_1\varepsilon_2\cdots\varepsilon_k\rangle = |{\varepsilon_1}\rangle\otimes |{\varepsilon_2}\rangle\otimes\cdots\otimes |{\varepsilon_k}\rangle$, where each $|{\varepsilon_i}\rangle$ is the standard basis element in $\C^D$ for the $\varepsilon_i$-th coordinate, $\varepsilon_i\in\intbrr{0:D-1}$. This corresponds to the position
$x=\sum_{m=1}^k\varepsilon_m D^{k-m}$.
Instead of using the Fourier transform $\hat{F}_{D^k}$ on $\mathcal{H}_N$ to construct a quantization, one replaces it with the \emph{Walsh transform} $W_{D^k}$, which is defined on tensor product states using very small DFT blocks as 
\begin{align*}
W_{D^k}(v^{(1)}\otimes\cdots\otimes v^{(k)}) &=\hat{F}_D v^{(k)}\otimes \hat{F}_D v^{(2)}\otimes\cdots\otimes \hat{F}_D v^{(1)}.
\end{align*}
Thus $W_{D^k}=(\hat{F}_D)^{\otimes k}\tilde{R}$,
where $(\hat{F}_D)^{\otimes k}$ is the $k$-fold tensor product of the standard $D\times D$ DFT matrix $\hat{F}_D$, and  $\tilde{R}:v^{(1)}\otimes v^{(2)}\otimes\cdots\otimes v^{(k)}\mapsto v^{(k)}\otimes v^{(k-1)}\otimes\cdots\otimes v^{(1)}$ is the dit reversal map. 
Then analogous to the Balazs--Voros construction \eqref{eqn:bv}, but using $W_{D^k}$ in place of $\hat{F}_{D^k}$, one defines the Walsh quantization of the $D$-baker map as
\begin{align}\label{eqn:wbaker}
B_k^\Wa = W_{D^k}^{-1}\begin{pmatrix}W_{D^{k-1}}&0&0\\0&\ddots&0\\0&0&W_{D^{k-1}}\end{pmatrix}.
\end{align}
Its action on tensor product states is
\begin{align}\label{eqn:walsh-action}
B_k^\Wa(v^{(1)}\otimes\cdots\otimes v^{(k)})=v^{(2)}\otimes v^{(3)}\otimes\cdots\otimes v^{(k)}\otimes \hat{F}_D^\dagger v^{(1)}.
\end{align}
To distinguish this matrix from the Balazs--Voros quantization $\hat{B}_N$ in \eqref{eqn:bv}, we use the slightly cumbersome notation $B_k^\Wa$. However, the Walsh quantization will only appear in this subsection and in Section~\ref{sec:walsh}.

The matrix $B_k^\Wa$ is called a {Walsh} quantization of the classical $D$-baker map, in the sense that it satisfies a classical-quantum correspondence \cite{AN}  involving observables $a\in C^\infty(\T^2)$ that are quantized according to a \emph{Walsh} {quantization}, rather than according to the more usual Weyl quantization. In fact, as shown in \cite{TracyScott}, $B_k^\Wa$ is \emph{not} a quantization for the baker's map according to the Weyl quantization. However, the Walsh quantized baker map allows for a more explicit understanding of the eigenvalues and explicit construction of some eigenstates \cite{AN}, and allows for the advantage that one can localize a state in both position and momentum under the Walsh harmonic analysis.

As noted in \cite{Gutkin}, these Walsh quantization matrices also essentially coincide with certain quantizations, in the sense of \cite{pzk,qgraphs}, of the angle-expanding map $x\mapsto Dx\;\mathrm{mod}\;1$ on $[0,1]$. The value statistics argument for the doubling map matrices in \cite[\S8]{pw} can be applied to the $D=2$ Walsh baker case here, after making adjustments for a difference of negative sign choice in the matrix. This would show Gaussian value statistics for randomly chosen eigenvectors of \eqref{eqn:wbaker} for $D=2$ in the position basis. 
However, QUE properties and the statements concerning a general coherent state basis must be handled differently; in particular,
the matrix powers for the specific $D=2$ case in position basis have a particularly simple pattern that can be analyzed just by matrix multiplication as in \cite[\S8]{pw}.
In the $D>2$ case, as well as in the coherent state bases for any $D\ge2$ (which are needed for the QUE property), the matrices have more complicated patterns, leading to Proposition~\ref{prop:walsh-powers}.

The matrices $B_k^\Wa$ have high eigenspace degeneracies, and it was shown in \cite{AN} by explicit construction of certain eigenstates that QUE does not hold for these models. However, we show here that sequences of randomly chosen eigenbases \emph{do} satisfy QUE with high probability as $k\to\infty$, so that the non-equidistributing states are rare. Additionally, we obtain the properties in Theorem~\ref{thm:rw} for (randomly chosen) actual eigenbases of $B_k^\Wa$. 
Due to the specifics of the Walsh quantization, we consider \emph{$(k,\ell)$-coherent state bases}, for $\ell=\ell(k)\in\intbrr{0:k}$, in addition to the position and momentum bases (which are included as the $\ell=0$ and $\ell=k$ cases). These coherent state bases
are defined at the start of Section~\ref{sec:walsh} and consist of vectors localized in small rectangles in phase space.
\begin{thm}\label{thm:walsh-main}
Take a random orthonormal eigenbasis $(\psi^{(k,m)})_{m=1}^{D^k}$ of $B_k^\Wa$ by choosing a random orthonormal basis (according to Haar measure) within each eigenspace. Let $\Omega_k$ be the probability space from which such bases are drawn, and let $\ell=\ell(k)\in\intbrr{0:k}$. Then the following hold as $k\to\infty$:
\begin{enumerate}[(i),leftmargin=*]
\item Gaussian value statistics w.h.p.: For a unit vector $\psi$, let 
\[
\mu^\psi_{k,\ell}:=\frac{1}{D^k}\sum_{[\varepsilon'\cdot\varepsilon]\in\mathcal{R}_{k,\ell}}\delta_{\sqrt{D^k}\langle\varepsilon'\cdot\varepsilon|\psi\rangle}
\]
be the empirical probability distribution of the scaled coordinates of $\psi$ in the $(k,\ell)$-coherent state basis.
Then there is a sequence of sets $\Pi_k\subseteq\Omega_k$ with $\P[\Pi_k]\xrightarrow{k\to\infty}1$ with the following property: For any orthonormal basis $(\tilde{\psi}^{(k,m)})_{m=1}^{D^k}$ in $\Pi_k$ and any sequence $(m_k)_k$, $m_k\in\intbrr{1:D^k}$, the sequence $(\mu^{\tilde{\psi}^{(k,m_k)}}_{k,\ell(k)})_k$ converges weakly to $N_\C(0,1)$ as $k\to\infty$.

\item QUE w.h.p.: There is a sequence of sets $\Gamma_k\subseteq\Omega_k$ with $\P[\Gamma_k]\xrightarrow{k\to\infty}1$ such that any sequence of orthonormal bases $(\tilde{\psi}^{(k,m)})_{m=1}^{D^k}$ from $\Gamma_k$ equidistributes in phase space as $k\to\infty$: For every sequence $(m_k)_k$, $m_k\in\intbrr{1:D^k}$, and any observable $a\in \operatorname{Lip}(\T^2)$,
\begin{align*}
\lim_{k\to\infty}\langle \tilde{\psi}^{(k,m_k)}|\operatorname{Op}_{k,\ell}^\Wa(a)|\tilde{\psi}^{(k,m_k)}\rangle &= \int_{\T^2}a(\x)\,d\x.
\end{align*}

\item 
\begin{sloppypar}
For nearly all $[\varepsilon'\cdot\varepsilon],[\delta'\cdot\delta]\in\mathcal{R}_{k,\ell}$, the scaled coordinates $\sqrt{D^k}\{\langle\varepsilon'\cdot\varepsilon|\psi^{(k,m)}\rangle\}_{[\varepsilon'\cdot\varepsilon]\in\mathcal{R}_{k,\ell}}$
have moments $\E|\sqrt{D^k}\langle\varepsilon'\cdot\varepsilon|\psi^{(k,m)}\rangle|^{n}$ and autocorrelation functions $\E[D^{2k}|\langle\varepsilon'\cdot\varepsilon|\psi^{(k,m)}\rangle|^2|\langle\delta'\cdot\delta|\psi^{(k,m)}\rangle|^2]$ that agree to leading order with those of the standard complex Gaussian.
\end{sloppypar}

\item $\ell^p$ norms: 
Let $\|\cdot\|_{p,(k,\ell)}$ and $\|\cdot\|_{\infty,(k,\ell)}$ denote the $\ell^p$ and $\ell^\infty$ norms taken in the $(k,\ell)$ coherent state basis. For $g\sim N_\C(0,1)$ and $p>0$,
\begin{align*}
\E\|\psi^{(k,m)}\|_{p,(k,\ell)}^p=\frac{\E|g|^p}{(D^k)^{\frac{p}{2}-1}}[1+o(1))],\quad\text{and}\quad\E\|\psi^{(k,m)}\|_{\infty,(k,\ell)}\le \frac{C\sqrt{k\log D^k}}{\sqrt{D^k}}(1+o(1)).
\end{align*}

\item 
Let $Z^r_{N,\psi^{(k,m)}}=\{x\in\Z_N:\Re\psi^{(k,m)}(x),\Re\psi^{(k,m)}(x+1)\text{ have opposite signs}\}$ be the number of sign changes\footnote{As in Theorem~\ref{thm:rw}, we do not count $x$ where $\Re\psi^{(k,m)}(x)=0$ or $\Re\psi^{(k,m)}(x+1)=0$. It appears numerically that it is possible to have $(P_\jj)_{xx}=0$, where $P_\jj$ is projection onto the $\jj$th eigenspace, for certain $D,k,\jj,x$, but by Theorem~\ref{thm:walsh-proj} this cannot happen for many $x$.} of the real part of $\psi^{(k,m)}$, and similarly define $Z^i_{N,\psi^{(k,m)}}$ for sign changes of $\Im\psi^{(k,m)}$. Let $|Z^{r/i}_{N,\psi^{(k,m)}}|$ be the associated random measures $|Z^{r/i}_{N,\psi^{(k,m)}}|=\sum_{x\in Z^{r/i}_{N,\psi^{(k,m)}}}\delta_{x/N}$, which is scaled so that the support is in $[0,1]$. Then the expected limit distributions $\frac{2}{N}\E|Z_{N,\psi^{(k,m)}}^{r}|$ and $\frac{2}{N}\E|Z_{N,\psi^{(k,m)}}^{i}|$ converge weakly to the uniform distribution $\operatorname{Unif}([0,1])$ as $N\to\infty$. In particular, the expected value of the number of sign changes of the real or imaginary part of any $\psi^{(k,m)}$ is $\frac{N}{2}(1+o(1))$.
\end{enumerate}
\end{thm}

The coherent states for which the properties in (iii) hold are determined by the ``good'' sets $G_{k,\ell}$ and $GP_{k,\ell}$ in Theorem~\ref{thm:walsh-proj}. In particular, the statement for moments holds for $|\varepsilon'\cdot\varepsilon\rangle\in G_{k,\ell}$, and the statement for autocorrelation functions holds for pairs 
with $|\varepsilon'\cdot\varepsilon\rangle,|\delta'\cdot\delta\rangle\in G_{k,\ell}$ and $(|\varepsilon'\cdot\varepsilon\rangle,|\delta'\cdot\delta\rangle),(|\delta'\cdot\delta\rangle,|\varepsilon'\cdot\varepsilon\rangle)\in GP_{k,\ell}$.

Properties (i), (iii), (iv), and (v) will follow, once we prove Proposition~\ref{prop:walsh-powers} on matrix powers of $B_k^\Wa$, from a similar kind of argument as for Theorem~\ref{thm:rw}, though with a different smooth approximation. For (ii), we will need to look at the specific Walsh quantization method to prove a local Weyl law that holds within each eigenspace.
Since we work with an entire eigenbasis rather than a single random state $\psi$, in (i) and (ii) the proof is more technical than in Theorem~\ref{thm:rw}, and will use the details used in \cite[\S6]{pw}. 

\begin{rmk}
The eigenvalues of $B_k^\Wa$ are $(4k)$-th roots of unity for $D\ge3$ and $(2k)$-th roots of unity for $D=2$. Thus the spectral arc length scale involved is at most $\frac{1}{4k}=\frac{1}{4\log_D N}$ (or $\frac{1}{2k}$ if $D=2$), which is smaller than the $|I(N)|$ allowed in Theorem~\ref{thm:rw}. 
However, due to the special construction of $B_k^\Wa$, its matrix properties are still computable even beyond the Ehrenfest time (Section~\ref{subsec:walsh-proof}), allowing for the more precise spectral estimates.
\end{rmk}

\begin{rmk}
We note that similar eigenvector and QUE results for the quantum cat map, which has high eigenspace degeneracies for certain dimensions $N$, were proved independently in the PhD thesis of N. Schwartz \cite{Schwartz-thesis}.
\end{rmk}

\section{Background and proof outline}\label{sec:setup}

Throughout this article, generic constants such as $C$, $c$, or $C_M$ may change value from line to line without further indication.

\subsection{Quantization on the torus}\label{subsec:quant}

In this section we review some standard properties of quantization on the two-torus phase space $\T^2=\R^2/\Z^2$. For a more detailed review, see for example \cite{bdb,qmaps}.

\subsubsection{States on the torus}
First we must identify the Hilbert space of states to associate with the torus.
One starts with the usual quantization of states with phase space $\R^2$. These are given by tempered distributions $\psi\in\mathscr{S}'(\R)$. 
On such states, the position and momentum operators are defined (in the position basis) via $(Q\psi)(q)=q\psi(q)$ and $(P\psi)(q)=-i\hbar\frac{\partial\psi}{\partial q}(q)$. 
For any $\hbar\in(0,1]$, there are quantum phase space translations $U(q,p)=e^{\frac{i}{\hbar}(pQ-qP)}$, which give a representation of the Weyl--Heisenberg group on $\mathscr{S}'(\R)$.  
Recall that $e^{\frac{i}{\hbar}(pQ-qP)}=e^{-\frac{i}{2\hbar}qp}e^{\frac{i}{\hbar}pQ}e^{-\frac{i}{\hbar}qP}$, and that $e^{-\frac{i}{\hbar}qP}$ is the translation operator $e^{-\frac{i}{\hbar}qP}\psi(q')=\psi(q'-q)$.

To go to the phase space $\T^2$, one now makes the restriction that $\psi$ be periodic in both position $q$ and momentum $p$, that is, $U(1,0)\psi=\psi$ and $U(0,1)\psi=\psi$. (More generally, one can allow  phase factors $e^{i\alpha}\psi$ and $e^{i\beta}\psi$, but for simplicity we choose $\alpha=\beta=0$.) With these periodic conditions, one can show we must have $(2\pi\hbar)^{-1}=N\in\N$ in order to have nonzero states. Choosing $N\in\N$ results in an $N$-dimensional space $\mathcal{H}_N$ of distributions, whose elements (with a choice of normalization) can be written as,
\begin{align}\label{eqn:psi}
\psi(q) &= \frac{1}{\sqrt{N}}\sum_{j\in\Z}c_j\delta\Big(q-\frac{j}{N}\Big),
\end{align}
with $c_{j+N}=c_j$.
This is written in the ``position representation'', for which we define the $N$ basis states,
\begin{align}
|x\rangle\equiv e_x(\cdot) &= \frac{1}{\sqrt{N}}\sum_{v\in\Z}\delta\Big(\cdot-\frac{x}{N}+v\Big),\quad x\in\Z_N.
\end{align}
In this basis, the state in \eqref{eqn:psi} is written $|\psi\rangle=\sum_{x=0}^{N-1}c_x|x\rangle$.
One can also switch to the ``momentum representation'' in terms of the momentum basis $\{|p\rangle\}$, via a discrete Fourier transform (DFT), $|p\rangle=\hat{F}_N^{-1}e_p(\cdot)$.

We now return to the phase translations $U(k_1,k_2)=e^{2\pi i N(k_2Q-k_1P)}$, where we have used $\hbar=(2\pi N)^{-1}$, and define the scaling
$T(k_1,k_2):=U(\frac{k_1}{N},\frac{k_2}{N})$.
For the torus position basis $\{|x\rangle\}_{x=0}^{N-1}$, we see $e^{2\pi i q_2Q}|x\rangle=e^{2\pi i q_2x/N}|x\rangle$, and $e^{-2\pi ik_1P}|x\rangle=|x+k_1\rangle$. 
Thus on the torus position basis states, the phase translation $T(k_1,k_2)$ acts as
\begin{align} 
\begin{aligned}
T(k_1,k_2)|x\rangle &= e^{-\pi i k_1k_2/N}e^{2\pi i k_2Q}e^{-2\pi i k_1P}|x\rangle \\
&= e^{-\pi i k_1k_2/N}e^{2\pi i k_2(x+k_1)/N}|x+k_1\rangle.
\end{aligned}\label{eqn:T(k)}
\end{align}

\subsubsection{Quantization of observables} 
Now that we have defined states as elements of the Hilbert space $\mathcal{H}_N$, we can define operators on these states corresponding to classical observables $f\in C^\infty(\T^2)$. Just as for the states, the quantization of observables on the torus starts with quantization of observables on $\R^2$, followed by a reduction to $\T^2=\R^2/\Z^2$. The quantization we work with is the Weyl quantization, which can be written on the torus as follows.
\begin{defn}\label{def:weyl}
The \emph{Weyl quantization} of an observable $f\in C^\infty(\T^2)$ is
\begin{align*}
\operatorname{Op}_N^\W(f):=\sum_{k\in\Z^2}\tilde{f}(k)T(k),
\end{align*}
where
$\tilde{f}(k) = \tilde{f}(k_1,k_2)= \int_{\T^2}f(q,p)e^{-2\pi i(qk_2-pk_1)}\,dq\,dp$ and $T(k) = T(k_1,k_2)=e^{2\pi i(k_2Q-k_1P)}$ is the phase space translation written in \eqref{eqn:T(k)}. 
\end{defn}
\begin{lem}\label{lem:useful}
Here we collect two useful properties concerning Weyl quantization.
\begin{enumerate}[(i)]
\item \cite[Prop. 3.10]{bdb}, \cite[Lemma 7]{DNW} For any integer $M\ge3$, $N\ge1$, and $f\in C^\infty(\T^2)$,
\begin{align}
\frac{1}{N}\operatorname{Tr}\operatorname{Op}_N^\W(f) &= \int_{\T^2}f(\x)\,d\x+\mathcal{O}_M\Big(\frac{\|f\|_{C^M}}{N^M}\Big),
\end{align}
where $\|f\|_{C^M}:=\sum_{|\gamma|\le M}\|\partial^\gamma f\|_\infty$, with multi-index notation $\gamma=(\gamma_1,\gamma_2)\in\N_0^2$ and $|\gamma|=|\gamma_1|+|\gamma_2|$.

\item Calder\'on--Vaillancourt, e.g. see \cite[Lemma 9]{DNW}: There is a constant $C$ so that for any $a\in C^\infty(\T^2)$ and $N\ge1$,
\begin{align*}
\|\operatorname{Op}_N^\W(a)\|_{\mathcal{B}(\mathcal{H}_N)} \le C\|a\|_{C^2}.
\end{align*}

\end{enumerate}
\end{lem}
We will also make use of the following standard Fourier decay estimate in the proofs:
For $M\in\N$, there is a constant $C_M$ so that for $f\in C^M(\T^2)$ and $k\in \Z^2\setminus\{0\}$, 
\begin{align}\label{eqn:fourier-bound}
	|\tilde{f}(k)| &\le \frac{C_M\max_{|\alpha|=M}\|\partial^\alpha f\|_\infty}{\|k\|_2^M}. 
\end{align}

Besides the Weyl quantization, another useful quantization is the \emph{anti-Wick quantization}, which has a nice comparison with the Weyl quantization in the semiclassical limit $\hbar\to0$.
The anti-Wick quantization is defined in terms of coherent states, which will be defined in Section~\ref{subsec:coh}.

\subsubsection{Quantization of area-preserving maps on the torus}

Finally, now that we have quantum states and observables, we can quantize the classical baker's map $B:\T^2\to\T^2$. 
The baker's map was first quantized in \cite{BV}, and it is this quantization we use here. Other quantizations of the baker's map have also been used, for example another quantization which preserves a classical symmetry was introduced in \cite{Saraceno}, and a large class of baker quantizations based on qubits that includes the former was studied in \cite{SchackCaves,TracyScott}.
As defined in the introduction, the quantization we primarily study has the explicit formula,
\begin{align}
\hat{B}_N &= \hat{F}_N^{-1}\begin{pmatrix}\hat{F}_{N/2}&0\\0&\hat{F}_{N/2}\end{pmatrix},
\end{align} 
for $\hat{F}_N$ the $N\times N$ DFT matrix.
It satisfies an Egorov theorem Theorem~\ref{thm:egorov}, which is a rigorous classical-quantum correspondence, proved in \cite{DNW}. The Egorov theorem states that for appropriate observables $a$, unitary conjugation by $\hat{B}_N$, i.e. $\hat{B}_N\operatorname{Op}_N^\W(a)\hat{B}_N^{-1}$, looks like Weyl quantization of composition with the classical baker's map, $\operatorname{Op}_N^\W(a\circ B^{-1})$, as $N\to\infty$.
Using the Egorov theorem, quantum ergodicity \eqref{eqn:qe} was then proved for this model in \cite{DNW}.

\subsection{Special sets and regions in the torus}\label{subsec:sets}
To discuss various regions of coordinates of the $N\times N$ matrices $\hat{B}_N^k$, 
in this section we define several subsets of $\intbrr{0,N-1}^2$, which also correspond to regions in $\T^2$.
Because we work with matrix notation like $\langle x|\hat{B}_N|y\rangle$, we note that the $x$ coordinate corresponds to the $x$th \emph{row} and the $y$ coordinate cooresponds to the $y$th \emph{column}, and additionally the coordinate pair $(x,y)=(0,0)$ corresponds to the top left corner. This is drawn for example in Figure~\ref{fig:avoid}.

\subsubsection{Discontinuity set} The classical baker's map $B:[0,1]^2\to[0,1]^2$ is discontinuous at the edges of the square $[0,1]^2$ as well as along the vertical line $q=1/2$. The $j$th iterates of the baker's map are also discontinuous along the vertical lines $q=\frac{k}{2^j}$, $k=0,\ldots,2^j-1$. Let the set of coordinates to avoid due to discontinuities be (with parameters $J\in\N,\delta\in(0,1/2),\gamma\in(0,1/2)$ to be chosen later)
\begin{align}
B_{J,\delta,\gamma,N}&=\bigg\{(x,y)\in\intbrr{0:N-1}^2:\left|\frac{y}{N}-\frac{k}{2^J}\right|\le \delta\text{ any }k\in\Z,\; \text{ or }\frac{x}{N}\le \gamma\text{ or }\frac{x}{N}\ge1-\gamma\bigg\}.
\end{align}
This set has size $\#B_{J,\delta,\gamma,N}\le C(2^J\delta N^2+\gamma N^2)$. 
Note that we just have $\delta\ll 2^{-J}$ and $\gamma\ll1$ in order for $\#B_{J,\delta,\gamma,N}$ to have a ``small'' size $o(N^2)$.
An example set is shown in Figure~\ref{fig:discont}. 

\subsubsection{Classical sets} Let the pairs of points $(x,y)\in\intbrr{0:N-1}^2$ with $(x/N,y/N)$ close to the graph of the classical map $q\mapsto 2^kq\mod 1$ on $\R/\Z$ be
\begin{align}
C^W_{k,N}&=\{(x,y)\in\intbrr{0:N-1}^2:d_{\Z/N\Z}(x,2^ky)\le W\},
\end{align}
where $d_{\Z/N\Z}(x,2^ky)$ denotes the distance in $\Z/N\Z$. An example set $C_{k,N}^W$ is drawn in Figure~\ref{fig:slopes}. The letter $C$ is for ``classical set'', and the set $C_{k,N}^W$ can be seen as a thickened version of the graph of $q\mapsto 2^kq$.
This set has size $\#C_{j,N}^W\le (2W+1)N$, since for each column $y$ there are $2\lfloor W\rfloor+1$ coordinates $x$ within $W$ of $y$  (Figure~\ref{fig:slopes}). We will show that away from these classical sets, the entries $\hat{B}_N^k$ are generally small (Theorem~\ref{prop:mpowers}).

\subsubsection{Union to avoid}
The total set of points to avoid is the union of the ``classical'' sets and ``discontinuity'' sets defined above up to time $J$,
\begin{align}
A^W_{J,\delta,\gamma,N}&:= B_{J,\delta,\gamma,N}\cup \bigcup_{k=1}^JC_{k,N}^W.
\end{align}
It has size $\#A_{J,\delta,\gamma,N}^W\le C(2^J\delta N^2+\gamma N^2+JWN)$.

We will need a symmetrized version of this set in order to control entries of $B_N^{-k}$, and will also include additional points relating to diagonal entries (which will be relevant for Section~\ref{subsec:off-diag}),
\begin{align}\label{eqn:tilde-A}
\tilde{A}_{J,\delta,\gamma,N}^W:=\{(x,y):(x,y)\in A_{J,\delta,\gamma,N}^W\text{ or }(y,x)\in A_{J,\delta,\gamma,N}^W,\text{ or }(x,x)\in{A}_{J,\delta,\gamma,N}^W\text{ or }(y,y)\in{A}_{J,\delta,\gamma,N}^W\}. 
\end{align}
The size of this set restricted to a region around the diagonal will be estimated in Eq.~\eqref{eqn:est-diag-bad} in a similar way as Eq.~\eqref{eqn:da} below.

\begin{figure}[!ht]
\makebox[\textwidth][c]{
\begin{subfigure}[t]{.5\textwidth}
\captionsetup{width=.9\linewidth}
\includegraphics{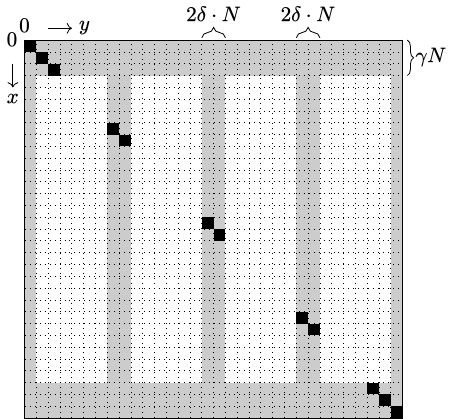}
\caption{$N\times N$ grid ($N=32$) with an example set $B_{J,\delta,\gamma,N}$ for $J=2$, $\delta=1/N$, and $\gamma=3/N$ in gray. The black squares are those points $(x,x)$ on the diagonal.}\label{fig:discont}
\end{subfigure}
\hfill
\begin{subfigure}[t]{.5\textwidth}
\captionsetup{width=.9\linewidth}
\includegraphics{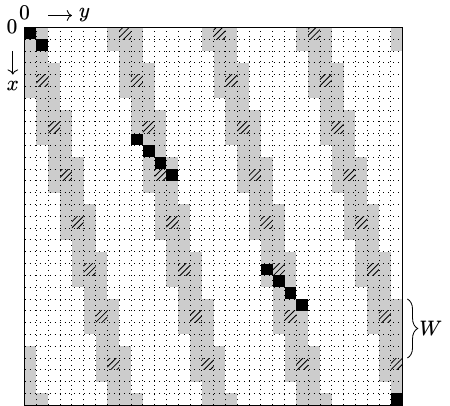}
\caption{An example set $C_{J,N}^W$ for $N=32$, $J=2$, and $W=5$ in gray. The black squares are those on the diagonal, and the hatched squares are the points where $x=2^Jy\;\mathrm{mod}\, N$.}\label{fig:slopes}
\end{subfigure}
}
\caption{The ``bad set'' $B_{J,\delta,\gamma,N}$ is to be excluded due to discontinuities and diffraction effects. The ``classical set'' $C_{k,N}^W$ is where we expect $\hat{B}_N^k$ to be (relatively) large. 
Away from $B_{J,\delta,\gamma,N}$ and $\bigcup_{k=1}^J C_{k,N}^W$, the matrices $\hat{B}_N^k$, $k=1,\ldots,J$, will have small entries (Theorem~\ref{prop:mpowers}).}\label{fig:avoid}
\end{figure}

\subsubsection{Diagonal points to avoid}
It will be useful to define the set of diagonal points $(x,x)$ that are in $A_{J,\delta,\gamma,N}^W$  so we can avoid them. Define the set of these points $x$ as
\begin{align}
\da_{J,\delta,\gamma,N}^W&=\{x\in\intbrr{0:N-1}:(x,x)\in A_{J,\delta,\gamma,N}^W\}.
\end{align}
An example is shown along the diagonals in Figure~\ref{fig:avoid}. 

We can give an estimate of the size of the diagonal in $C_{k,N}^W$ by computing the area of the bounding region shown in Figure~\ref{fig:area}. The bounding region depicted in blue northwest hatching has area $\sqrt{2}\operatorname{length}(ab)$.
The slope of the segment $av_2$ is $-2^k$, and the segment $v_1v_2$ has length $(2W+1)+2^k$, so the length of segment $ac$ is $\frac{(2W+1)+2^k}{2^k}$. Using the law of sines in the triangle $acb$, since we know $\angle a=\frac{\pi}{4}$ and $\angle c=\frac{\pi}{2}+\tan^{-1}(2^{-k})$, gives the length of $ab$ which we see is $\le CW$ for a constant $C$. Thus the contribution to $\#\da_{J,\delta,\gamma,N}$ from each $C_{k,N}^W$ is $\le C2^kW$.
Taking all $C_{k,N}^W$ for $k=1,\ldots,J$, and adding in the coordinates in $B_{J,\delta,\gamma,N}$, results in the estimate
\begin{align}\label{eqn:da}
\begin{aligned}
\#\da_{J,\delta,\gamma,N}^W &\le C_1\left(\gamma N+2^J\delta N+ \sum_{k=1}^J 2^kW\right)\\
&\le C_2(\gamma N+2^J\delta N+2^JW).
\end{aligned}
\end{align}

\begin{figure}[!ht]
\includegraphics{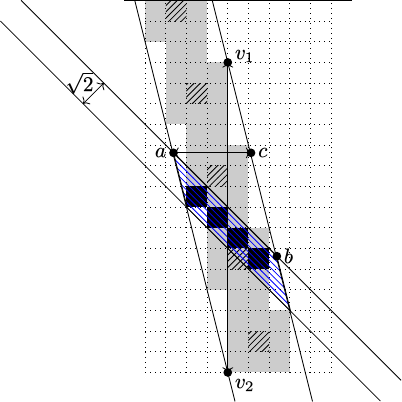}
\caption{Example region (shown in blue northwest hatching) used to bound the size of the diagonal set $\da_{J,\delta,\gamma,N}^W$.}\label{fig:area}
\end{figure}

\subsubsection{Good set in $\T^2$}
Finally, because we will discuss integration over different regions of the torus, define the subset of $\T^2=\R^2/\Z^2$ that is effectively the complement of the discontinuity set $B_{J,\delta,\gamma,N}$ rescaled to fit in $\T^2$,
\begin{align}
\mathcal{G}_{J,\delta,\gamma,N}:=\left\{(q,p)\in\T^2:\forall \ell\in\Z,\;\left|q-\frac{\ell}{2^J}\right|>\delta,\;p\in(\gamma,1-\gamma)\right\}.
\end{align}
This is taken from the definition of the set ``$\mathcal{D}_{n,\delta,\gamma}$'' in \cite{DNW}.
For all $k\in \intbrr{0:J-1}$, we have the inclusion $B^k\mathcal{G}_{J,\delta,\gamma,N}\subset\mathcal{G}_{J-k,2^k\delta,\gamma/2^k}$. There is also the immediate inclusion $\mathcal{G}_{J-k,2^k\delta,\gamma/2^k}\subseteq\mathcal{G}_{1,2^k\delta,\gamma/2^k}$.

\subsection{Outline for Theorems~\ref{thm:Pmat} and \ref{thm:lweyl}}\label{subsec:proof-outline}

In this section we outline the proof steps for Theorems~\ref{thm:Pmat} and \ref{thm:lweyl}, which are as follows.
\begin{align*}
\parbox{3.5cm}{Matrix elements of $\hat{B}_N^k$\\(Thm.~\ref{prop:mpowers})} \;\xrightarrow{\text{Prop.~\ref{prop:mp-sf}}}\; \parbox{2.7cm}{Matrix elements\\ of $P_{I(N)}$, $Q_{N,I(N)}$} \xrightarrow{\text{Prop.~\ref{thm:window}}}\;\parbox{2.7cm}{windowed local\\Weyl law}
\end{align*}
The proof of Theorem~\ref{thm:Pmat} is completed through Theorem~\ref{prop:mpowers} and Proposition~\ref{prop:mp-sf}, and the proof of Theorem~\ref{thm:lweyl} is complete after Proposition~\ref{thm:window}.

The first step is the following choices of parameters and theorem on the matrix elements of $\hat{B}_N^k$, which explains Figure~\ref{fig:Bpowers}. The proof is done using that $\hat{B}_N$ evolves coherent states (to be defined in Section~\ref{subsec:coh}) nicely according to the classical baker's map.
\begin{defn}\label{def:param}
Suppose $|I(N)|\log N\to\infty$ as $N\to\infty$,
and define parameters
\begin{align}\label{eqn:param2}
J&=(\log_2N)\varepsilon(N),\qquad \delta=10\sqrt{\frac{\log_2 N}{N}},\qquad {\gamma=N^{-1/3}},\qquad W=N^{\frac{1}{2}+2\varepsilon(N)},
\end{align}
where $\varepsilon(N)$ is any choice of function such that $\varepsilon(N)\to0$ at a slow enough rate that $|I(N)|J(N)=\varepsilon(N)|I(N)|\log_2 N\to\infty$. 
\end{defn}
Note that $2^J=N^{\varepsilon(N)}=\exp\left(\varepsilon(N)\log N\right)\to\infty$ as $N\to\infty$. We may think of $J$ and $W$ as integer-valued by taking floors or ceilings. What matters is just the growth rate as $N\to\infty$. We will implicitly assume that $N$ is large enough so e.g. $J\ge1$. 
Some examples of allowable $\varepsilon(N)$ are $\frac{1}{(|I(N)|\log N)^\alpha}$ for any $0<\alpha<1$, and $\frac{1}{\log(|I(N)|\log N)}$.
\begin{thm}[time evolution/matrix powers]\label{prop:mpowers}
For the parameters defined in \eqref{eqn:param2}, we know the matrix element sizes of $\hat{B}_N^k$ away from the discontinuity set $B_{J,\delta,\gamma,N}$. Roughly, they are small away from the classical graph, and may be up to size $C2^{-k/2}$ near the classical graph. More precisely, there is $r(N)\to0$ and a numerical constant $C$ so that
\begin{align}
|(\hat{B}_N^k)_{xy}|&\le2^{-k/2}r(N),\;\forall k\in\intbrr{1:J},(x,y)\not\in A_{J,\delta,\gamma,N}^W,\label{eqn:sbound}\\
|(\hat{B}_N^k)_{xy}|&\le C2^{-k/2},\;\forall k\in\intbrr{1:J},(x,y)\not\in B_{J,\delta,\gamma,N}.\label{eqn:cbound}
\end{align}
The decay rate $r(N)$ depends only on $N$ and the parameters in \eqref{eqn:param2}, and in particular it is uniform in $x,y,k$. 
With the choices \eqref{eqn:param2}, we have the explicit estimate $r(N)=\mathcal{O}\left(N^{-1/12}+\exp\left(-\frac{\pi}{2}N^{2\varepsilon(N)}\right)\right)$.
\end{thm}
The proof will be given in Section~\ref{sec:Bpowers-proof}.

Next, we use the control on the matrix elements of $\hat{B}_N^k$ up to time $J$ to determine values of the spectral projection matrix $P_I=\sum_{j:\theta^{(j,N)}\in I}|\varphi^{(j,N)}\rangle\langle\varphi^{(j,N)}|$ on intervals $I$. We state it first for general $I_N$ and $J_N$ (not assuming the parameter choices in \eqref{eqn:param2}).
\begin{prop}[matrix powers to spectral function]\label{prop:mp-sf}
Let $N\in2\N$. Suppose for some $J_N\ge1$, $r_N>0$, and $S\subset\intbrr{0:N-1}^2$, that
\begin{align}\label{eqn:Bentry-bound}
|(\hat{B}_N^k)_{xy}|&\le 2^{-k/2}r_N,\;\forall k\in\intbrr{1:J_N},(x,y)\in S.
\end{align}
Then for any interval $I_N\subseteq\R/(2\pi\Z)$, 
\begin{align}\label{eqn:diag-stronger}
\left|(P_{I_N})_{xx}-\frac{|I_N|}{2\pi}\right|
&\le\frac{|I_N|}{2\pi}\left[\frac{2\pi}{|I_N|J_N}+\Big(1+\frac{2\pi}{|I_N|J_N}\Big)2r_N\right],\quad\text{ for }(x,x)\in S,
\intertext{and}
\label{eqn:offdiag-stronger}
\left|(P_{I_N})_{xy}\right| &\le \frac{|I_N|}{2\pi}\left[ \frac{4\pi}{|I_N|J_N}+\Big(1+\frac{2\pi}{|I_N|J_N}\Big)6r_N\right],
\end{align}
for $x\ne y$ such that $(x,y),(y,x),(x,x),(y,y)\in S$.

Under the setting of Theorem~\ref{prop:mpowers} and Definition~\ref{def:param}, 
the above estimates imply,
\begin{align}
(P_{I(N)})_{xx}&=\frac{|I(N)|}{2\pi}(1+\mathcal{O}(\mathcal{R}_\mathrm{d}(N))),\quad\text{ for }x\not\in\da_{J,\delta,\gamma,N}^W,\\
(P_{I(N)})_{xy}&=\mathcal{O}(\mathcal R_{\mathrm{od}}(N))|I(N)|,\quad\text{ for } x\ne y,\;(x,y)\not\in\tilde{A}_{J,\delta,\gamma,N}^W,
\end{align}
for decay rates ($\mathcal{R}_\mathrm{d}(N),\mathcal{R}_\mathrm{od}(N)\to0$),
\begin{align}\label{eqn:decayrates}
\mathcal{R}_\mathrm{d}(N)=\mathcal{R}_\mathrm{od}(N)&=\frac{1}{|I(N)|J}+r_N,
\end{align}
and where
\begin{align*}
\tilde{A}_{J,\delta,\gamma,N}^W:=\{(x,y):(x,y)\in A_{J,\delta,\gamma,N}^W\text{ or }(y,x)\in A_{J,\delta,\gamma,N}^W,\text{ or }x\in\da_{J,\delta,\gamma,N}^W\text{ or }y\in\da_{J,\delta,\gamma,N}^W\}
\end{align*}
as defined in \eqref{eqn:tilde-A}.
\end{prop}

The proof is given in Section~\ref{sec:p-sf-proof}. For an idea of the size of remainders $\mathcal{R}_\mathrm{d}(N)$ and $\mathcal{R}_{\mathrm{od}}(N)$, note that if $r_N=\mathcal{O}(N^{-1/12}+\exp(-\frac{\pi}{2}N^{2\varepsilon(N)})$ as in the Theorem~\ref{prop:mpowers} with the choices \eqref{eqn:param2}, then $r_N\ll \frac{1}{|I(N)|J}$, 
and the decay rates in \eqref{eqn:decayrates} are of order $\frac{1}{|I(N)|J}$, which decays but at a rate slower than $1/\log N$.

At this point, once the above matrix entry bounds for $\hat{B}_N^k$ and $P_{I(N)}$ are proved (which is done in Sections~\ref{sec:Bpowers-proof} and \ref{sec:p-sf-proof} respectively), then Theorem~\ref{thm:Pmat} is proved. The extension to $q_N$ in Theorem~\ref{thm:Qmat} is given in Appendix~\ref{sec:qN}.

The above proposition implies we know the matrix entries of $P_{I(N)}$ fairly well; on the diagonal, they are generally around $\frac{|I(N)|}{2\pi}$, while on the off-diagonal they are $o(|I(N)|)$ except possibly in the set $\tilde{A}_{J,\delta,\gamma,N}^W$. This is illustrated in Figure~\ref{fig:P}, where the projection matrix entries for $N=1000$ are plotted (in absolute value) for the interval $I=[2.1,3]$. The large off-diagonal entries of $P_{[2.1,3]}$ lie near the lines described by $x\approx 2^ky\mod N$ or $y\approx 2^kx\mod N$ for small $k$, which are the forward and backward iterates of the classical baker map in position space. Such points must be contained in the excluded region $\tilde{A}_{J,\delta,\gamma,N}^W$.

\begin{figure}[!ht]
\includegraphics[height=2.3in]{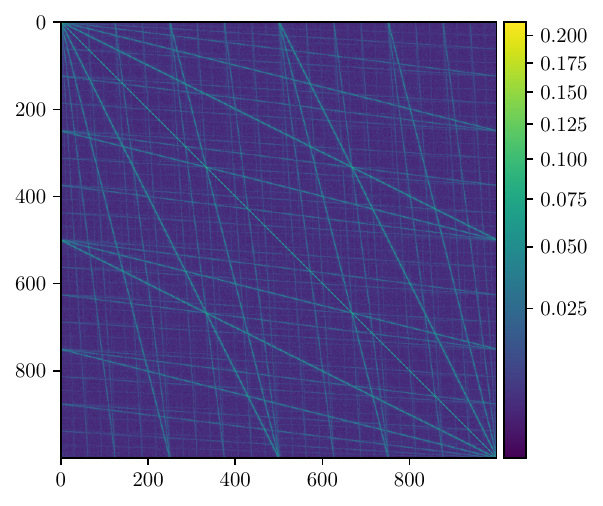}\;
\includegraphics[height=2.3in]{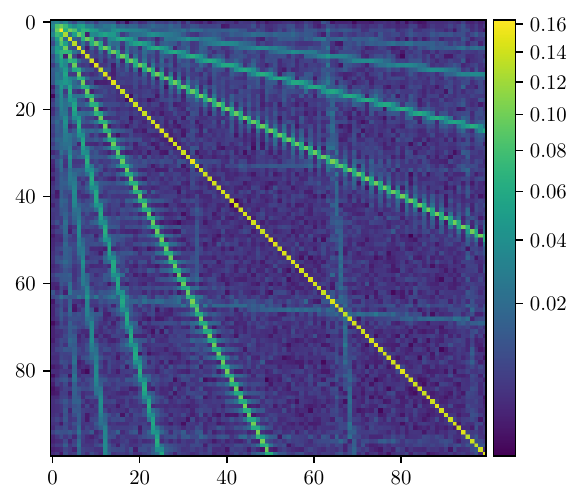}
\caption{The absolute value of the matrix elements of the projection matrix $P_{[2.1,3]}$ for $N=1000$, plotted on a (nonlinear) power scale. The left image shows the entire matrix $P_{[2.1,3]}$, while the right image is zoomed in to show the top left corner containing matrix entries $(x,y)$ with $x,y<100$. Most of the diagonal entries are generally close to $\frac{|I(N)|}{2\pi}=0.143\ldots$,
and the large off-diagonal entries visually appear to follow the shape of the set $\tilde{A}_{J,\delta,\gamma,N}^W$; outside this set the entries appear to be typically small, reflecting Proposition~\ref{prop:mp-sf}.}\label{fig:P}
\end{figure}

Finally, knowing the entries of the projection matrix $P_{I(N)}$ as in Proposition~\ref{prop:mp-sf} allows us to prove a windowed local Weyl law, which will give Theorem~\ref{thm:lweyl}.

\begin{prop}[spectral projection to windowed local Weyl law]\label{thm:window}
\begin{sloppypar}
Define parameters as in Definition~\ref{def:param}, and let $(e^{i\theta^{(j,N)}},\varphi^{(j,N)})_j$ be the eigenvalue-eigenvector pairs for $\hat{B}_N$, and $P_{I(N)}=\sum_{\theta^{(j,N)}\in I(N)}|\varphi^{(j,N)}\rangle\langle\varphi^{(j,N)}|$.
Suppose that with the decay rates defined in \eqref{eqn:decayrates} and $r(N)$ as taken in the end of Theorem~\ref{prop:mpowers}, 
\begin{align}
(P_{I(N)})_{xx}&=\frac{|I(N)|}{2\pi}(1+\mathcal{O}(\mathcal{R}_\mathrm{d}(N))),\quad\text{ for }x\not\in\da_{J,\delta,\gamma,N}^W,\\
(P_{I(N)})_{xy}&=|I(N)|\mathcal{O}(\mathcal{R}_\mathrm{od}(N)),\quad\text{ for } x\ne y,\;(x,y)\not\in\tilde{A}_{J,\delta,\gamma,N}^W.
\end{align}
Then  for any classical observable $f\in C^\infty(\T^2)$,
\begin{align}\label{eqn:lweyl}
\frac{2\pi}{N|I(N)|}\sum_{\theta^{(j,N)}\in I(N)}\langle\varphi^{(j,N)}|\operatorname{Op}_N^\W(f)|\varphi^{(j,N)}\rangle &= \int_{\T^2}f(q,p)\,dq\,dp +\mathcal{O}\left((1+\|f\|_{C^3})\frac{1}{|I(N)|J}\right).
\end{align}
\end{sloppypar}
\end{prop}
The proof is given in Section~\ref{sec:lweyl-proof}, and the extension \eqref{eqn:lweyl-q} to $q_N$ is given in Appendix~\ref{sec:qN}.
\begin{rmk}
The proposition does not have anything to do with $\varphi^{(j,N)}$ being eigenvectors; it just needs that they are orthonormal vectors and that the projection matrix $P_{I(N)}$ is the orthogonal projection onto their span and satisfies the stated matrix entry assumptions.
\end{rmk}

\section{Time evolution of \texorpdfstring{$\hat{B}_N$}{B_N}, Proof of Theorem~\ref{prop:mpowers}}
\label{sec:Bpowers-proof}

In this section we use coherent states evolution and integral estimates to estimate the matrix entries $|\langle x|\hat{B}_N^k|y\rangle|$. Since we assume that $(x,y)\not\in B_{J,\delta,\gamma,N}$, we will be able to use the action of $\hat{B}_N^k$ on coherent states for regions away from the discontinuity regions, and use the location of $(x,y)$ to bound the terms coming from the discontinuity regions. In what follows we will eventually take $\sigma=1$, but may leave it in place for generality and agreement with notation in \cite{DNW}.

We note that coherent state evolution is more precise than what we need; for the position basis matrix elements, we do not need the momentum information in the end, and the quantum graph models from \cite{pzk} also satisfy this matrix power property (even in a stronger sense \cite{pw}) despite not having the correct coherent state evolution \cite{TracyScott}. However, because the Fourier transform acts conveniently on coherent states, and the matrices $\hat{B}_N$ are built from DFT matrices, using coherent state evolution is particularly convenient.

\subsection{Coherent states}\label{subsec:coh}
The $\R^2$ coherent state centered at the origin with squeezing parameter $\sigma>0$ is the Gaussian wavepacket
$
\Psi_{0,\sigma}(q):=\left(\frac{\sigma}{\pi\hbar}\right)^{1/4}e^{-\frac{\sigma q^2}{2\hbar}}.$
The coherent state $\Psi_{\x,\sigma}$ at a point $\x=(q_0,p_0)\in\R^2$ is obtained by phase space translation, 
\begin{align*}
\Psi_{\x,\sigma}(q):=(U(q_0,p_0)\Psi_{0,\sigma})(q) &= \left(\frac{\sigma}{\pi\hbar}\right)^{1/4}e^{-i\frac{p_0q_0}{2\hbar}}e^{i\frac{p_0q}{\hbar}}e^{-\frac{\sigma(q-q_0)^2}{2\hbar}}\\
&=(2N\sigma)^{1/4}e^{-i\pi Nq_0p_0}e^{2\pi i Np_0q}e^{-\sigma N\pi(q-q_0)^2}.\numberthis\label{eqn:coh-state}
\end{align*}
To project to the torus $\T^2$, one can first periodicize $\Psi_{\x,\sigma}$ to make the ``cylinder'' coherent state, $\Psi_{\x,\sigma,\mathcal{C}}(q):=\sum_{z\in\Z}\Psi_{\x,\sigma}(q+z)$, and then construct the torus coherent state as
\begin{align}
\Psi_{\x,\sigma,\T^2}(j):=\frac{1}{\sqrt{N}}\Psi_{\x,\sigma,\mathcal{C}}\Big(\frac{j}{N}\Big),\quad j\in\intbrr{0:N-1}.
\end{align}
The torus coherent states form an overdetermined system; there is the resolution of the identity (\cite[Lemma 3.8]{bdb}),
\begin{align}
\operatorname{Id}_{\mathcal{H}_N} &= N\int_{\T^2}|\Psi_{\x,\sigma,\T^2}\rangle\langle\Psi_{\x,\sigma,\T^2}|\,dq\,dp.
\end{align} 
The evolution of coherent states under $\hat{B}_N$ was proved in \cite{DNW}. The main two results we will need from there are as follows.
This first lemma allows us to approximate coherent states on the torus by coherent states on $\R^2$, which have the simple Gaussian formula \eqref{eqn:coh-state}. 
\begin{lem}[\cite{DNW}, Lemma 3]\label{lem:3}
Let $\x=(q_0,p_0)$ with $q_0\in(\delta,1-\delta)$ for some $0<\delta<1/2$. Then as $N\to\infty$,
\begin{equation}
\forall q\in[0,1),\quad \Psi_{\x,\sigma,\mathcal{C}}(q)=\Psi_{\x,\sigma}(q)+\mathcal{O}((\sigma N)^{1/4}e^{-\pi N\sigma\delta^2}),
\end{equation}
with the error estimate uniform for $\sigma N\ge1$.
\end{lem}
By writing $\Psi_{\x,\sigma,\mathcal{C}}(q)=\Psi_{x,\sigma}(q)+\Psi_{x,\sigma}(q+1)+\Psi_{x,\sigma}(q-1)+\sum_{|z|=2}^\infty\Psi_{\x,\sigma}(q+z)$, we can also check that there is a numerical constant $C$ so that for $y\in\intbrr{0:N-1}$,
\begin{align}\label{eqn:lem3bound}
|\langle y|\Psi_{\x,\sigma,\T^2}\rangle|\le C\left(\frac{\sigma}{N}\right)^{1/4}e^{-\pi N\sigma d_{\R/\Z}(\frac{y}{N},q_0)^2},
\end{align}
where $\x=(q_0,p_0)$ and $d_{\R/\Z}(a,b)=\min_{\nu\in\Z}|a-b+\nu|$.

By iterating the coherent state evolution \cite[Prop. 5]{DNW}, we obtain the behavior of powers $\hat{B}_N^k$ on coherent states.
\begin{prop}\label{prop:iter5}
Let $\delta,\gamma\in(0,1)$, $\sigma\in[\frac{1}{N},N]$, and $2^j\le\sqrt{N\sigma}$.
For $\x\in\mathcal{G}_{j,\delta,\gamma,N}$, any $k\in\intbrr{1:j}$, and any $m\in\intbrr{0:N-1}$, we have the norm bound
\begin{align}\label{eqn:tel}
\big\|\hat{B}_N^k\Psi_{\x,\sigma,\T^2}-e^{iN\pi\Theta_k(\x)} \Psi_{B^k\x,\sigma/4^k,\T^2}\big\| &= \mathcal{O}(N^{3/4}\sigma^{1/4}\exp(-\pi N\theta)),
\end{align}
where $\Theta_k(\x)=\sum_{\ell=0}^{k-1}\Theta(B^\ell\x)$ with $\Theta(q_0,p_0)=\begin{cases}0,&q_0\in(\delta,\frac{1}{2}-\delta)\\q_0+\frac{p_0+1}{2},&q_0\in(\frac{1}{2}+\delta,1-\delta)\end{cases}$, and  $\theta=\min(\sigma\delta^2,\gamma^2/\sigma)$. The implied constant is uniform in the allowed $\delta,\gamma,\sigma$, in $k$, and in $\x\in\mathcal{G}_{j,\delta,\gamma,N}$.
\end{prop}
\begin{proof}
\cite[Proposition 5]{DNW} gives for $\x\in\mathcal{G}_{1,\delta,\gamma,N}$ and $\sigma\in[1/N,N]$, that
\begin{align}\label{eqn:prop5}
\big\|\hat{B}_N\Psi_{\x,\sigma,\T^2} -e^{iN\pi\Theta(\x)}\Psi_{B\x,\sigma/4,\T^2}\big\|=\mathcal{O}(N^{3/4}\sigma^{1/4}\exp(-\pi N\theta)),
\end{align}
with implied constant uniform in $\x\in\mathcal{G}_{1,\delta,\gamma,N}$ and in the allowed $\delta,\gamma$ and $\sigma\in[1/N,N]$. 
Writing the left side of \eqref{eqn:tel} as a telescoping sum, we obtain the bound
\begin{multline*}
\bigg\|\sum_{\ell=0}^{k-1}e^{i\sum_{i=0}^{\ell-1}\Theta(B^i\x)}\hat{B}_N^{k-\ell-1}\left(\hat{B}_N\Psi_{B^\ell\x,\frac{\sigma}{4^\ell},\T^2}-e^{i\pi N\Theta(B^\ell\x)}\Psi_{B^{\ell+1}\x,\frac{\sigma}{4^{\ell+1}},\T^2}\right)\bigg\| \\
\le \sum_{\ell=0}^{k-1}\frac{1}{2^{\ell/2}}\mathcal{O}(N^{3/4}\sigma^{1/4}\exp(-\pi N\theta))
=\mathcal{O}(N^{3/4}\sigma^{1/4}\exp(-\pi N\theta)).
\end{multline*}
In the above we used $\mathbf{y}=B^\ell\x\in\mathcal{G}_{j-\ell,2^\ell\delta,\gamma/2^\ell,N}\subseteq\mathcal{G}_{1,2^\ell\delta,\gamma/2^\ell,N}$ for $\x\in\mathcal{G}_{j,\delta,\gamma,N}$ and $\ell\in\intbrr{0:k-1}$, and also that the parameter $\theta$ ends up being independent of the time evolution.
\end{proof}

\subsection{Proof of Theorem~\ref{prop:mpowers}}

Using the resolution of the identity
$
\operatorname{Id}_{\mathcal{H}_N}=N\int_{\T^2}|\Psi_{\z,\sigma,\T^2}\rangle\langle\Psi_{\z,\sigma,\T^2}|\,d\z$,
we apply Proposition~\ref{prop:iter5} to $\z\in\mathcal{G}_{J,\frac{\delta}{2},\gamma,N}$, where the parameters are those defined in \eqref{eqn:param2}. Note that for sufficiently large $N$, we have $2^J\le\sqrt{N\sigma}$ by definition of $J$ and $\sigma=1$. We leave the parameter $\sigma$ in the notation below for generality, agreement with notation in \cite{DNW}, and to keep track of coherent state scaling.  For $j\le J$,
\begin{align}
\nonumber\big|\langle x|\hat{B}_N^j|y\rangle\big| &= \left|N\int_{\T^2}\langle x|\hat{B}_N^j|\Psi_{\z,\sigma,\T^2}\rangle\langle\Psi_{\z,\sigma,\T^2}|y\rangle\,d\z\right| \\
\label{eqn:jaction}&\le \left|N\int_{\mathcal{G}_{J,\frac{\delta}{2},\gamma,N}}e^{iN\pi \Theta_j(\z)}\langle x|\Psi_{B^j\z,\sigma/4^j,\T^2}\rangle\langle\Psi_{\z,\sigma,\T^2}|y\rangle\,d\z\right|+\\
\nonumber&\hspace{2cm}+ \left|N\int_{\T^2\setminus\mathcal{G}_{J,\frac{\delta}{2},\gamma,N}}\langle x|\hat{B}_N^j|\Psi_{\z,\sigma,\T^2}\rangle\langle\Psi_{\z,\sigma,\T^2}|y\rangle\,d\z \right|+
\mathcal{O}({N^{3/2}\sigma^{1/2}}e^{-2\pi N\omega}),
\end{align}
where $\omega:=\min(\sigma\delta^2/4,\gamma^2/\sigma)$.
The discontinuity region $\T^2\setminus\mathcal{G}_{J,\frac{\delta}{2},\gamma,N}$ is small, and the integral involving this region is bounded using  \eqref{eqn:lem3bound} as,
\begin{align}\label{eqn:discontsetbound}
\left|N\int_{\T^2\setminus\mathcal{G}_{J,\frac{\delta}{2},\gamma,N}}\!\!\!\langle x|\hat{B}_N^j|\Psi_{\z,\sigma,\T^2}\rangle\langle\Psi_{\z,\sigma,\T^2}|y\rangle\,d\z \right|&\le N\int_{\T^2\setminus\mathcal{G}_{J,\frac{\delta}{2},\gamma,N}}\!\!\!\|\Psi_{\z,\sigma,\T^2}\|_2C\left(\frac{\sigma}{N}\right)^{1/4}e^{-\sigma N\pi d_{\R/\Z}(\frac{y}{N},z_0)^2}\,d\z,
\end{align}
where $z_0$ is the position coordinate of $\z=(z_0,z_1)$. The set $\T^2\setminus\mathcal{G}_{J,\frac{\delta}{2},\gamma,N}$ consists of two parts, the two horizontal strips of height $\gamma=N^{-1/3}$ at the top and bottom of $\T^2$, and the $2^{J}+1$ vertical strips of width at most $\delta$. Over the two horizontal strips, integrating shows the contribution to \eqref{eqn:discontsetbound} is $\mathcal{O}(N^{-1/12}\sigma^{-1/4})$.
For the vertical strips, since we assume the matrix entry coordinates $(x,y)\not\in B_{J,\delta,\gamma,N}$, we have $\left|\frac{y}{N}-\frac{k}{2^J}\right|>\delta$. For the vertical strips of  $\T^2\setminus\mathcal{G}_{J,\frac{\delta}{2},\gamma,N}$, we also have that $z_0$ is within $\delta/2$ of some $k/2^J$. So $\left|z_0-\frac{y}{N}\right|\ge \left|\frac{y}{N}-\frac{k}{2^J}\right|-\left|z_0-\frac{k}{2^J}\right|>\frac{\delta}{2}$.
Thus the contribution from the $2^J+1$ vertical strips to \eqref{eqn:discontsetbound} is
\begin{align*}
\le (2^J+1)\delta N^{3/4}\sigma^{1/4}\|\Psi_{\z,\sigma,\T^2}\|_2Ce^{-\sigma N\pi \delta^2/4}
\le C\sigma^{1/4}(\log N)^{1/4}N^{1/4+\varepsilon(N)}N^{-25\pi}.
\end{align*}
So in total the integral \eqref{eqn:discontsetbound} over the discontinuity region $\T^2\setminus\mathcal{G}_{J,\frac{\delta}{2},\gamma,N}$ is of order $\mathcal{O}(N^{-1/12}\sigma^{-1/4})$.

Continuing with \eqref{eqn:jaction}, by Lemma~\ref{lem:3}, for $\z\in\mathcal{G}_{J,\frac{\delta}{2},\gamma,N}$ and $j\le J$,
\begin{align*}
\sqrt{N}\left|\langle x|\Psi_{B^j\z,\sigma/4^j,\T^2}\rangle\right|&=\Big|\Psi_{B^j\z,\sigma/4^j}\Big(\frac{x}{N}\Big)\Big|+\mathcal{O}\big((\sigma N)^{1/4}2^{-j/2}e^{-\pi N\sigma\delta^2/4}\big),\\
\sqrt{N}\left|\langle\Psi_{\z,\sigma,\T^2}|y\rangle\right|&=
\Big|\Psi_{\z,\sigma}\Big(\frac{y}{N}\Big)\Big|+\mathcal{O}\big((\sigma N)^{1/4}e^{-\pi N\sigma\delta^2/4}\big).
\end{align*}
Plugging this back into \eqref{eqn:jaction},  we obtain the bound
\begin{multline}
\bigg|N\int_{\mathcal{G}_{J,\frac{\delta}{2},\gamma,N}}e^{iN\pi \Theta_j(\z)}\langle x|\Psi_{B^j\z,\frac{\sigma}{4^j},\T^2}\rangle\langle\Psi_{\z,\sigma,\T^2}|y\rangle\,d\z\bigg|\le \\
\label{eqn:ja2}
\int_{\T^2}\Big|\Psi_{B^j\z,\sigma/4^j}\Big(\frac{x}{N}\Big)\Big|\Big|\Psi_{\z,\sigma}\Big(\frac{y}{N})\Big|\,d\z + \mathcal{O}(2^{-j/2}(\sigma N)^{1/2}e^{-\pi N\sigma\delta^2/4}).
\end{multline}
Collecting the error terms from \eqref{eqn:jaction} so far, we have
$\mathcal{O}({N^{3/2}\sigma^{1/2}}e^{-2\pi N\omega})$, $\mathcal{O}(N^{-1/12}\sigma^{-1/4})$, and 
$\mathcal{O}((\sigma N)^{1/2}2^{-j/2}e^{-\pi N\sigma\delta^2/4})$,
for which the total sum is (for $\sigma=1$), 
\begin{align}\label{eqn:error1}
\mathcal{O}(N^{3/2}e^{-\pi N\min(\delta^2/4,\gamma^2)})+\mathcal{O}(\gamma N^{1/4}).
\end{align}
The choices \eqref{eqn:param2} of $\delta,\gamma$ ensure this is $\mathcal{O}(N^{-1/12})$, which is $o(2^{-J/2})$.

The leading order term in \eqref{eqn:ja2} is
\begin{align*}
\int_{\T^2}\Big|\Psi_{B^j\z,\sigma/4^j}\Big(\frac{x}{N}\Big)\Big|\Big|\Psi_{\z,\sigma}\Big(\frac{y}{N})\Big|\,d\z &=(2N\sigma)^{1/2}2^{-j/2}\int_{0}^1dq\, e^{-\sigma N\pi 4^{-j}(\frac{x}{N}-(2^jq-k(q)))^2}e^{-\sigma N\pi(\frac{y}{N}-q)^2},
\end{align*}
where $k(q)=\floor{2^jq}$.
This $q$-integral with a Gaussian-like and Gaussian term will be small unless the centers of the two Gaussian(-like) terms are very close, which will be quantified using the sets $C_{j,N}^W$. 
There is the term $k(q)=\floor{2^jq}$ in the first term, which effectively puts a new Gaussian peak (width of order the standard deviation $\sim N^{-1/2}$) at each $\frac{x}{2^jN}+\frac{k}{2^j}$, $k=0,\ldots,2^j$. However, only the 3 $k$-values such that $\frac{x}{2^jN}+\frac{k}{2^j}$ is close to $\frac{y}{N}$ will have any chance of contributing.

First split up the integral into the region within $r$ of the center $\frac{y}{N}$, and the region further than $r$. We will take $r=\frac{1}{N^{1/4}}$. Since $r$ is much larger than  the Gaussian standard deviation $N^{-1/2}$, the integral over the region $|q-\frac{y}{N}|>r$ is small.
Write
\begin{align*}
(2N\sigma)^{1/2}&2^{-j/2}\int_{0}^1dq\, e^{-\sigma N\pi 4^{-j}(\frac{x}{N}-(2^jq-k(q)))^2}e^{-\sigma N\pi(\frac{y}{N}-q)^2} \\
&\le (2N\sigma)^{1/2}2^{-j/2}\left(\int_{B_{r}(\frac{y}{N})}dq\, e^{-\sigma N\pi \big(q-\frac{x}{2^jN}-\frac{k(q)}{2^j}\big)^2}e^{-\sigma N\pi(q-\frac{y}{N})^2} +\int_{|q-\frac{y}{N}|\ge r}dq\,e^{-\sigma N\pi(q-\frac{y}{N})^2}\right).\numberthis\label{eqn:rsplit}
\end{align*}
Recalling the complementary error function $\operatorname{erfc}(z)=1-\frac{2}{\sqrt{\pi}}\int_0^ze^{-t^2}\,dt$, the right integral over the region far from the center is,
\begin{align*}
(2N\sigma)^{1/2}2^{-j/2}\int_{|q-\frac{y}{N}|\ge r}e^{-\sigma N\pi(q-\frac{y}{N})^2} \,dq
&=2^{-j/2}\sqrt{2}\operatorname{erfc}(r\sqrt{\sigma N\pi})=\mathcal{O}\Big(\frac{e^{-\pi N^{1/2}}}{N^{1/4}}\Big),\numberthis\label{eqn:error2}
\end{align*}
since $r\sqrt{N}= N^{1/4}\to\infty$ as $N\to\infty$ and there is the asymptotic expansion
$
\operatorname{erfc}(z)=\frac{e^{-z^2}}{z\sqrt{\pi}}(1+\mathcal{O}(z^{-2}))$ as $z\to\infty$.

In the remaining integral over $B_r(\frac{y}{N})$ in \eqref{eqn:rsplit}, since $q$ is restricted to $B_{r}(\frac{y}{N})$, we can know $k(q)$ in the exponent. Note that for sufficiently large $N$, that $r\le 2^{-J}$, since then $2^{-J}=\frac{1}{N^{\varepsilon(N)}}\ge\frac{1}{N^{1/4}}$. Then $\big|q-\frac{y}{N}\big|\le r\le 2^{-J}$, so that $|2^jq-\frac{y}{N}2^j|\le1$. Letting $a=\floor{\frac{y}{N}2^j}$, then we must have $k(q)=\floor{2^jq}\in\{a-1,a,a+1\}$.

Thus for any $q\in B_{r}(\frac{y}{N})$, the quantity $k(q)$ takes one of only three values  $a$, $a-1$, or $a+1$. 
As we do not care about factors of 3 (we just want to avoid any growing factors like $2^J$), we put all three in to obtain,
\begin{multline*}
(2N\sigma)^{1/2}2^{-j/2}\int_{B_{r}(\frac{y}{N})} e^{-\sigma N\pi\left(q-\frac{x}{2^jN}-\frac{k(q)}{2^j}\right)^2}e^{-\sigma N\pi(q-\frac{y}{N})^2}\,dq\\
\le (2N\sigma)^{1/2}2^{-j/2}\sum_{\ell=-1,0,1}\int_{-\infty}^\infty e^{-\sigma N\pi(q-s_\ell)^2}e^{-\sigma N\pi q^2}\,dq,
\end{multline*}
where $s_\ell=\frac{x}{2^jN}+\frac{a+\ell}{2^j}-\frac{y}{N}$.
The Gaussian integral is
$
\int_{-\infty}^\infty e^{-\sigma N\pi(q-s)^2}e^{-\sigma N\pi q^2}\,dq =\frac{1}{(2\sigma N)^{1/2}}e^{-\sigma N\pi s^2/2}$,
so we obtain
\begin{multline}
(2N\sigma)^{1/2}2^{-j/2}\sum_{\ell=-1,0,1}\int_{-\infty}^\infty e^{-\sigma N\pi(q-s_\ell)^2}e^{-\sigma N\pi q^2}\,dq\\
\le2^{-j/2}\left(e^{-\frac{\sigma \pi}{2N 4^j}(x-2^jy+Na)^2}+e^{-\frac{\sigma \pi}{2N4^j}(x-2^jy+N(a+1))^2}+e^{-\frac{\sigma \pi}{2N4^j}(x-2^jy+N(a-1))^2}\right).\label{eqn:g-error}
\end{multline}
The right hand side is always upper bounded by $3\cdot 2^{-j/2}$, yielding \eqref{eqn:cbound} with error terms \eqref{eqn:error1} and \eqref{eqn:error2}, which are both $o(2^{-J/2})$.

If $(x,y)\not\in C_{j,N}^W$, then $W<d_{\Z/N\Z}(x,2^jy)= \min_{\alpha\in\Z}|x-2^jy-N\alpha|$, and so we obtain the bound for the right hand side of \eqref{eqn:g-error},
\begin{align}\label{eqn:rbound}
\le 3\cdot 2^{-j/2}e^{-\frac{\sigma \pi}{2N4^j}W^2}.
\end{align}
In total, then from \eqref{eqn:jaction}, for $(x,y)\not\in C_{j,N}^W$, we have
$
(\hat{B}_N^j)_{xy}\le 2^{-j/2}r(N),
$
where collecting \eqref{eqn:error1}, \eqref{eqn:error2}, and \eqref{eqn:rbound} and applying \eqref{eqn:param2} show,
\begin{align}\label{eqn:rN}
r(N)=\mathcal{O}\left(N^{-1/12}+\exp\left(-\frac{\pi}{2}N^{2\varepsilon(N)}\right)\right).
\end{align}
\qed

\section{Spectral projection elements, Proof of Proposition~\ref{prop:mp-sf}}\label{sec:p-sf-proof}

In this section, we prove Proposition~\ref{prop:mp-sf}, to go from matrix elements of $\hat{B}_N^k$ to those of the spectral projection $P_{I(N)}$. 
This will complete the proof of Theorem~\ref{thm:Pmat}. 

We start with the method used in \cite{pw} to approximate $P_{I(N)}$ using trigonometric polynomials involving powers of the unitary matrix, here $\hat{B}_N$. However we will also need estimates on the off-diagonal elements of $P_{I(N)}$,
as well as knowledge on the positions of $(x,y)$ where we have ``good'' estimates on $(P_{I(N)})_{xy}$.

The particular trigonometric polynomials we use to estimate $\Chi_{I(N)}$ are the Beurling--Selberg (or Selberg) polynomials \cite[\S45.20]{Selberg}, 
 \cite{Vaaler}, \cite{Montgomery2000}. 
First, the Beurling function is for $z\in\C$,
\begin{align}
B(z)&=\bigg(\frac{\sin \pi z}{\pi}\bigg)^2\left(\sum_{n=0}^\infty\frac{1}{(z-n)^2}-\sum_{n=-\infty}^{-1}\frac{1}{(z-n)^2}+\frac{2}{z}\right).
\end{align}
It is entire of exponential type\footnote{that is, it satisfies the growth condition that for every $\varepsilon>0$, there is $A_\varepsilon$ so that $|B(z)|\le A_\varepsilon e^{(2\pi+\varepsilon)|z|}$ for all $z\in\C$.} $2\pi$ and so has Fourier transform (with $e^{-ikx}$ normalization) supported in $[-2\pi,2\pi]$. It also satisfies
\begin{align}
\operatorname{sgn}(x)\le B(x)\text{ for } x\in\R,\quad\text{and }\quad \int_\R \big(B(x)-\operatorname{sgn}(x)\big)\,dx=1.
\end{align}
A plot of $B(x)$ is shown in \cite[Fig.1]{Montgomery2000}.
If $F(z)$ is entire of exponential type $2\pi$ with $\operatorname{sgn}(x)\le F(x)$ for $x\in\R$, then $\int_\R (F(x)-\operatorname{sgn}(x))\,dx\ge1$, so the Beurling function is an extremizer of this approximation problem.

To obtain approximants for a finite interval $I=[a,b]\subset\R$, one constructs Selberg's functions, for $\frac{D}{2\pi}\ge1$,
\begin{align}
g_{I,D}^{(+)}(z)&=\frac{1}{2}\left(B\left(\frac{D}{2\pi}(b-z)\right)+B\left(\frac{D}{2\pi}(z-a)\right)\right),
\end{align}
which approximate $\Chi_I(x)$ from above;
$g_{I,D}^{(+)}(x)\ge \Chi_I(x)$ for all $x\in\R$, and $\int_{-\infty}^\infty (g_{I,D}^{(+)}(x)-\Chi_I(x))\,dx =\frac{2\pi}{D}$. Additionally, the Fourier transform $\hat{g_{I,D}^{(+)}}(k)=\frac{1}{{2\pi}}\int_\R g_{I,D}^{(+)}(x)e^{-ikx}\,dx$ is supported in $[- D,D]$. 
One can also construct minorants $g_{I,D}^{(-)}(z)$ satisfying $g_{I,D}^{(-)}(x)\le \Chi_I(x)$ for all $x\in\R$ and $\int_{-\infty}^\infty (\Chi_I(x)-g_{I,D}^{(-)}(x))\,dx=\frac{2\pi}{D}$, with Fourier transform supported also in $[-D,D]$.

Since we consider intervals in $\R/(2\pi\Z)$ corresponding to arcs on the unit circle, we take the periodic versions,
\begin{align*}
G^{(\pm)}_{I,D}(x) &= \sum_{j\in\Z}g_{I,D}^{(\pm)}(x-2\pi j),
\end{align*}
which are trigonometric polynomials of degree $\le D$, and satisfy $\hat{G_{I,D}^{(\pm)}}(k)=\hat{g_{I,D}^{(\pm)}}(k)$ for $k\in\Z$.
They can be written as the Fourier series
\begin{align}\label{eqn:G}
G_{I,D}^{(\pm)}(x) &=\frac{|I|\pm 2\pi D^{-1}}{2\pi}+\sum_{\ell=1}^{\floor{D}}\left(\hat{g_{I,D}^{(\pm)}}(\ell)e^{i\ell x}+\hat{g_{I,D}^{(\pm)}}(-\ell)e^{-i\ell x}\right).
\end{align}

\subsection{Diagonal estimates}\label{subsec:diag-power}
We keep the general case of any interval $I_N\subseteq\R/(2\pi\Z)$ and $J_N\ge1$.
The diagonal estimates are very similar to \cite{pw}. Here we present a simplification for this case. Starting from the Selberg polynomials $G^{(\pm)}_{I_N,J_N}:\R/(2\pi\Z)\to\R$, we define their analogs on the unit circle by setting 
\[
F^{(\pm)}_{I_N,J_N}(e^{it}):=G^{(\pm)}_{I_N,J_N}(t).
\]
Then by the spectral theorem, the projection matrix $P_{I_N}=\sum_{j:\theta^{(j,N)}\in I_N}|\varphi^{(j,N)}\rangle\langle\varphi^{(j,N)}|$ satisfies,
\begin{align}\label{eqn:sand2}
F^{(-)}_{I_N,J_N}(\hat{B}_N)_{xx}&\le (P_{I_N})_{xx}\le F^{(+)}_{I_N,J_N}(\hat{B}_N)_{xx},\quad \text{ any }x\in\intbrr{0:N-1}.
\end{align}
Since $G_{I_N,J_N}^{(\pm)}$ are trigonometric polynomials of degree $\le J_N$, from \eqref{eqn:G} we have
\begin{align}\label{eqn:fourier}
F^{(\pm)}_{I_N,J_N}(\hat{B}_N)&= \frac{|I_N|}{2\pi}\Big(1\pm \frac{2\pi}{|I_N|J_N}\Big)\operatorname{Id}+\sum_{\ell=1}^{\floor{J_N}}\Big(\hat{g_{I_N,J_N}^{(\pm)}}(\ell)\hat{B}_N^\ell +\hat{g_{I_N,J_N}^{(\pm)}}(-\ell)\hat{B}_N^{-\ell}\Big).
\end{align}
The identity term is already the value we want.
Using the Fourier coefficient bound,
\begin{align}
|\hat{g_{I_N,J_N}^{(\pm)}}(\ell)| &\le \frac{1}{2\pi}\int_\R |g_{I_N,J_N}^{(\pm)}|\,dx \le \frac{1}{2\pi}(|I_N|+2\pi J_N^{-1}),
\end{align}
we obtain that the non-identity terms in \eqref{eqn:fourier} are bounded as follows (allowing for $x\ne y$ to keep generality),
\begin{align}\label{eqn:Jerror}
\left|\sum_{\ell=1}^{\floor{J_N}}\Big(\hat{g_{I_N,J_N}^{(\pm)}}(\ell)(\hat{B}_N^\ell)_{xy} +\hat{g_{I_N,J_N}^{(\pm)}}(-\ell)(\hat{B}_N^{-\ell})_{xy}\Big)\right|&\le \frac{|I_N|}{2\pi}\Big(1+\frac{2\pi}{|I_N|J_N}\Big)\sum_{\ell=1}^{\floor{J_N}}(|(\hat{B}_N^\ell)_{xy}|+|(\hat{B}_N^\ell)_{yx}|).
\end{align}
Now recall the set $S\subset\intbrr{0:N-1}^2$ from  the hypotheses of Proposition~\ref{prop:mp-sf}. For $(x,y)\in S$, we are given that $|(\hat{B}_N^\ell)_{xy}|\le 2^{-\ell}r_N$ for all $\ell\in\intbrr{1:J_N}$,  and so \eqref{eqn:Jerror} becomes, for $(x,y)$ such that $(x,y),(y,x)\in S$,
\begin{align*}
\left|\sum_{\ell=1}^{\floor{J_N}}\Big(\hat{g_{I_N,J_N}^{(\pm)}}(\ell)(\hat{B}_N^\ell)_{xy} +\hat{g_{I_N,J_N}^{(\pm)}}(-\ell)(\hat{B}_N^{-\ell})_{xy}\Big)\right|&
\le \frac{|I_N|}{2\pi}\Big(1+\frac{2\pi}{|I_N|J_N}\Big)2r_N\sum_{\ell=1}^{\floor{J_N}}2^{-\ell}\\
\\
&\le\frac{|I_N|}{2\pi}\Big(1+\frac{2\pi}{|I_N|J_N}\Big)2r_N,\numberthis\label{eqn:Jerror2}
\end{align*}
for $(x,y),(y,x)\in S$.
Returning to the diagonal entries $(x,x)$, we obtain the bound for any $x$ with $(x,x)\in S$,
\begin{align}\label{eqn:p-final}
\begin{aligned}
\left|(P_{I_N})_{xx}-\frac{|I_N|}{2\pi}\right|
&\le\frac{|I_N|}{2\pi}\left[\frac{2\pi}{|I_N|J_N}+\Big(1+\frac{2\pi}{|I_N|J_N}\Big)2r_N\right],
\end{aligned}
\end{align}
which proves \eqref{eqn:diag-stronger}.

Now we specialize to $I(N)$, $J$, and $r(N)$ as in \eqref{eqn:param2} in Theorem~\ref{prop:mpowers}. The set of $x$ with $(x,x)\in S$ is taken to be those $x\not\in\da_{J,\delta,\gamma,N}^W$. (The set $S$ is $({A}_{J,\delta,\gamma,N}^W)^c$.)
By the definition of $\varepsilon(N)$, we have $|I(N)|J\to\infty$, so the bound in \eqref{eqn:p-final} is $o(|I(N)|)$, more specifically, $|I(N)|\mathcal{O}(\mathcal{R}_\mathrm{d}(N))$ with $\mathcal{R}_\mathrm{d}(N)=\frac{1}{|I(N)|J}+r(N)$ as defined in \eqref{eqn:decayrates}.

\subsection{Off-diagonal estimates}\label{subsec:off-diag}
For $(x,y)$ such that $(x,x),(y,y)\in S$, then using the spectral theorem, Cauchy--Schwarz, and that $G_{I_N,J_N}^{(+)}- \Chi_{I_N}\ge0$,
\begin{align*}
\left|F_{I_N,J_N}^{(+)}(\hat{B}_N)_{xy}-(P_{I_N})_{xy}\right| &=\left|\sum_{j=1}^{N}\left(G_{I_N,J_N}^{(+)}(\theta^{(j)})-\Chi_{I_N}(\theta^{(j)})\right)\langle x|\varphi^{(j)}\rangle\langle\varphi^{(j)}|y\rangle\right| \\
&\le\left(F_{I_N,J_N}^{(+)}(\hat{B}_N)_{xx}-(P_{I_N})_{xx}\right)^{1/2}\left(F_{I_N,J_N}^{(+)}(\hat{B}_N)_{yy}-(P_{I_N})_{yy}\right)^{1/2}\\
&\le \frac{2}{J_N}+\frac{|I_N|}{2\pi}\Big(1+\frac{2\pi}{|I_N|J_N}\Big)4r_N.
\end{align*}
For $x\ne y$ with $(x,y),(y,x)\in S$, by \eqref{eqn:Jerror2} we obtain,
\begin{align*}
\left|F_{I_N,J_N}^{(+)}(\hat{B}_N)_{xy}\right|\le \frac{|I_N|}{2\pi}\Big(1+\frac{2\pi}{|I_N|J_N}\Big)2r_N.
\end{align*}
Thus for $x\ne y$ with $(x,x),(y,y),(x,y),(y,x)\in S$,
\begin{align}
\left|(P_{I_N})_{xy}\right| &\le \frac{|I_N|}{2\pi}\left[ \frac{4\pi}{|I_N|J_N}+\Big(1+\frac{2\pi}{|I_N|J_N}\Big)6r_N\right].
\end{align}

Specializing to $I(N),J$, and $r(N)$ as in \eqref{eqn:param2} in Theorem~\ref{prop:mpowers}, by equation \eqref{eqn:sbound}, the set $S$ is $(A_{J,\delta,\gamma,N}^W)^c$, and the bound \eqref{eqn:offdiag-stronger} is $\mathcal{O}\left(|I(N)|\left[r(N)+\frac{1}{|I(N)|J}\right]\right)$. \qed

\subsection{Proof of eigenvalue Weyl law, Corollary~\ref{cor:weyl}}\label{subsec:evcounting}

With Theorem~\ref{prop:mpowers} and Proposition~\ref{prop:mp-sf} proved, this establishes the projection matrix estimates \eqref{eqn:P-diag} and \eqref{eqn:P-offdiag}.

Note that by the choices in \eqref{eqn:param2}, 
\begin{align*}
\#\da_{J,\delta,\gamma,N}^W &\le C(\gamma N+2^J\delta N+2^JW)=\mathcal{O}(N^{2/3})=o(N|I(N)|).
\end{align*}
Then using the diagonal entry estimate \eqref{eqn:P-diag} and that $|(P_{N,I(N)})_{xx}|\le1$,
\begin{align*}
\#\{j:\theta^{(j,N)}\in I(N)\}=\operatorname{tr}P_{N,I(N)} &= \sum_{x\not\in\da_{J,\delta,\gamma,N}^W}\frac{|I(N)|}{2\pi}\left(1+o(1)\right)+\sum_{x\in\da_{J,\delta,\gamma,N}^W}(P_{N,I(N)})_{xx}\\
&=\frac{N|I(N)|}{2\pi}\left(1+o(1)\right).
\end{align*}

After proving Theorem~\ref{thm:Qmat} in Appendix~\ref{subsec:qN}, the remaining part of Corollary~\ref{cor:weyl}, the extension \eqref{eqn:weyl-qn} to $q_N$, will follow similarly, as described in
Appendix~\ref{appsubsec:ev}.

\section{Windowed local Weyl law and quantum ergodicity}
\label{sec:lweyl-proof}

In this section we prove Proposition~\ref{thm:window} to go from matrix entries of $P_{I(N)}$ to the windowed local Weyl law. This will prove Theorem~\ref{thm:lweyl}. In Section~\ref{subsec:qe}, we will then discuss the application to prove windowed quantum ergodicity (Theorem~\ref{thm:qv}), for which the proof will be given in Appendix~\ref{sec:wqe}.

In what follows we drop the $N$ superscript on the eigenvalues $\theta^{(j,N)}$ and (orthonormal) eigenvectors $\varphi^{(j,N)}$.

\subsection{Extracting diagonal terms}
Let $\{|x\rangle\}_{x=0}^{N-1}$ be the position basis. We know that (Lemma~\ref{lem:useful}) $\sum_{x=0}^{N-1}\langle x|\operatorname{Op}_N^\W(f)|x\rangle =\operatorname{Tr}(\operatorname{Op}_N^\W(f))= N\int_{\T^2}f(q,p)\,dq\,dp+\mathcal{O}_M(\frac{\|f\|_{C^M}}{N^{M-1}})$, so we will extract this term (times $\frac{|I(N)|}{2\pi}$) from the following expansion and show the remaining terms are small, i.e. $o(N|I(N)|)$. Write,
\begin{align*}
\sum_{j:\theta^{(j)}\in I(N)}\langle\varphi^{(j)}|\operatorname{Op}_N^\W(f)|\varphi^{(j)}\rangle
&=\sum_{x,y=0}^{N-1}\sum_{\theta^{(j)}\in I(N)}\langle \varphi^{(j)}|y\rangle\langle y|\operatorname{Op}_N^\W(f)|x\rangle\langle x|\varphi^{(j)}\rangle\\
&=\sum_{x=0}^{N-1}\langle x|\operatorname{Op}_N^\W(f)|x\rangle(P_{I(N)})_{xx}+\sum_{\substack{x,y=0\\x\ne y}}^{N-1}\langle y|\operatorname{Op}_N^\W(f)|x\rangle (P_{I(N)})_{xy}.\numberthis\label{eqn:op-exp}
\end{align*}
By assumption, we know $(P_{I(N)})_{xx}=\frac{|I(N)|}{2\pi}(1+\mathcal{O}(\mathcal{R}_\mathrm{d}(N)))$ for $x\not\in\da_{J,\delta,\gamma,N}^W$. 
For $x\in\da_{J,\delta,\gamma,N}^W$, we will just use the inequality $(P_{I(N)})_{xx}\le1$.
Then considering just the first sum in \eqref{eqn:op-exp}, which consists only of diagonal terms, we write
\begin{align*}
\sum_{x=0}^{N-1}\langle x&|\operatorname{Op}_N^\W(f)|x\rangle(P_{I(N)})_{xx}\\
&=\sum_{x\not\in\da_{J,\delta,\gamma,N}^W}\langle x|\operatorname{Op}_N^\W(f)|x\rangle\frac{|I(N)|}{2\pi}(1+\mathcal{O}(\mathcal{R}_\mathrm{d}(N)))+\sum_{x\in \da_{J,\delta,\gamma,N}^W}\langle x|\operatorname{Op}_N^\W(f)|x\rangle(P_{I(N)})_{xx}\\
&=\sum_{x=0}^{N-1}\langle x|\operatorname{Op}_N^\W(f)|x\rangle\frac{|I(N)|}{2\pi}(1+\mathcal{O}(\mathcal{R}_\mathrm{d}(N)))+\mathcal{O}\left(\|\operatorname{Op}_N^\W(f)\|\right)\#\da_{J,\delta,\gamma,N}^W.
\end{align*}
Since $\#\da_{J,\delta,\gamma,N}^W\le C(\gamma N+2^J\delta N+2^JW)$ and $\|\operatorname{Op}_N^\W(f)\|\le C\|f\|_{C^2}$ (Lemma~\ref{lem:useful}),
in total we have,
\begin{multline}\label{eqn:est-diag-terms}
\sum_{x=0}^{N-1}\langle x|\operatorname{Op}_N^\W(f)|x\rangle(P_{I(N)})_{xx} =\\ \frac{N|I(N)|}{2\pi}\left(\frac{1}{N}\operatorname{Tr}\operatorname{Op}_N^\W(f)+{\|f\|_{C^2}\mathcal{O}(\mathcal{R}_\mathrm{d}(N))}+\|f\|_{C^2}\mathcal{O}\bigg(\frac{\gamma+2^J\delta+\frac{2^JW}{N}}{|I(N)|}\bigg)\right).
\end{multline}
The non-trace terms on the right side are $o(1)$, more precisely $\mathcal{O}\left(\|f\|_{C^2}\mathcal{R}_d(N)\right)+o\big(\frac{\log N}{N^{1/3}}\big)$, by the choices \eqref{eqn:param2}, \eqref{eqn:decayrates}, and growth estimate $|I(N)|\gg\frac{1}{\log N}$.
Since $\frac{1}{N}\operatorname{Tr}\operatorname{Op}_N^\W(f)=\int_{\T^2}f(\x)\,d\x+\mathcal{O}(\frac{\|f\|_{C^3}}{N^3})$, it just remains to show the right-most sum over off-diagonal terms in \eqref{eqn:op-exp} is $o(N|I(N)|)$.

\subsection{Off-diagonal terms}
For the off-diagonal sum $\sum_{\substack{x,y=0\\x\ne y}}^{N-1}\langle y|\operatorname{Op}_N^\W(f)|x\rangle (P_{I(N)})_{xy}$, we will break the sum up into cases depending on the location of $(x,y)$. 
First consider a single term, recalling $\tilde{f}(k)=\int_{\T^2}f(q,p)e^{-2\pi i(qk_2-pk_1)}\,dq\,dp$ and $T(k)=e^{2\pi i(k_2Q-k_1P)}$ (defined in \eqref{eqn:T(k)}),
\begin{align*}
\langle y|\operatorname{Op}_N^\W(f)|x\rangle&=\sum_{k\in\Z^2}\tilde{f}(k)\langle y|T(k)|x\rangle \\
&=\sum_{k\in\Z^2}\tilde{f}(k)e^{-i\pi k_1k_2/N}e^{2\pi ik_2(x+k_1)/N}\langle y|x+k_1\rangle_{\Z/N\Z}\\
&= \sum_{\substack{m_1\in \Z\\k_2\in\Z}}\tilde{f}([y-x]_N+Nm_1,k_2)e^{i\pi [y-x]_Nk_2/N}e^{i\pi m_1k_2}e^{2\pi ik_2x/N},
\end{align*}
where we let $[y-x]_N\in[-N/2,N/2]$ be the representative of $y-x$ in $\Z/N\Z$ with $|[y-x]_N|=d_{\Z/N\Z}(x,y)$.
Since the Fourier coefficients decay away from $(0,0)$ since $f$ is smooth, we split up the sum into cases $m_1=0$ and $m_1\ne0$,
\begin{align*}
|\langle y|\operatorname{Op}_N^\W(f)|x\rangle| 
&\le\sum_{k_2\in\Z}|\tilde{f}([y-x]_N,k_2)|+\sum_{\substack{k_2\in\Z\\m_1\in\Z\setminus\{0\}}}|\tilde{f}([y-x]_N+Nm,k_2)|.
\end{align*}
We have thrown out all the phases here, but we also do not know the phases of $(P_{I(N)})_{xy}$ which multiply it.
Let $M\ge3$ be a fixed integer for the rest of the proof (we will at the end choose a specific value of $M$).
Using the Fourier decay \eqref{eqn:fourier-bound}
$|\tilde{f}(k)|\le\frac{C_M\|f\|_{C^M}}{\|k\|_2^M}$,
then for $x\ne y$,
\begin{align*}
|\langle y|\operatorname{Op}_N^\W(f)|x\rangle| &\le{C_M\|f\|_{C^M}}\bigg(\sum_{k_2\in\Z}\frac{1}{([y-x]_N^2+k_2^2)^{M/2}}
+\sum_{\substack{k_2\in\Z\\m_1\in\Z\setminus\{0\}}}\frac{1}{((Nm_1+[y-x]_N)^2+k_2^2)^{M/2}}\bigg).
\end{align*}
Since these summands are decreasing functions in $|k_2|$, we can approximate by integrals.
The first sum is, for $x\ne y$,
\begin{align*}
\sum_{k_2\in\Z}\frac{1}{([y-x]_N^2+k_2^2)^{M/2}} &\le \int_\R \frac{1}{([y-x]_N^2+z^2)^{M/2}}\,dz+\frac{1}{|[y-x]_N|^M}\\
&\le\frac{c_M}{|[y-x]_N|^{M-1}}.
\end{align*}
The second sum is similar, also using that $|Nm_1+[y-x]_N|\ge N|m_1|-|[y-x]_N|\ge N(|m_1|-\frac{1}{2})$,
\begin{align*}
\sum_{\substack{k_2\in\Z\\m_1\in\Z\setminus\{0\}}}\frac{1}{((Nm_1+[y-x]_N)^2+k_2^2)^{M/2}} 
&\le  \sum_{m_1\in\Z\setminus\{0\}}\frac{c_M}{|Nm_1+[y-x]_N|^{M-1}} \\
&\le \sum_{m_1\in\Z\setminus\{0\}}\frac{c_M}{N^{M-1}(|m_1|-\frac{1}{2})^{M-1}}
\le \frac{c'_M}{N^{M-1}}.
\end{align*}
Thus for any $(x,y)\in\intbrr{1:N-1}^2$ and a new constant $C_M$,
\begin{align}\label{eqn:op-est}
|\langle y|\operatorname{Op}_N^\W(f)|x\rangle|&\le C_{M}\|f\|_{C^M}\left(\frac{1}{|[y-x]_N|^{M-1}}+\frac{1}{N^{M-1}}\right).
\end{align}
The key point is that this decays away from the diagonal in $\intbrr{0:N-1}^2$.
This can be viewed as an analogue of the rapid decay away from the diagonal for the Schwartz kernel of semiclassical pseudodifferential operators on $\R^n$, cf. \cite[\S9.3.2]{zworski}.

\subsubsection{Away from the diagonal}
Let $V=\log N$.
Using just the bound $|(P_{I(N)})_{xy}|\le1$, the decay from \eqref{eqn:op-est} shows that the sum over terms where $|[x-y]_N|>V$ is small, though there are $N^2(1-o(1))$ such terms,
\begin{align*}
\sum_{\substack{x,y=0;\,x\ne y\\|[x-y]_N|> V}}^{N-1}|\langle y|\operatorname{Op}_N^\W(f)|x\rangle| &\le C_{M}\|f\|_{C^M}\sum_{y=0}^{N-1}\sum_{\substack{x:x\ne y\\|[x-y]_N|>V}}\left(\frac{1}{|[y-x]_N|^{M-1}}\right)+C_{M}\|f\|_{C^M}\frac{N^2}{N^{M-1}}\\
&\le C_{M}\|f\|_{C^M}\sum_{y=0}^{N-1}\sum_{d=V+1}^\infty\frac{2}{d^{M-1}}+C_{M}\|f\|_{C^M}\frac{1}{N^{M-3}}\\
&\le \tilde{C}_{M}\|f\|_{C^M}\left(\frac{N}{V^{M-2}}+\frac{1}{N^{M-3}}\right).
\numberthis\label{eqn:est-away-diag}
\end{align*}
This is $o(N|I(N)|)$ since $M\ge 3$ and $\frac{N}{\log N}=o(N|I(N)|)$.

\subsubsection{Near the diagonal}
We are left with the sum over terms near the diagonal. For this we will need to consider whether or not $(x,y)\in\tilde{A}_{j,\delta,\gamma,N}^W$, and we will need to use that we know the structure of $\tilde{A}_{j,\delta,\gamma,N}^W$. There are of order $VN\gg N$ pairs $(x,y)$ with $|[x-y]_N|\le V$, which is too many to make the sum $o(N)$ (and we actually need $o(N|I(N)|)$), without knowing where $(P_{I(N)})_{xy}=o(|I(N)|)$ and where we only have $|(P_{I(N)})_{xy}|\le1$. However, the structure of $\tilde{A}_{J,\delta,\gamma,N}^W$ will mean there are only $o(N|I(N)|)$ pairs of $(x,y)$ such that both $|[x-y]_N|\le V$ and $(x,y)\in \tilde{A}_{J,\delta,\gamma,N}^W$. The remaining $VN(1-o(1))$ pairs $(x,y)$ will have the estimate $(P_{I(N)})_{xy}=o(|I(N)|)$, which with the decay in \eqref{eqn:op-est} will be enough to produce the desired $o(N|I(N)|)$ bound.

The estimate on the number of pairs with  $|[x-y]_N|\le V$ and $(x,y)\in\tilde{A}_{J,\delta,\gamma,N}^W$ is
\begin{align*}
\#\{(x,&y)\in\tilde{A}_{J,\delta,\gamma,N}^W,|[x-y]_N|\le V\} \\
&\le\#\{(x,y)\in{A}_{J,\delta,\gamma,N}^W,|[x-y]_N|\le V\}
+\#\{(y,x)\in{A}_{J,\delta,\gamma,N}^W,|[x-y]_N|\le V\}\,+\\
&\quad\;+\#\{(x,y):x\in\da_{J,\delta,\gamma,N},|[x-y]_N|\le V\}
+\#\{(x,y):y\in\da_{J,\delta,\gamma,N},|[x-y]_N|\le V\}\\
&\le 2\,\#\{(x,y)\in{A}_{J,\delta,\gamma,N}^W,|[x-y]_N|\le V\} + 4V\,\#\da_{J,\delta,\gamma,N}.
\end{align*}
We estimate $\#\{(x,y)\in{A}_{J,\delta,\gamma,N}^W,|[x-y]_N|\le V\}$ similarly as in Figure~\ref{fig:area} in Section~\ref{subsec:sets}; to count the number of $(x,y)$ with $|[x-y]_N|\le V$ and $(x,y)\in C_{J,N}^W$, we compute the area of the shaded region in Figure~\ref{fig:Varea}. 
The estimates are just those for $\#\da_{J,\delta,\gamma,N}$ (equation \eqref{eqn:da}) multiplied by $V+1$ (or by $2V$ for the intersection with $B_{J,\delta,\gamma,N}$ in Figure~\ref{fig:discont2}), yielding

\begin{figure}[!ht]
\makebox[\textwidth][c]{
\begin{subfigure}[t]{.5\textwidth}
\captionsetup{width=.9\linewidth}
\includegraphics{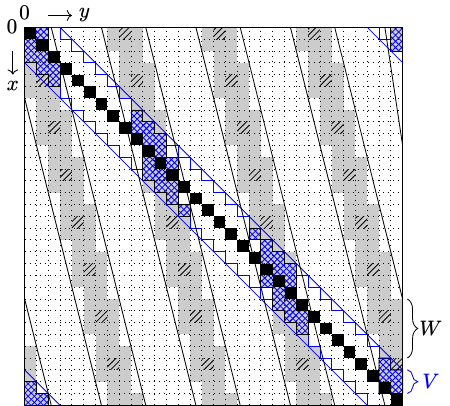}
\caption{Bounding regions for $C_{J,N}^W$ and $\{(x,y)\in C_{J,N}^W\text{ and }|[x-y]_N|\le V\}$.}\label{fig:slopes2}
\end{subfigure}
\hfill
\begin{subfigure}[t]{.5\textwidth}
\captionsetup{width=.9\linewidth}
\includegraphics{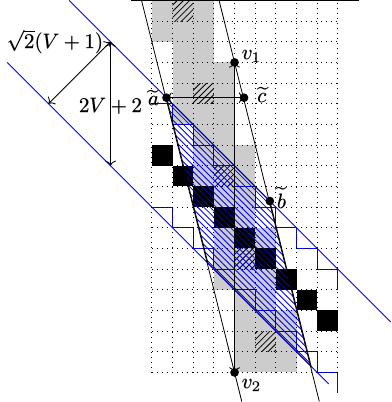}
\caption{Zoomed in diagram of (A), now showing the bounding region with blue northwest hatching.}
\label{fig:Varea}
\end{subfigure}
}
\caption{Coordinates $(x,y)\in C_{J,N}^W$ with $|[x-y]_N|\le V$ with $V=2$.}
\end{figure}

\begin{figure}[!ht]
\includegraphics{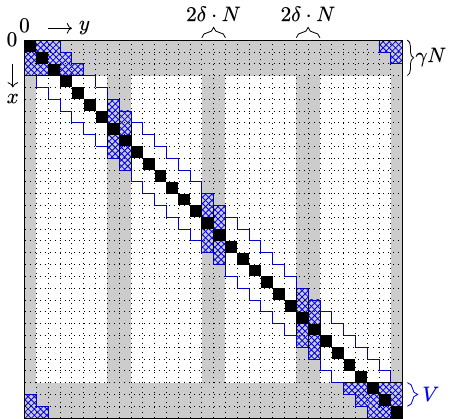}
\caption{Coordinates $(x,y)\in B_{2,\delta,\gamma,N}$ with $|[y-x]_N|\le V$ in blue cross-hatching.}\label{fig:discont2}
\end{figure}

\begin{align*}
\sum_{\substack{(x,y)\in \tilde{A}_{J,\delta,\gamma,N}^W\\x\ne y,\,|[x-y]_N|\le V}}|\langle y|\operatorname{Op}_N^\W(f)|x\rangle||(P_{I(N)})_{xy}|&\le C\|f\|_{C^2}\sum_{\substack{(x,y)\in \tilde{A}_{J,\delta,\gamma,N}^W\\x\ne y,\,|[x-y]_N|\le V}}1\\
&\le C\|f\|_{C^2} C(\gamma NV+2^J\delta NV+2^JWV).
\numberthis\label{eqn:est-diag-bad}
\end{align*}
Evaluating the terms using \eqref{eqn:param2} and $V=\log N$, then \eqref{eqn:est-diag-bad} is $\mathcal{O}(N^{2/3}\log N)=o(N|I(N)|)$.

For the pairs not in $\tilde{A}_{J,\delta,\gamma,N}^W$, we have $(P_{I(N)})_{xy}=|I(N)|\mathcal{O}(\mathcal{R}_\mathrm{od}(N))$, and so using \eqref{eqn:op-est}, 
\begin{align*}
\sum_{\substack{(x,y)\not\in \tilde{A}_{J,\delta,\gamma,N}^W\\x\ne y,\,|[x-y]_N|\le V}}|\langle y|&\operatorname{Op}_N^\W(f)|x\rangle||(P_{I(N)})_{xy}|\\
&\le C_M\|f\|_{C^M}\mathcal{O}(|I(N)|\mathcal{R}_\mathrm{od}(N))\sum_{y=0}^{N-1}\sum_{\substack{x:x\ne y\\|[x-y]_N|\le V}}\left(\frac{1}{|[y-x]_N|^{M-1}}+\frac{1}{N^{M-1}}\right)\\
&\le {C}_M\|f\|_{C^M}\mathcal{O}(|I(N)|\mathcal{R}_\mathrm{od}(N))\left(2Nc_M+\frac{2NV}{N^{M-1}}\right),\numberthis\label{eqn:est-diag-good}
\end{align*}
which is $o(|I(N)|N)$ since $\mathcal{R}_\mathrm{od}(N)\to0$.

Thus combining \eqref{eqn:est-away-diag}, \eqref{eqn:est-diag-bad}, \eqref{eqn:est-diag-good}, we obtain,
\begin{multline}\label{eqn:6.7}
\left|\sum_{\substack{x,y=0\\x\ne y}}^{N-1}\langle y|\operatorname{Op}_N^\W(f)|x\rangle (P_{I(N)})_{xy}\right|\\
\le CN|I(N)|\|f\|_{C^M}\bigg[\frac{1}{|I(N)|}\left(\frac{1}{V^{M-2}}+\frac{1}{N^{M-2}}\right)+\frac{1}{|I(N)|}\left(\gamma V+2^J\delta V+2^JWV/N\right)+\\
+\mathcal{R}_\mathrm{od}(N)\left(1+\frac{V}{N^{M-1}}\right)\bigg],
\end{multline}
which (as we already checked each piece) is $N|I(N)|\|f\|_{C^M}\cdot o(1)$.
Finally, taking $M=3$ and using \eqref{eqn:6.7}, \eqref{eqn:est-diag-terms}, and \eqref{eqn:op-exp}, we obtain the desired result \eqref{eqn:lweyl}.
The error estimate in \eqref{eqn:lweyl} follows as the slowest decay rate in \eqref{eqn:6.7} is $\mathcal{O}(\|f\|_{C^3}\mathcal{R}_{\mathrm{od}}(N))$, and $\mathcal{R}_\mathrm{od}(N)=\frac{1}{|I(N)|J}+r(N)$ with $r(N)\ll\frac{1}{|I(N)|J}$. \qed

\subsection{Windowed quantum ergodicity}\label{subsec:qe}

Once we have the windowed local Weyl law Theorem~\ref{thm:lweyl}, we can replace the usual local Weyl law (which involves an average over all eigenstates) with the windowed version in the standard proof of quantum ergodicity. 
For completeness, and because the discontinuities here require some consideration, we provide the details for Theorem~\ref{thm:qv} and Corollary~\ref{cor:qe} in 
Appendix~\ref{sec:wqe}. 
We also note that the (non-windowed) quantum ergodicity proof in \cite{DNW} used a slightly different approach starting from exponential decay of classical correlations; we use the method described in \cite[\S15.4]{zworski}, \cite{MOK} since it is a bit easier to replace the local Weyl law with the windowed version in this procedure.

\section{Random quasimodes}\label{sec:rw} 

For $I(N)$ with $|I(N)|\log N\to\infty$, let $S_{I(N)}=\operatorname{span}\{\varphi^{(j)}:\theta^{(j)}\in I(N)\}$.
Consider a random quasimode $\psi$ that is a random linear combination of the eigenvectors in $S_{I(N)}$,
\begin{align}\label{eqn:rw}
\psi(x) &= \frac{1}{\sqrt{\dim S_{I(N)}}}\sum_{j:\theta^{(j,N)}\in I(N)} g_j \varphi^{(j,N)}(x),\qquad \text{for }g_i\text{ iid }N_\C(0,1).
\end{align}
By Gaussian concentration, for large $N$ (which implies large $\dim S_{I(N)}$ by Corollary~\ref{cor:weyl}), this normalization means $\psi$ is approximately a random vector chosen according to Haar measure from the unit sphere of $S_{I(N)}$.

Since $I(N)$ are allowed to shrink, we can take $I(N)=[\theta,\theta+o(1)]$ for some fixed $\theta\in\R/(2\pi\Z)$ and a rate $|I(N)|=o(1)$ such that $|I(N)|\log N\to\infty$. In this case the eigenangle is asymptotically fixed as $\theta$.
In general though, we are free to take any sequence of intervals with $|I(N)|\log N\to\infty$.

Taking the expected value over the random coefficients shows the mean of $\psi$ is the zero vector, and gives the covariance relation (including for $x=y$)
\begin{align*}
\E \psi(x)\overline{\psi(y)} &= \frac{1}{\dim S_{I(N)}}(P_{I(N)})_{xy},
\end{align*}
along with $\E\psi(x)\psi(y)=0$, which completely determine the behavior of $\psi$ as a complex Gaussian vector. The point of Theorem~\ref{thm:rw} is that the asymptotic behavior of the matrix entries of $P_{I(N)}$ in Theorem~\ref{thm:Pmat} is enough to imply the desired statistics.

\subsection{Proof of Theorem~\ref{thm:rw}}\label{subsec:rw-proof}
Let $(\Omega_N,\mathbb{P})$ be the probability space from which $\psi$ is drawn.

\subsubsection{Part (i)}
Gaussian value statistics will follow from a characteristic function argument, which was used in \cite{DiaconisFreedman} to solve the ``projection pursuit'' problem. 
The application here is similar to that in \cite{pw}, though is simpler here  since we only consider a single random vector $\psi$, rather than an entire orthonormal basis of random vectors. In this case, we can just use the complex version of the theorem in \cite{DiaconisFreedman} instead of the later quantitative versions developed in \cite{Meckes,ChatterjeeMeckes}. 
(The latter can still be used to form a full basis of random quasimodes with the desired property, similarly as used in Theorem~\ref{thm:walsh-main} or \cite[Theorem 2.5]{pw}.)

\begin{thm}[Adapted from Theorem 1.1 in \cite{DiaconisFreedman}]\label{thm:conv}
Let $N\in\N$, and let $P^{(N)}$ be an $N\times N$ self-adjoint projection matrix onto a subspace $V^{(N)}$ of $\C^N$. 
Suppose there is a function $1\le L(N)\le N$, with $L(N)\to\infty$ as $N\to\infty$, such that for any $\varepsilon>0$, as $N\to\infty$,
\begin{align}
\label{eqn:dfnorm}\frac{1}{N}\#\left\{x\in\intbrr{0:N-1}:\left|\big\|P^{(N)}e_x\big\|_2^2-\frac{L(N)}{N}\right|>\varepsilon\frac{L(N)}{N}\right\} & \to 0\\
\label{eqn:dfinnerproduct}\frac{1}{N^2}\#\left\{(x,y)\in\intbrr{0:N-1}^2:\left|\langle P^{(N)}e_x,P^{(N)}e_y\rangle\right|>\varepsilon\frac{L(N)}{N}\right\}& \to0,
\end{align}
where $e_x$ be the $x$th standard basis vector. Let $z\sim N_\C(0,I_N)$ and set $v:=P^{(N)}z\sim N_\C(0,P^{(N)})$, and define the empirical distribution $\mu_z^{(N)}$ of the coordinates of $v$ scaled by $\sqrt{N/L(N)}$,
\begin{equation*}
\mu_z^{(N)} := \frac{1}{N}\sum_{x=0}^{N-1} \delta_{\frac{\sqrt{N}}{\sqrt{L(N)}}\langle P^{(N)}z,e_x\rangle}.
\end{equation*}
Then $\mu_z^{(N)}$ converges weakly in probability to the standard complex Gaussian $N_\C(0,1)$, i.e. for any $f:\C\to\C$ bounded and continuous and any $\varepsilon>0$,
\begin{equation*}
\mathbb{P}\left[\left|\int f\,d{\mu}_z^{(N)}-\mathbb{E}_{N_\C(0,1)}[f]\right|>\varepsilon\right] \xrightarrow{N\to\infty} 0.
\end{equation*}
\end{thm}
In our case we set $L(N)=\dim S_{I(N)}$ and $P^{(N)}=P_{I(N)}$. From the definition of $\psi(x)$, we see that $\frac{\sqrt{N}}{\sqrt{L(N)}}\langle P^{(N)}z,e_x\rangle\sim \sqrt{N}\psi(x)$ in distribution. By Corollary~\ref{cor:weyl}, $\frac{L(N)}{N}=\frac{1}{N}\dim S_{I(N)}=\frac{|I(N)|}{2\pi}(1+o(1))$. Then
Theorem~\ref{thm:Pmat} shows the conditions \eqref{eqn:dfnorm} and \eqref{eqn:dfinnerproduct} hold. 
Theorem~\ref{thm:conv} then implies the empirical measure $\frac{1}{N}\sum_{x=0}^{N-1}\delta_{\sqrt{N}\psi(x)}$ converges weakly in probability to $N_\C(0,1)$.

\subsubsection{Part (ii)}\label{subsubsec:hw}

This follows from the Hanson--Wright inequality (specifically, the version \cite[Thoerem 1.1]{RudelsonVershynin-HW}) and windowed generalized Weyl law Theorem~\ref{thm:lweyl}.
Let $d=\dim S_{I(N)}$, and let $M_{I(N)}$ be the $N\times d$ matrix whose columns are the eigenvectors $\varphi^{(j)}$ in $S_{I(N)}$, so that $P_{I(N)}=M_{I(N)}M_{I(N)}^\dagger$ and we can set $\psi=\frac{1}{\sqrt{d}}M_{I(N)}g$ where $g\sim N_\C(0,I_d)$. Then
\begin{align*}
\langle \psi|\operatorname{Op}_N^\W(a)|\psi\rangle_{\C^N} &= \frac{1}{d}\langle g|M_{I(N)}^\dagger\operatorname{Op}_N^\W(a)M_{I(N)}|g\rangle_{\C^d}.
\end{align*}
The Hanson--Wright inequality gives concentration about the mean for such a quadratic form. Using Theorem~\ref{thm:lweyl}, the mean is 
\begin{align*}
\E\left[\frac{1}{d}\langle g|M_{I(N)}^\dagger\operatorname{Op}_N^\W(a)M_{I(N)}|g\rangle_{\C^d}\right]&= \frac{1}{d}\operatorname{Tr}\big(M_{I(N)}^\dagger\operatorname{Op}_N^\W(a)M_{I(N)}\big) \\
&= \frac{1}{d}\sum_{\theta^{(j)}\in I(N)}\langle\varphi^{(j)}|\operatorname{Op}_N^\W(a)|\varphi^{(j)}\rangle
=\int_{\T^2}a(\x)\,d\x +\mathfrak{r}_N(a),
\end{align*}
where $|\mathfrak{r}_N(a)|\le \mathfrak{r}_N(1+\|a\|_{C^3})$
for some $\mathfrak{r}_N\to0$.
By the Hanson--Wright inequality applied with $g\sim N_\C(0,I_d)$, then for $t\ge 2|\mathfrak{r}_N(a)|$ so that $t-|\mathfrak{r}_N(a)|\ge t/2$,
\begin{align*}
\P\left[\left|\langle\psi|\operatorname{Op}_N^\W(a)|\psi\rangle - \int_{\T^2}a(\x)\,d\x\right|>t\right] &\le \P\left[\left|\langle\psi|\operatorname{Op}_N^\W(a)|\psi\rangle - \E\langle \psi|\operatorname{Op}_N^\W(a)|\psi\rangle\right|>t-|\mathfrak{r}_N(a)|\right]\\
&\le 2\exp\left[-C\min \left(\frac{t^2}{\|a\|_{C^2}^2},\frac{t}{\|a\|_{C^2}}\right)d\right].
\end{align*}
Recall $d=\frac{|I(N)|N}{2\pi}(1+o(1))$, and let $\varepsilon_N=\max(2\mathfrak{r}_N,d^{-1/4})\to0$. Then since $\varepsilon_N(1+\|a\|_{C^3})\ge 2|\mathfrak{r}_N(a)|$,
\begin{align}\label{eqn:hwbound}
\P\left[\left|\langle\psi|\operatorname{Op}_N^\W(a)|\psi\rangle - \int_{\T^2}a(\x)\,d\x\right|>(1+\|a\|_{C^3})\varepsilon_N\right] 
&\le 2\exp\left[-C|I(N)|^{1/2}N^{1/2}\right].
\end{align}
Finally, letting $(a_\ell)_{\ell=1}^\infty$ be a countable dense set in $C^\infty(\T^2)$ in the $C^3$ norm, define
\begin{align*}
\Gamma_N=\left\{\psi_N\in\Omega_N:\forall \ell\in\intbrr{1:N},\,\left|\langle\psi_N|\operatorname{Op}_N^\W(a_\ell)|\psi_N\rangle-\int_{\T^2}a_\ell(\x)\,d\x\right|\le(1+\|a_\ell\|_{C^3})\varepsilon_N\right\}.
\end{align*}
Then $\P[\Gamma_N^c]\le 2N\exp(-C|I(N)|^{1/2}N^{1/2})\to0$ by \eqref{eqn:hwbound}, and using $\|\operatorname{Op}_N^\W(a-a_\ell)\|\le C\|a-a_\ell\|_{C^2}$, one obtains for any sequence of $\psi_N$ with $\psi_N\in\Gamma_N$, that
\begin{align*}
\lim_{N\to\infty}\langle \psi_N |\operatorname{Op}_N^\W(a)|\psi_N\rangle = \int_{\T^2}a(\x)\,d\x,\quad \forall a\in C^\infty(\T^2). 
\end{align*}

\subsubsection{Part (iii)}

Since $\psi(x)\sim \frac{1}{\sqrt{\dim S_{I(N)}}}N_\C(0,(P_{I(N)})_{xx})$, then for $x\not\in\da_{J,\delta,\gamma,N}^W$ with parameters as in \eqref{eqn:param2}, \eqref{eqn:P-diag} implies,
\begin{align}
\E|\sqrt{N}\psi(x)|^{m} &= \frac{N^{m/2}}{(\dim S_{I(N)})^{m/2}}(P_{I(N)})_{xx}^{m/2}\,\E|g|^m=\E|g|^m(1+o(1)),
\end{align}
where $g\sim N_\C(0,1)$.

The autocorrelation function computations are immediate by Isserlis' (or Wick's) theorem. For $(x,y)\not\in\tilde{A}_{J,\delta,\gamma,N}^W$ with $x\ne y$, Isserlis' theorem for complex Gaussians followed by equations \eqref{eqn:P-diag} and \eqref{eqn:P-offdiag} implies,
\begin{align*}
\E[|\psi(x)|^2|\psi(y)|^2] &= \E[|\psi(x)|^2]\E[|\psi(y)|^2]+\E\big[\psi(x)\overline{\psi(y)}\big]\E\big[\overline{\psi(x)}\psi(y)\big]\\
&=\frac{1}{(\dim S_{I(N)})^2}\left[\left(\frac{|I(N)|}{2\pi}\right)^2(1+o(1))+o(|I(N)|^2)\right],
\end{align*}
and so 
\begin{align}
\E[N^2|\psi(x)|^2|\psi(y)|^2]&=1+o(1).
\end{align} 
This matches in the limit $N\to\infty$ with the value for $g,g'\sim N_\C(0,1)$ iid, which is  $\E[|g|^2|g'|^2]=1$.

For $x\in\da_{J,\delta,\gamma,N}^W$ or $(x,y)\in\tilde{A}_{J,\delta,\gamma,N}^W$, we do not necessarily have the above standard Gaussian behavior, as the covariances in $P_{I(N)}$ may be different (e.g. see Figure~\ref{fig:P} numerically). Here we demonstrate an explicit example where the variance $(P_{I(N)})_{xx}$ is not $\frac{|I(N)|}{2\pi}(1+o(1))$, and so the moments $\E|\sqrt{N}\psi(x)|^m$ will be a different value.
\begin{prop}[exceptional coordinate]\label{lem:notall}
Let $I(N)=[-\pi/2,\pi/2]$. Then there is a sequence of $N\to\infty$ such that 
\begin{align}
(P_{I(N)})_{00}\ge 0.89182655+o(1).
\end{align}
\end{prop}
\begin{proof}[Proof of Proposition~\ref{lem:notall}]
We first show that for any $k\in\N$ with $2^k$ dividing $N$, that $(\hat{B}_N^k)_{00}=2^{-k/2}$. In fact we will show recursively that the first row $(\hat{B}_N^k)_{0y}$ is $2^{-k/2}\delta_{y\in \frac{N}{2^k}\Z}$ for such $k$. Let $r_j=\frac{N}{2^j}\in\N$ for any $1\le j\le k$. We note that for any \emph{even} row index $x=0,2,\ldots,N-2$ of $\hat{B}_N$, that $(\hat{B}_N)_{xy}=\frac{1}{\sqrt{2}}(\delta_{x/2}(y)+\delta_{x/2+N/2}(y)) $.
Now assuming $(\hat{B}_N^{k-1})_{0y}=\left(\frac{1}{2^{(k-1)/2}}\delta_{y\in r_{k-1}\Z}\right)_y$, we compute,
\begin{align*}
(\hat{B}_N^k)_{0y} &= \sum_{\ell=0}^{N-1}\langle 0|\hat{B}_N^{k-1}|\ell\rangle\langle\ell|\hat{B}_N|y\rangle \\
&=\frac{1}{2^{(k-1)/2}}\sum_{j=0}^{2^{k-1}-1}\langle jr_{k-1}|\hat{B}_N|y\rangle
=\frac{1}{2^{k/2}}\sum_{j=0}^{2^{k-1}-1}(\delta_{r_{k-1}j/2}(y)+\delta_{r_{k-1}j/2+N/2}(y)),
\end{align*}
since $r_{k-1}=\frac{N}{2^{k-1}}=2r_k$ is even. Since $\frac{r_{k-1}j}{2}=r_kj$, we see the above is $(\hat{B}_N^k)_{0y}=2^{-k/2}\delta_{y\in r_k\Z}$, as desired. 

Now for convenience, take $N=2^K$, so that $r_j=\frac{N}{2^j}\in\N$ for any $j\le K$. (One can take other sequences of $N$ and apply the same argument, as long as the largest power of $2$ dividing $N$ is growing.) 
Now we apply the argument in \cite[\S9]{pw}, using the piecewise continuous approximation $h_{\delta}$ to $\Chi_{[-\pi/2,\pi/2]}$,
\[
h_\delta(x)=\begin{cases}
1,&-\frac{\pi}{2}+\delta\le x\le\frac{\pi}{2}-\delta\\
\frac{1}{\delta}\left(x+\frac{\pi}{2}\right),&-\frac{\pi}{2}\le x\le-\frac{\pi}{2}+\delta\\
-\frac{1}{\delta}\left(x-\frac{\pi}{2}\right),&\frac{\pi}{2}-\delta\le x\le\frac{\pi}{2}\\
0,&|x|\ge\frac{\pi}{2}
\end{cases}.
\]
Computing with the $K$th partial Fourier sums $S_K h_\delta$ for $\delta=K^{-3/4}$ eventually implies (see \cite[\S9]{pw}),
\begin{align*}
(P_{I(N)})_{00} \ge (h_\delta(\hat{B}_N))_{00} &=((S_Kh_\delta)(\hat{B}_N))_{00}+o(1)\\
&\ge 0.89182655-o(1).
\end{align*}
\end{proof}
One also expects off-diagonal coordinates $(x,y)$ corresponding to the graph of small powers of the classical map $B$ to show deviations, for example those seen in Figure~\ref{fig:P}.
We note that the proof of Proposition~\ref{lem:notall} here is specific to the Balasz--Voros quantization $\hat{B}_N$. However similar pictures as Figure~\ref{fig:P} for e.g. the Saraceno quantization \cite{Saraceno} suggest similar exceptional coordinates exist as well.

\subsubsection{Part (iv)}
Using that $\psi(x)\sim \frac{1}{\sqrt{\dim S_{I(N)}}}N_\C(0,(P_{I(N)})_{xx})$ and that $\#\da_{J,\delta,\gamma,N}^W=\mathcal{O}(N^{2/3})$ for parameters in \eqref{eqn:param2},
then with Theorem~\ref{thm:Pmat},
\begin{align*}
\E\|\psi\|_p^p &= \sum_{x\not\in\da_{J,\delta,\gamma,N}^W}\E|\psi(x)|^p+\sum_{x\in\da_{J,\delta,\gamma,N}^W}\E|\psi(x)|^p \\
&=N(1-o(1))\frac{\E|g|^p}{N^{p/2}}+\mathcal{O}(N^{2/3})\frac{\E|g|^p(2\pi)^{p/2}}{N^{p/2}|I(N)|^{p/2}}\\
&=\frac{\E|g|^p}{N^{p/2-1}}\left[1-o(1)+\frac{\mathcal{O}(N^{-1/3})(2\pi)^{p/2}}{|I(N)|^{p/2}}\right]
=\frac{\E|g|^p}{N^{p/2-1}}\left[1+o(1)\right],
\end{align*}
since $|I(N)|\gg\frac{1}{\log N}$ so $N^{1/3}|I(N)|^{p/2}\to\infty$ for any $0<p<\infty$.

For $p=\infty$, we just use that the expected maximum of $N$ centered real subgaussian variables is upper bounded by $\sqrt{2\sigma_\mathrm{max}^2\log N}$, which holds without any covariance relation assumptions. Here we have $N$ (complex) Gaussians $\psi(x)\sim \frac{1}{\sqrt{\dim S_{I(N)}}}N_\C(0,(P_{I(N)})_{xx})$, for $x\in\intbrr{0:N-1}$. Thus since $\sigma_x^2=\frac{(P_{I(N)})_{xx}}{\dim S_{I(N)}}\le\frac{1}{\dim S_{I(N)}}=\frac{2\pi}{N|I(N)|}(1+o(1))$,
\begin{align*}
\E\left[\max_{x\in\intbrr{0:N-1}}|\psi(x)|\right] &\le \frac{C\sqrt{\log N}}{\sqrt{N|I(N)|}}(1+o(1))= \frac{o(\log N)}{\sqrt{N}},
\end{align*}
since $\frac{1}{|I(N)|}=o(\log N)$.

\subsubsection{Part (v)}
Let $\mathfrak{Z}_\psi:=\{x\in\Z_N:\Re\psi(x)=0\text{ or }\Im\psi(x)=0\}$. Such a zero value can only happen (with nonzero probability) if $(P_{I(N)})_{xx}=0$; if $(P_{I(n)})_{xx}\ne0$ then we may assume $\Re\psi(x)\ne0,\Im\psi(x)\ne0$.
 We choose parameters as in \eqref{eqn:param2} and assume $N$ is large enough that the asymptotics \eqref{eqn:P-diag} start to kick in, so that $(P_{I(N)})_{xx}\ne0$ for all $x\not\in\da_{J,\delta,\gamma,N}^W$, and so $\mathfrak{Z}_\psi\subseteq\da_{J,\delta,\gamma,N}^W$.

First we want to show for any $f\in C(\R/\Z)$, that $\frac{2}{N}\E\sum_{x\in Z_{N,\psi}^r}f(x/N)\to\int_0^1f(t)\,dt$ as $N\to\infty$. Write,
\begin{align*}
\E\sum_{x\in Z_{N,\psi}^r}f\Big(\frac{x}{N}\Big) &= \sum_{x=0}^{N-1}f\Big(\frac{x}{N}\Big)\E[\oneb_{Z_{N,\psi}^r}(x)]\\
&=\sum_{\substack{x\in\Z_N:\\(x,x+1)\not\in \tilde{A}_{J,\delta,\gamma,N}^W}}f\Big(\frac{x}{N}\Big)\E[\oneb_{Z_{N,\psi}^r}(x)] +\mathcal{O}\big(\|f\|_\infty \#\{x:(x,x+1)\in\tilde{A}_{J,\delta,\gamma,N}^W\}\big),
\end{align*}
where we chose the parameters $J,\delta,\gamma,W$ as in \eqref{eqn:param2} in Theorem~\ref{prop:mpowers}.
The same computation as for estimating the size of $\da_{J,\delta,\gamma,N}^W$ (since both are just counting intersections with a line of slope $-1$) gives the estimate, 
\begin{align*}
\#\{x:(x,x+1)\in\tilde{A}_{J,\delta,\gamma,N}^W\}&\le\#\{x:(x,x+1)\in\tilde{A}_{J,\delta,\gamma,N}^W\}\\
&\le 4\#\da_{J,\delta,\gamma,N}^W.
\end{align*}
This is $o(N^{2/3})$ by the choice of parameters \eqref{eqn:param2}.
Then for $(x,x+1)\not\in\tilde{A}_{J,\delta,\gamma,N}^W$, we have $x,x+1\not\in\da_{J,\delta,\gamma,N}^W$ and $(P_{I(N)})_{xx},(P_{I(N)})_{x+1,x+1}\ne0$, then
\begin{align*}
\E[\oneb_{Z_{N,\psi}^r}(x)]&=\P[\Re\psi(x),\Re\psi(x+1)\text{ have different signs}]\\
&=2\P[\Re\psi(x)>0,\Re\psi(x+1)<0].
\end{align*}
We can explicitly compute this in terms of a $2\times 2$ submatrix of $P_{I(N)}$, corresponding to the coordinates $x$ and $x+1\;\mathrm{mod}\;N$. Since $\psi\sim N_{\C^N}(0,P_{I(N)},0)$, then $\Re\psi\sim N(0,\frac{1}{2}\Re P_{I(N)})$, and $(\Re\psi(x),\Re\psi(x+1))\sim N(0,\Sigma)$, where $\Sigma=\begin{pmatrix}(P_{I(N)})_{xx}&\Re (P_{I(N)})_{x,x+1}\\
\Re (P_{I(N)})_{x+1,x}&(P_{I(N)})_{x+1,x+1}\end{pmatrix}$.
The diagonal of $P_{I(N)}$ is real, and for the off-diagonal all we care is that it is small, so a bound like $|(\Re P_{I(N)})_{xy}|\le |(P_{I(N)})_{xy}|=o(|I(N)|)$ will be sufficient.

One explicitly computes that for $(X,Y)\sim N(0,\Sigma)$, that
\begin{align*}
\P[X>0,Y<0] &= \frac{1}{2}-\frac{1}{2\pi}\cos^{-1}\left(\frac{-\Sigma_{12}}{\sqrt{\Sigma_{11}\Sigma_{22}}}\right).
\end{align*}
Then for $(x,x+1)\not\in\tilde{A}_{J,\delta,\gamma,N}$, 
\begin{align*}
\P[\Re\psi(x)>0,\Re\psi(x+1)<0] &= \frac{1}{2}-\frac{1}{2\pi}\left(\frac{\pi}{2}+o(1)\right) = \frac{1}{4}+o(1).
\end{align*}
Thus
\begin{align*}
\frac{2}{N}\E\sum_{x\in Z_{N,\psi}^r}f\Big(\frac{x}{N}\Big) &= \frac{2}{N}\sum_{\substack{x\in\Z_N:\\(x,x+1)\not\in \tilde{A}_{J,\delta,\gamma,N}^W}}f\Big(\frac{x}{N}\Big)\Big(\frac{1}{2}+o(1)\Big)+\mathcal{O}(\|f\|_\infty N^{-1/3})\\
&= (1+o(1))\frac{1}{N}\sum_{x=0}^{N-1}f\Big(\frac{x}{N}\Big)+\mathcal{O}(\|f\|_\infty N^{-1/3})\xrightarrow{N\to\infty}\int_0^1f(x)\,dx,
\end{align*}
as desired.
Since $\Im\psi\sim N(0,\frac{1}{2}\Re P_{I(N)})$ as well, the same result holds with $Z_{N,\psi}^i$ in place of $Z_{N,\psi}^r$. 
To obtain the mean number of sign changes, take $f\equiv1$ in the above. \qed

One could in principle compute other quantities such as moments of the $Z_{N,\psi}^r$ in terms of the matrix entries of $P_{I(N)}$, though the asymptotics we use from Theorem~\ref{thm:Pmat} for the mean are in general likely not enough; one should need more precise information on the off-diagonal values of $P_{I(N)}$.

\section{Random eigenstates of the Walsh quantized baker map}\label{sec:walsh}

\subsection{Walsh quantization background}

We provide an overview on Walsh quantization on the torus here. For further details, see \cite{AN}. 
The Walsh transform $W_{D^k}$ and Walsh quantization $B_k^\Wa$ of the $D$-baker map were written out in Section~\ref{sec:walshintro}. Here we define Walsh coherent states and quantization of observables. Recall we view the Hilbert space $\mathcal{H}_N$ for $N=D^k$ as the $k$-fold tensor product $(\C^D)^{\otimes k}$.

Let $\ell\in\intbrr{0:k}$, and for $\eta\in\intbrr{0:D-1}$, let $|\eta\rangle\in\C^D$ be the standard basis vector for the $\eta$th coordinate.
For $\varepsilon=\varepsilon_1\ldots\varepsilon_\ell\in\intbrr{0:D-1}^\ell$ and $\varepsilon'=\varepsilon_{\ell+1}\ldots\varepsilon_{k}\in\intbrr{0:D-1}^{k-\ell}$, the $(k,\ell)$-coherent state $|\varepsilon'\cdot\varepsilon\rangle$ is
\begin{align*}
|\varepsilon'\cdot\varepsilon\rangle = |{\varepsilon_1}\rangle\otimes\cdots\otimes |{\varepsilon_\ell}\rangle\otimes\hat{F}_D^\dagger |{\varepsilon_k}\rangle\otimes\cdots\otimes\hat{F}_D^\dagger |{\varepsilon_{\ell+1}}\rangle.
\end{align*}
(The notation with the separator $\cdot$ is reminiscent of viewing $B$ using symbolic dynamics with a decimal place separating $q$ and $p$.
While the symbol notation and indexing used here for $\varepsilon'$ may not be standard, it is convenient for keeping track of cyclic rotations in the proof of Proposition~\ref{prop:walsh-powers}.)
When $\ell=k$, then this reduces to the position basis, and in this case we may denote position coordinates by $x$ or $y$ ranging from $0$ to $D^k-1$, rather than by $|\varepsilon\rangle$.
We will in this case use the relation $x=\sum_{m=1}^k\varepsilon_m D^{k-m}$.
The coherent state $|\varepsilon'\cdot\varepsilon\rangle$ is localized on a quantum rectangle $[\varepsilon'\cdot\varepsilon]=\{(q,p)\in\T^2:
b(\varepsilon)\le q<b(\varepsilon)+ D^{-\ell},b(\varepsilon')\le p<b(\varepsilon')+D^{-(k-\ell)}\}$ of area $D^{-k}$, where $b(\delta_1,\ldots,\delta_j)=\sum_{i=1}^j\delta_i D^{-i}$ is the value of the $D$-ary number with leading digits $\delta_1,\ldots,\delta_j$ followed by zeros.
The set of all $(k,\ell)$-coherent states forms an orthonormal basis of $\mathcal{H}_{D^k}$. The index set for the coherent states or quantum rectangles will be denoted
\begin{align*}
\mathcal{R}^{k,\ell}:=\{[\varepsilon'\cdot\varepsilon]:\varepsilon\in\intbrr{0:D-1}^\ell,\,\varepsilon'\in\intbrr{0:D-1}^{k-\ell}\}.
\end{align*}
The Walsh--anti-Wick quantization of a classical observable $a\in \mathrm{Lip}(\T^2)$ is
\begin{align*}
\operatorname{Op}_{k,\ell}(a):=D^k\sum_{[\varepsilon'\cdot\varepsilon]\in\mathcal{R}^{k,\ell}}|\varepsilon'\cdot\varepsilon\rangle\langle\varepsilon'\cdot\varepsilon|\int_{[\varepsilon'\cdot\varepsilon]}a(\x)\,d\x.
\end{align*}
With these definitions, it was shown in \cite[\S3]{AN} that $B_k^\Wa$ satisfies a classical-quantum correspondence principle (Egorov theorem) and quantum ergodic theorem in the semiclassical limit $\ell(k)\to\infty$, $k-\ell(k)\to\infty$ as $k\to\infty$. Additionally, due to the tensor product structure,
$(B_k^\Wa)^k=(\hat{F}_D^\dagger)^{\otimes k}$,
so that for $D=2$, $(B_k^\Wa)^{2k}=I_{2^k}$, and for $D\ge3$, $(B_k^\Wa)^{4k}=I_{D^k}$. The eigenvalues of the $D^k\times D^k$ matrix $B_k^\Wa$ are thus $(4k)$th roots of unity for $D\ge3$ and $(2k)$th roots of unity for $D=2$. Each eigenspace has high degeneracy, with the same leading order dimension (Corollary~\ref{cor:mult}).

\subsection{Main intermediate results}
The main result we need to prove Theorem~\ref{thm:walsh-main} is the following projection matrix estimates. The proof to then go from Theorems~\ref{thm:walsh-proj} and \ref{thm:walsh-lweyl} below to Theorem~\ref{thm:walsh-main} is similar to the proof of Theorem~\ref{thm:rw} or to the proof of \cite[Theorem 2.5]{pw}. 

\begin{thm}[Projection matrix estimates]\label{thm:walsh-proj}
Let $D\ge3$. For each $k$ choose an arbitrary $\ell=\ell(k)$, and consider the family of $(k,\ell)$-coherent states. For $\jj=0,\ldots,4k-1$, let $P_\jj$ be the orthogonal projection onto the eigenspace $E_{(\jj)}$ of $e^{2\pi i\jj/(4k)}$. There is a subset $G_{k,\ell}$ of $(k,\ell)$-coherent states $|\varepsilon'\cdot\varepsilon\rangle$ for which the following the diagonal estimates hold for any $\jj\in\intbrr{0:4k-1}$ in the limit $k\to\infty$,
\begin{align}\label{eqn:wdiag}
\langle \varepsilon'\cdot\varepsilon|P_\jj|\varepsilon'\cdot\varepsilon\rangle &= \frac{1}{4k}(1+o(1)), \quad|\varepsilon'\cdot\varepsilon\rangle\in G_{k,\ell},
\end{align}
and $\#G_{k,\ell}\ge D^k\left(1-o\big(\frac{1}{4k}\big)\right)$.
We also have the off-diagonal estimates for all $\jj$,
\begin{align}\label{eqn:woffdiag}
\langle \delta'\cdot\delta|P_\jj|\varepsilon'\cdot\varepsilon\rangle &=o\bigg(\frac{1}{4k}\bigg), \quad |\delta'\cdot\delta\rangle\ne|\varepsilon'\cdot\varepsilon\rangle,\; (|\delta'\cdot\delta\rangle,|\varepsilon'\cdot\varepsilon\rangle)\in GP_{k,\ell}\subseteq \intbrr{1:D^k}^2,
\end{align}
where $GP_{k,\ell}$ is a set of $(k,\ell)$-coherent state pairs with $\#GP_{k,\ell}\ge (D^k)^2\left(1-o\big(\frac{1}{4k}\big)\right)$, again in the limit $k\to\infty$.

Additionally, specializing to the position basis, there are at least $D^k\left(1-o\big(\frac{1}{4k}\big)\right)$ coordinates $x\in\intbrr{1:D^k}$ such that both $\langle x|P_\jj|x+1\rangle=o\big(\frac{1}{4k}\big)$ and $\langle x+1|P_\jj|x\rangle=o\big(\frac{1}{4k}\big)$.

All rates of decay above depend only on $k$ and $D$, and can be taken independent of $\jj$, $\ell$, and of the particular element in $G_{k,\ell}$ or $GP_{k,\ell}$.

When $D=2$, all of the above hold with $4k$ replaced by $2k$.
\end{thm}
As we will see in the proof (Section~\ref{subsec:proof-thm8.2}), more precise estimates can be given for the error terms in \eqref{eqn:wdiag} and \eqref{eqn:woffdiag} as well as for the sizes and locations of the good sets $G_{k,\ell}$ and $GP_{k,\ell}$. For example, setting the parameter $r(k)=k/2$ in the proof gives the estimate
\begin{align}\label{eqn:walsh-gset}
\# G_{k,\ell}^c &\le \mathcal{O}(D^{k/2}).
\end{align}

Taking the trace of $P_\jj$ like in Section~\ref{subsec:evcounting}, Eq.~\eqref{eqn:wdiag} shows
\begin{cor}\label{cor:mult}
For $D\ge3$, the degeneracy of each eigenspace is $\frac{D^k}{4k}(1+o(1))$. For $D=2$, the degeneracy of each eigenspace is $\frac{2^k}{2k}(1+o(1))$.
\end{cor}

Since observables $\operatorname{Op}_{k,\ell}(a)$ are diagonal in the $(k,\ell)$-coherent state basis, Eq.~\eqref{eqn:wdiag} also implies,
\begin{thm}[generalized Weyl law in a single eigenspace]\label{thm:walsh-lweyl}
Let $D\ge3$, and for $\jj=0,\ldots,4k-1$, let $E_{(\jj)}$ be the eigenspace of $e^{2\pi i\jj/(4k)}$ for $B_k^\Wa$. Choose any orthonormal basis $(\phi^{(m)})_m$ of $E_{(\jj)}$. Then for any $a\in C^\infty(\T^2)$,
\begin{align}
\frac{4k}{D^k}\sum_{\phi^{(m)}\in E_{(\jj)}}\langle\phi^{(m)}|\operatorname{Op}_{k,\ell}(a)|\phi^{(m)}\rangle =\int_{\T^2}a(\x)\,d\x+o(1)\|a\|_\infty,\quad \text{as }k\to\infty.
\end{align}
The same holds for $D=2$ with $4k$ replaced by $2k$.
\end{thm}
\begin{proof}
Since $\operatorname{Op}_{k,\ell}(a)$ is diagonal in the basis of $(k,\ell)$-coherent states $|\varepsilon'\cdot\varepsilon\rangle$, then
\begin{align*}
\sum_{\phi^{(m)}\in E_{(\jj)}}\langle\phi^{(m)}|\operatorname{Op}_{k,\ell}(a)|\phi^{(m)}\rangle &=\sum_{\phi^{(m)}\in E_{(\jj)}}\sum_{|\varepsilon'\cdot\varepsilon\rangle}\langle\phi^{(m)}|\varepsilon'\cdot\varepsilon\rangle\langle\varepsilon'\cdot\varepsilon|\operatorname{Op}_{k,\ell}(a)|\varepsilon'\cdot\varepsilon\rangle\langle\varepsilon'\cdot\varepsilon|\phi^{(m)}\rangle\\
&= \sum_{|\varepsilon\cdot\varepsilon\rangle}\langle\varepsilon'\cdot\varepsilon|P_\jj|\varepsilon'\cdot\varepsilon\rangle\langle \varepsilon'\cdot\varepsilon|\operatorname{Op}_{k,\ell}(a)|\varepsilon'\cdot\varepsilon\rangle.
\end{align*}
Using the diagonal projection matrix asymptotics in Theorem~\ref{thm:walsh-proj}, and that $|\langle\varepsilon'\cdot\varepsilon|\operatorname{Op}_{k,\ell}(a)|\varepsilon'\cdot\varepsilon\rangle|\le \|a\|_\infty$, shows that for $D\ge3$ this is $\frac{1}{4k}\operatorname{tr}\operatorname{Op}_{k,\ell}(a)+o\big(\frac{D^k}{4k}\|a\|_{\infty}\big)$. 
Multiplying by $\frac{4k}{D^k}$ and using $\operatorname{tr}\operatorname{Op}_{k,\ell}(a)=D^k\int_{\T^2}a(\x)\,d\x$ then gives the result.
\end{proof}

\subsection{Time evolution of the Walsh baker quantization}\label{subsec:walsh-proof}

In order to prove Theorem~\ref{thm:walsh-proj}, we need the following results on entries of matrix powers of $B_k^\Wa$. We note that the case $D=2$ in the position basis is very similar to the doubling map quantization studied in \cite[\S8.2]{pw}, and can be covered by the same analysis. In that case, understanding the matrix powers in the position basis could be done simply by analyzing the structure of separate regions of the matrix under matrix multiplication. However, in a general $(k,\ell)$-coherent state basis, the matrix structure is much more ``scrambled''.
For $D\ge3$, the structure is also further scrambled by a dit flip/reflection-like map for powers between $2k$ and $4k$.
As a result, we instead rely on the action \eqref{eqn:walsh-action} of $B_k^\Wa$ on tensor product states, and count nonzero matrix entries by counting solutions to a resulting set of equations.
An example of the early time evolution of $B_k^\Wa$ for $D=3$ and $k=4$ in a $(k,\ell=2)$-coherent state basis is shown in Fig.~\ref{fig:walsh-evolution}.

\begin{figure}[htb]
\includegraphics[width=6in]{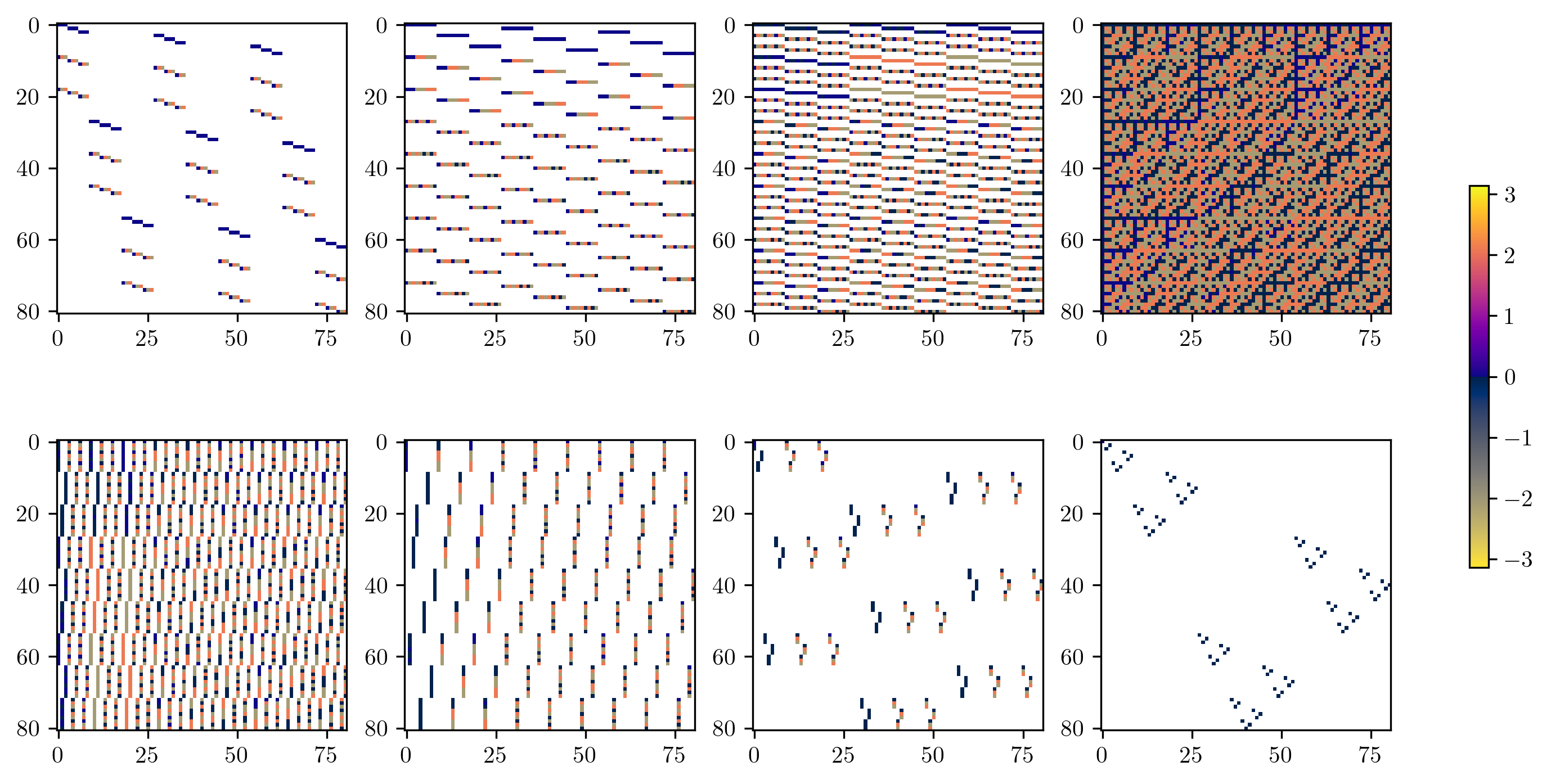}
\caption{Matrix powers $(B_k^\Wa)^t$ for $t\in\intbrr{1:2k}$ with $D=3$, $k=4$, $N=D^k=81$, in a $(k,\ell=2)$-coherent state basis. The nonzero matrix entries of $(B_k^\Wa)^t$ are shown in color according to their phase angle, while zero entries are shown in white.}\label{fig:walsh-evolution}
\end{figure}

\begin{prop}\label{prop:walsh-powers}
For $n\in\N$ and $j\in\Z$, let $[j]_{n}=j\;\mathrm{mod}\;n$ and $[j]_n\in\intbrr{0:n-1}$.  
Define
\[
\eta_k(j):=\begin{cases}
[j]_{2k},& [j]_{4k}\in\intbrr{0:k}\cup\intbrr{2k:3k}\\
2k-[j]_{2k},& [j]_{4k}\in\intbrr{k:2k-1}\cup\intbrr{3k:4k-1}
\end{cases},
\]
which is a triangle-shaped function, increasing and decreasing between $0$ and $2k$.
For $D\ge2$, $j\in\N$, and any $0\le \ell\le k$, then
\begin{enumerate}[(i)]
\item For $[j]_{2k}\ne 0$, there are exactly $D^{\eta_k(j)}$ $(k,\ell)$-coherent state basis vectors $|\varepsilon'\cdot\varepsilon\rangle$ such that $\langle \varepsilon'\cdot\varepsilon|(B_k^\Wa)^j|\varepsilon'\cdot\varepsilon\rangle\ne0$. For $[j]_{4k}=2k$, there is $D^{\eta_k(j)}=1$ such $|\varepsilon'\cdot\varepsilon\rangle$ if $D$ is odd, and $2^k$ such $|\varepsilon'\cdot\varepsilon\rangle$ if $D$ is even. For $[j]_{4k}=0$, there are $D^k$ solutions.
\item If  $\langle \varepsilon'\cdot\varepsilon|(B_k^\Wa)^j|\varepsilon'\cdot\varepsilon\rangle\ne0$, then it has absolute value $|\langle \varepsilon'\cdot\varepsilon|(B_k^\Wa)^j|\varepsilon'\cdot\varepsilon\rangle|=D^{-\eta_k(j)/2}$.
\item There are $D^k\cdot D^{\eta_k(j)}$ non-zero entries $\langle \delta'\cdot\delta|(B_k^\Wa)^j|\varepsilon'\cdot\varepsilon\rangle$, and for these entries, $|\langle\delta'\cdot\delta|(B_k^\Wa)^j|\varepsilon'\cdot\varepsilon\rangle|=D^{-\eta_k(j)/2}$.
\item For $[j]_{2k}\ne 0$, there are $D^{\eta_k(j)}$ position basis vectors $|x\rangle$ such that $\langle x+1|(B_k^\Wa)^j|x\rangle\ne0$, and also $D^{\eta_k(j)}$ such that $\langle x-1|(B_k^\Wa)^j|x\rangle\ne0$, where $x\pm1$ is taken modulo $D^k$. For $[j]_{4k}=2k$, there are no solutions if $D$ is even, and one solution to each if $D$ is odd. For $[j]_{4k}=0$, there are no solutions to either equation. 
\end{enumerate}
\end{prop}
\begin{proof}[Proof of Proposition~\ref{prop:walsh-powers}]
For $D\ge3$, since $(B_k^\Wa)^{4k}=\operatorname{Id}_{D^k}$ and $(B_k^\Wa)^{2k+j}=((B_k^\Wa)^\dagger)^{2k-j}$ for $j\in\intbrr{0:2k}$ by unitarity, it suffices to prove the statements for $1\le j\le k$ and $2k\le j\le 3k$. 
For $D=2$, since $(B_k^\Wa)^{2k}=\operatorname{Id}_{2^k}$, it suffices to prove the desired properties just for $1\le j\le k$.

\vspace{1.5mm}
\noindent (i), (ii).
Let $\varepsilon=\varepsilon_1\ldots\varepsilon_\ell$ and $\varepsilon'=\varepsilon_{\ell+1}\ldots\varepsilon_{k}$, so
$
|\varepsilon'\cdot\varepsilon\rangle = |{\varepsilon_1}\rangle\otimes\cdots\otimes |{\varepsilon_\ell}\rangle\otimes\hat{F}_D^\dagger |{\varepsilon_k}\rangle\otimes\cdots\otimes\hat{F}_D^\dagger |{\varepsilon_{\ell+1}}\rangle.
$
First consider $1\le j\le \ell$; then using \eqref{eqn:walsh-action},
\begin{align*}
(B_k^\Wa)^j|\varepsilon'\cdot\varepsilon\rangle &= |{\varepsilon_{j+1}}\rangle
\otimes\cdots\otimes |{\varepsilon_\ell}\rangle\otimes \hat{F}_D^\dagger |{\varepsilon_{k}}\rangle\otimes \cdots\otimes\hat{F}_D^\dagger |{\varepsilon_{\ell+1}}\rangle\otimes\hat{F}_D^\dagger |{\varepsilon_1}\rangle\otimes\cdots\otimes \hat{F}_D^\dagger |{\varepsilon_j}\rangle.
\end{align*}
We dot this with $|\varepsilon'\cdot\varepsilon\rangle$ and consider solutions to $\langle \varepsilon'\cdot\varepsilon|(B_k^\Wa)^j|\varepsilon'\cdot\varepsilon\rangle\ne0$. 
Taking the dot product is clearer by writing in the following table format, where tensor product indices are written in the top row:
\begin{align*}
\def\arraystretch{1.3}
\begin{array}{c||ccc|ccc|ccc}
 & 1&\cdots & \ell-j &\ell-j+1 &\cdots & \ell&\ell+1 & &k\\\hline
|\varepsilon'\cdot\varepsilon\rangle&|{\varepsilon_1}\rangle & \cdots &|{\varepsilon_{\ell-j}}\rangle &|{\varepsilon_{\ell-j+1}}\rangle&\cdots& |{\varepsilon_{\ell}}\rangle&\hat{F}_D^\dagger |{\varepsilon_k}\rangle&\cdots &\hat{F}_D^\dagger |{\varepsilon_{\ell+1}}\rangle\\\hline
(B_k^\Wa)^j|\varepsilon'\cdot\varepsilon\rangle&|{\varepsilon_{j+1}}\rangle&\cdots  &|{\varepsilon_\ell}\rangle & \hat{F}_D^\dagger |{\varepsilon_k}\rangle &\cdots&\hat{F}_D^\dagger |{\varepsilon_{k-j+1}}\rangle &\hat{F}_D^\dagger |{\varepsilon_{k-j}}\rangle&\cdots& \hat{F}_D^\dagger |{\varepsilon_j}\rangle
\end{array}
\end{align*}
In the first section of length $\ell-j$, obtaining a nonzero dot product generates the $\ell-j$ restrictions $\varepsilon_1=\varepsilon_{j+1},\ldots,\varepsilon_{\ell-j}=\varepsilon_\ell$. Similarly, the third section of length $k-\ell$ with the DFT matrices  gives rise to $k-\ell$ similar restrictions. The middle section is of length $j$ and always produces nonzero dot products, regardless of the values of $\varepsilon,\varepsilon'$; the contribution from this section has absolute value
\begin{align}\label{eqn:wabs}
\left|\langle {\varepsilon_{\ell-j+1}}|\hat{F}_D^\dagger|{\varepsilon_k}\rangle\cdots\langle {\varepsilon_\ell}|\hat{F}_D^\dagger| {\varepsilon_{k-j+1}}\rangle\right|&=D^{-j/2}.
\end{align}
There are in total then $(\ell-j)+(k-\ell)=k-j$ restrictions for the variables $\varepsilon_1,\ldots,\varepsilon_k$.
For $j<k$, by Lemma~\ref{lem:la} below, with $L(x)=x$ and $A\subset\Z/k\Z$ the interval $\intbrr{\ell+1:k}\cup\intbrr{1:\ell-j}$, 
there are thus $D^j$ total solutions to $\langle \varepsilon'\cdot\varepsilon|(B_k^\Wa)^j|\varepsilon'\cdot\varepsilon\rangle\ne0$.

The consideration for $\ell< j\le k$ is similar. We write the table in abbreviated format, writing only whether the term is $I_D$, $\hat{F}_D^\dagger$, or $(\hat{F}_D^\dagger)^2=R_D$ where $R_D:|x\rangle\mapsto |-x\;\mathrm{mod}\;D\rangle$.  
The basis elements on which those operators act are not written, but we know the basis elements in the bottom row are always permuted cyclically by $j$ according to the action \eqref{eqn:walsh-action}.
\begin{align*}
\def\arraystretch{1.3}
\begin{array}{c||ccc|ccc|ccc}
 & 1&\cdots & \ell & \ell+1 &\cdots & k-j+\ell&k-j+\ell+1 & &k\\\hline
|\varepsilon'\cdot\varepsilon\rangle& I_D&\cdots& I_D&\hat{F}_D^\dagger &\cdots &\hat{F}_D^\dagger &\hat{F}_D^\dagger&\cdots&\hat{F}_D^\dagger \\\hline
(B_k^\Wa)^j|\varepsilon'\cdot\varepsilon\rangle& \hat{F}_D^\dagger &\cdots&\hat{F}_D^\dagger&\hat{F}_D^\dagger& \cdots &\hat{F}_D^\dagger &R_D&\cdots &R_D
\end{array}.
\end{align*}
The free variables appear when there is a matching of $\hat{F}_D^\dagger$ with $I_D$ or $R_D$, since the inner product will always have absolute value $D^{-1/2}$, and this occurs in the first $\ell$ entries and last $j-\ell$ entries. Thus the only restrictions are the $k-j$ from the middle section. If $k-j\ge1$ then applying Lemma~\ref{lem:la} below with $L(x)=x$ and $A=\intbrr{\ell+1:k-j+\ell}$ shows there are again $D^j$ total solutions $|\varepsilon'\cdot\varepsilon\rangle$. If $j=k$, then the middle section of the table does not exist, and the inner product is nonzero for any of the $D^j=D^k$ states $|\varepsilon'\cdot\varepsilon\rangle$. 

For $D\ge3$ and $2k\le j\le 3k$, we use a similar argument combined with the equation $(B_k^\Wa)^{2k+i}=(B_k^\Wa)^{2k}(B_k^\Wa)^i= R_D^{\otimes k} (B_k^\Wa)^i$, where we recall $R_D=(\hat{F}_D^\dagger)^2$ is the map $R_D:|x\rangle\mapsto |{-x\;\mathrm{mod}\;D}\rangle$.
The reference vector $|\varepsilon'\cdot \varepsilon\rangle$ is replaced by $R_D^{\otimes k}|\varepsilon'\cdot\varepsilon\rangle$ in the tables, which now read for $0\le i\le\ell$, 
\begin{align*}
\def\arraystretch{1.3}
\begin{array}{c||ccc|ccc|ccc}
 & 1&\cdots & \ell-i &\ell-i+1 &\cdots & \ell&\ell+1 & &k\\\hline
R_D^{\otimes k}|\varepsilon'\cdot\varepsilon\rangle& R_D & \cdots &R_D&R_D&\cdots& R_D&R_D\hat{F}_D^\dagger &\cdots &R_D\hat{F}_D^\dagger \\\hline
(B_k^\Wa)^i|\varepsilon'\cdot\varepsilon\rangle&I_D&\cdots  &I_D & \hat{F}_D^\dagger &\cdots&\hat{F}_D^\dagger &\hat{F}_D^\dagger &\cdots& \hat{F}_D^\dagger
\end{array},
\end{align*}
and for $\ell\le i\le k$,
\begin{align*}
\def\arraystretch{1.3}
\begin{array}{c||ccc|ccc|ccc}
& 1&\cdots & \ell & \ell+1 &\cdots & k-i+\ell&k-i+\ell+1 & &k\\\hline
R_D^{\otimes k}|\varepsilon'\cdot\varepsilon\rangle& R_D&\cdots& R_D&R_D\hat{F}_D^\dagger &\cdots &R_D\hat{F}_D^\dagger &R_D\hat{F}_D^\dagger&\cdots&R_D\hat{F}_D^\dagger \\\hline
(B_k^\Wa)^i|\varepsilon'\cdot\varepsilon\rangle& \hat{F}_D^\dagger &\cdots&\hat{F}_D^\dagger&\hat{F}_D^\dagger& \cdots &\hat{F}_D^\dagger &R_D&\cdots &R_D
\end{array}.
\end{align*}
For $0\le i\le \ell$, since $\hat{F}_DR_D\hat{F}_D^\dagger=R_D$ as $\hat{F}_D$ and $R_D$ commute, then since we only consider the inner product, we can replace the third section with $R_D$ in the top row and $I_D$ in the bottom row. This makes the first row all $R_D$, and each variable $\varepsilon_i$ shows up exactly once in the row.
If $i=0$, the system to solve is $-\varepsilon_m\;\mathrm{mod}\;D=\varepsilon_m$ for all $m\in\intbrr{1:k}$. For $D\ge3$ odd, this has one solution where all $\varepsilon_m=0$, while for $D\ge3$ even this has $2^k$ solutions where each $\varepsilon_m\in\{0,D/2\}$.

If $1\le i\le\ell$, then Lemma~\ref{lem:la} with $L(x)=-x$ and $A=\intbrr{\ell+1:k}\cup\intbrr{1:\ell-i}$ shows there are $D^i=D^{j-2k}=D^{\eta_k(j)}$ solutions to the system described by the table.

For $\ell< i\le k$, note that $R_D\hat{F}_D^\dagger=\hat{F}_D^\dagger R_D$, so the middle section in the last table above can be replaced with $R_D$ on the top row and $I_D$ on the bottom row. For $\ell \le i\le k-1$, Lemma~\ref{lem:la} with $L(x)=-x$ and $A=\intbrr{\ell+1:k-i+\ell}$ shows there are $D^i=D^{j-2k}=D^{\eta_k(j)}$ solutions. For $i=k$, the middle section of the table does not exist, and using  $R_D^2=I_D$ for the third section, shows the inner product is nonzero for any $|\varepsilon'\cdot\varepsilon\rangle$.
In all the cases, the same absolute value equality \eqref{eqn:wabs} holds in the cases the inner product is nonzero.

\vspace{1.5mm}
\noindent (iii). To count the nonzero off-diagonal entries, we use the same tables constructed above.
For each of the $D^k$ possible $|\varepsilon'\cdot\varepsilon\rangle$, which we can think of as fixed constants, we solve the linear system for the $k$ symbol coordinates of $|\delta'\cdot\delta\rangle$. Using the same argument as for the diagonal entries, we can count the number of restrictions, and similarly as in Lemma~\ref{lem:la} (but simpler since the $v_{[i+s]_D}$ variables are now constants and we do not need to consider the case $[j]_{2k}=0$ separately), there are $D^{\eta_k(j)}$ solutions $|\delta'\cdot\delta\rangle$. Thus in total there are $D^k\cdot D^{\eta_k(j)}$ solutions of $|\delta\cdot\delta\rangle,|\varepsilon'\cdot\varepsilon\rangle$. Moreover, for these nonzero matrix elements,  $|\langle\delta'\cdot\delta|(B_k^\Wa)^j|\varepsilon'\cdot\varepsilon\rangle|=D^{-\eta_k(j)/2}$ just as for the diagonal elements.

\vspace{1.5mm}
\noindent (iv). Since $\ell=k$ for the position basis, to count the number of $x$ such that $(B_k^\Wa)^j_{x\pm1,x}$ is nonzero, we can use just the two tables corresponding to $1\le j\le \ell=k$ and $0\le i\le \ell=k$ for $i+2k=j$. Since $\ell=k$, the tables have only two sections, and the interval $A$ containing the restrictions is always the first section, $\intbrr{1:k-j}$ or $\intbrr{1:k-i}$.

If $|x\rangle=|\varepsilon_1\rangle\otimes\cdots\otimes|\varepsilon_k\rangle$ with $x=\sum_{m=1}^k\varepsilon_mD^{k-m}\in\intbrr{0:D^k-1}$, and $|x\pm1\rangle=|\tilde{\varepsilon}_1\rangle\otimes\cdots\otimes|\tilde{\varepsilon}_k\rangle$, then $\tilde{\varepsilon}_m\in\{\varepsilon_m,\varepsilon_m\pm1\;\mathrm{mod}\;D\}$, depending on if we have to ``carry'' (or ``borrow'') ones when adding (subtracting). We can determine which value if we know the later place values $\varepsilon_{m+1},\ldots,\varepsilon_k$.
Let $f(\varepsilon_1,\ldots,\varepsilon_k):=(\tilde{\varepsilon}_1,\ldots,\tilde{\varepsilon}_k)$.
To solve $\langle x\pm1|(B_k^\Wa)^j|x\rangle\ne0$, we are solving the system, where $s=j$ (or $i$ if $2k\le j\le 3k$), and $\alpha\in\{\pm1\}$,
\begin{align*}
[\alpha f(\varepsilon)_m]_D-\varepsilon_{m+s}=0,\quad m\in A=\intbrr{1:k-s}.
\end{align*}
For $s\ge1$, the last $s$ variables, $\varepsilon_{k},\varepsilon_{k-1},\ldots,\varepsilon_{k-s+1}$, do not appear as the index $m$ in $f(\varepsilon)_m$, so can be taken to be the free variables. For any given values in $\intbrr{0:D-1}$ of these free variables, we can solve for the unique solution for the remaining variables by starting with the last equation $m=k-s$ and working upwards. In this way, we always know whether $f(\varepsilon)_m=\varepsilon_m$ or $\varepsilon_m\pm 1$, since it only depends on the values of $\varepsilon_n$ for $n>m$, which are either free or already solved. The map $x\mapsto [\alpha x]_D$ for $\alpha=\pm1$ is a bijection (permutation) on $\intbrr{0:D-1}$, so there is a unique solution $\varepsilon_m$. The $s$ free variables then generate $D^s$ solutions $|x\rangle$.

For $[j]_{4k}=2k$, corresponding to $s=i=0$ and the second to last table, we must solve the system $[-f(\varepsilon)_m]_D-\varepsilon_m=0$ for $m\in\intbrr{1:k}$. Starting with $m=k$, we always have $f(\varepsilon)_k=\varepsilon_k\pm1\;\mathrm{mod}\;D$ for the operation $x\pm1$,  and so we must have $\varepsilon_k\in \frac{1}{2}D\Z\mp\frac{1}{2}$, which has no solutions in $\Z$ if $D$ is even. If $D$ is odd, then there is a single solution in $\intbrr{0:D-1}$, $\varepsilon_k=\frac{D\mp 1}{2}$. For considering $x+1$, this is $\frac{D-1}{2}<D-1$ so there is no carrying of ones and $f(\varepsilon)_m=\varepsilon_m$ for all $m<k$. For $x-1$, $\varepsilon_k$ is $\frac{D+1}{2}>0$, so there is no borrowing of ones and $f(\varepsilon)_m=\varepsilon_m$ for all $m<k$. For $D$ odd and $m<k$, the equation $[-\varepsilon_m]_D-\varepsilon_m=0$ has only the solution $\varepsilon_m=0$, giving a total of one solution $|x\rangle=|0\rangle\otimes\ldots\otimes|0\rangle\otimes|\frac{D\mp1}{2}\rangle$.
\end{proof}

The following lemma was used in the proof of Proposition~\ref{prop:walsh-powers}.
\begin{lem}\label{lem:la}
For $\alpha\in\{\pm1\}$ and $b\in\Z$, let $L(x)=\alpha x+b$. Let $s\in\intbrr{1:k-1}$, and let $A\subset\Z/k\Z$ be an interval of length $k-s$ (which may wrap around past $k$). Consider the $(k-s)\times k$ system in variables $v_0,\ldots,v_{k-1}\in\intbrr{0:D-1}$,
\begin{align}\label{eqn:system}
[L(v_m)]_D - v_{[m+s]_k} = 0,\quad m\in A,
\end{align}
where $[y]_n\in\intbrr{0:n-1}$ denotes the representative of $y\;\mathrm{mod}\;n$.
Then there are $D^s$ solutions $v\in\intbrr{0:D-1}^k$ to \eqref{eqn:system}.
\end{lem}
\begin{proof}[Proof of Lemma~\ref{lem:la}]
Due to the cyclic symmetry of $[m+s]_k$, we may relabel the variables as $x_m=v_{[m+a]_k}$, where $a$ is the first entry of $A$, and take $A=\intbrr{0:k-s-1}$. Because $\alpha=\pm1$ and $b$ is just a shift, the map $x\mapsto [\alpha x+b]_D$ is a bijection (permutation) on $\intbrr{0:D-1}$. Thus solving the system from the bottom, starting with $m=k-s-1$, we see there are $s$ free variables $x_{k-1},x_{k-2},\ldots,x_{k-s}$, whose indices do not appear in $A$, and the rest of the $x_m$'s are determined by those. 
Then there are $D^s$ solutions by taking each free variable in $\intbrr{0:D-1}$.
\end{proof}

\subsection{Proof of Theorem~\ref{thm:walsh-proj}}\label{subsec:proof-thm8.2}
Going from Proposition~\ref{prop:walsh-powers} to Theorem~\ref{thm:walsh-proj} is an application of \cite[\S8]{pw}. We write the outline here for $D\ge3$. For $D=2$, one replaces instances of $4k$ with $2k$.
Instead of the Beurling--Selberg approximation or a Fourier series approximation like for Proposition~\ref{prop:mp-sf}, one takes the polynomial
\begin{align}
p_{k,\jj}(z)=1+\sum_{m=1}^{4k-1}\big(e^{-2\pi i\jj/(4k)}\big)^m z^m.
\end{align}
Since
$\frac{z^{4k}-1}{z-1}=1+z+z^2+\cdots+z^{4k-1}$ is zero at all $4k$-th roots of unity except for $z=1$, then $\frac{1}{4k}p_{k,\jj}(B_k^\Wa)$ is exactly the spectral projection onto the eigenspace of $e^{2\pi i\jj/(4k)}$. The matrix entry estimates in Proposition~\ref{prop:walsh-powers} for powers $1,\ldots,4k-1$, with the cut-off argument  in \cite[\S8.2--8.4/Fig.~8]{pw} to be able to ignore powers near $m=1,2k,4k$ (where the nonzero matrix entries are large), then prove Theorem~\ref{thm:walsh-proj}. 
To give a rough outline, one picks a cut-off $r(k)$, say $r(k)=k/2$, and then defines the set of good coordinates $G_{k,\ell}$ for the diagonal estimates as 
\begin{multline*}
G_{k,\ell}=\big\{|\varepsilon'\cdot\varepsilon\rangle:\langle\varepsilon'\cdot\varepsilon|(B_k^\Wa)^m|\varepsilon'\cdot\varepsilon\rangle=0\\\text{ for }m\in\intbrr{1:r(k)}\cup\intbrr{2k-r(k):2k+r(k)}\cup\intbrr{4k-r(k):4k-1} \big\}.
\end{multline*}
By Proposition~\ref{prop:walsh-powers}(i), for $D$ odd, this only excludes $\mathcal{O}(D^{r(k)})=o(D^k/(4k))$ coordinates. For $D$ even and $D\ge4$, this excludes $\mathcal{O}(D^{r(k)})+2^k$ coordinates,  with the extra term from the contribution of solutions when $m=2k$. However this quantity is still $o(D^k/(4k))$ for $D>2$. (When $D=2$, one stops the polynomial $p_{k,\jj}$ at $m=2k-1$ and there is no need to consider $m=2k$.)
For coherent states $|\varepsilon'\cdot\varepsilon\rangle\in G_{k,\ell}$, then
\begin{multline*}
\langle\varepsilon'\cdot\varepsilon|P_\jj|\varepsilon'\cdot\varepsilon\rangle = \frac{1}{4k}\Bigg(1+\sum_{m=r(k)+1}^{2k-r(k)-1}\big(e^{-2\pi i\jj/(4k)}\big)^m\langle\varepsilon'\cdot\varepsilon|(B_k^\Wa)^m|\varepsilon'\cdot\varepsilon\rangle +\\
+ \sum_{m=2k+r(k)+1}^{4k-r(k)-1}\big(e^{-2\pi i\jj/(4k)}\big)^m\langle\varepsilon'\cdot\varepsilon|(B_k^\Wa)^m|\varepsilon'\cdot\varepsilon\rangle\Bigg),
\end{multline*}
but the terms involving $\langle\varepsilon'\cdot\varepsilon|(B_k^\Wa)^j|\varepsilon'\cdot\varepsilon\rangle$ are small and in total only contribute $\mathcal{O}(D^{-r(k)/2})$ due to Proposition~\ref{prop:walsh-powers}(ii). This gives \eqref{eqn:wdiag} with $\mathcal{O}(D^{-r(k)/2})$ as a more precise $o(1)$ error term. The off-diagonal estimates \eqref{eqn:woffdiag} and those for coordinates $(x\pm1,x)$ are similar using the rest of Proposition~\ref{prop:walsh-powers}. For similar details see also \cite[\S8.4]{pw}.

\subsection{Proof of Theorem~\ref{thm:walsh-main}}\label{sec:awalsh}

\subsubsection{Part (i)} This follows from Theorem~\ref{thm:walsh-proj} using the same methods as in \cite[Theorem 2.5(b)/\S6.2(b)]{pw}, using the quantitative version of Theorem~\ref{thm:conv} which was developed in \cite{Meckes,ChatterjeeMeckes}. 
Let $(\psi^{(k,m)})_{m=1}^{D^k}\in\Omega_k$ denote an orthonormal eigenbasis for $B_k^\Wa$. 
The main part is to establish the following quantitative inequality, for some choice of $\varepsilon_k\to0$,
\begin{align}\label{eqn:gauss-prob}
\P\left[\max_{m\in\intbrr{1:D^k}}\left|\int f(x)\,d\mu^{\psi^{(k,m)}}_{k,\ell(k)}(x)-\E f(Z)\right|>\varepsilon_k\|f\|_{\operatorname{Lip}}\right]\le CD^k\exp\left[-c\left(\frac{D^k}{4k}\right)^{1/2}(1-o(1))\right],
\end{align}
where $f:\C\to\C$ is any bounded Lipschitz function with Lipschitz constant $\|f\|_{\mathrm{Lip}}:=\sup_{x\ne y}\frac{|f(x)-f(y)|}{|x-y|}<\infty$, and $Z\sim N_\C(0,1)$.

The quantitative estimates from \cite{Meckes,ChatterjeeMeckes} lead to
\begin{thm}[see {\cite[Corollary 6.2]{pw}}]
Let $\C^n=V^{[1]}\oplus\cdots\oplus V^{[\kappa]}$, and let $P^{[m]}$ be the orthogonal projection onto the subspace $V^{[m]}$. Suppose there is $A$ and $d_1,\ldots,d_\kappa\in\R_+$ so that
\begin{equation}\label{eqn:quant-cond}
	\sum_{x=1}^n\left|\|P^{[m]}e_x\|_2^2-\frac{d_m}{n}\right|\le A,\quad \forall m\in\intbrr{\kappa}.
\end{equation}
Choose a random orthonormal basis $(\phi^{[j]})_{j=1}^n$ for $\C^n$ by choosing a random orthonormal basis from each $V^{[m]}$ (according to Haar measure) and embedding it by inclusion in $\C^n$. For each $j\in\intbrr{n}$, let 
\[
\mu^{[j]}:=\frac{1}{n}\sum_{x=1}^n\delta_{\sqrt{n}\phi^{[j]}_x},
\]
the empirical distribution for the $j$th basis vector's coordinates. There are absolute numerical constants $C,c>0$ so that for any $f:\C\to\C$ bounded $L$-Lipschitz and $\varepsilon>\frac{2L(2A+3)}{(\min d_m-A)-1}$,
\begin{equation}
	\mathbb{P}\left[\max_{j\in\intbrr{1:n}}\left|\int f(x)\,d\mu^{[j]}(x)-\mathbb{E}f( Z)\right|>\varepsilon\right]\le Cn\exp\left(-\frac{c\varepsilon^2(\min d_m-A)}{L^2}\right),
\end{equation}
where $Z\sim N_\C(0,1)$.
\end{thm}
Applying this with $n=D^k$ and each $V^{[j]}$ denoting an eigenspace of $B_k^\Wa$, the diagonal estimates in Theorem~\ref{thm:walsh-proj} allow us to take all $d_m=D^k/(4k)$, and $A=o(D^k/(4k))$ in \eqref{eqn:quant-cond}. 
Setting
\[
\varepsilon_k=\max\left(\frac{4A+6}{(d_m-A)-1},\frac{1}{(D^k/4k)^{1/2}}\right)\to0
\]
yields \eqref{eqn:gauss-prob}.

Finally, let $(f_j)_{j=1}^\infty$ be a countable set of Lipschitz functions with compact support that is dense in $C_c(\C)$ in the $\|\cdot\|_\infty$ norm, and take
\begin{align}
\Pi_k=\left\{(\psi^{(k,m)})_{m=1}^{D^k}\in\Omega_k:\forall m,j\in\intbrr{1:D^k},\left|\int_\C f_j\,d\mu^{\psi^{(k,m)}}_{k,\ell(k)}-\E f_j(Z)\right|\le \varepsilon_k\|f_j\|_\mathrm{Lip}\right\}.
\end{align}
Applying \eqref{eqn:gauss-prob} with a union bound over $j\in\intbrr{1:D^k}$ shows that
\begin{align}
\P[\Pi_k^c] &\le C (D^k)^2\exp\left[-c\left(\frac{D^k}{4k}\right)^{1/2}(1-o(1))\right]\to0.
\end{align}
By the definition of $\Pi_k$, for any $f_j$, as $k\to\infty$,
\begin{align}\label{eqn:eq}
\int_\C f_j\,d\mu_{k,\ell(k)}^{\psi^{(k,m_k)}}\to\E f_j(Z),
\end{align}
for any sequence $(m_k)_k$, $m_k\in\intbrr{1:D^k}$.
Denseness of the $(f_j)_j$ implies \eqref{eqn:eq} also holds for compactly supported functions $f\in C_c(\C)$, and this vague convergence with tightness implies the desired weak convergence.

\subsubsection{Part (ii)} This follows from Theorems~\ref{thm:walsh-proj} (matrix powers) and \ref{thm:walsh-lweyl} (generalized Weyl law) using similar methods as in Theorem~\ref{thm:rw}(ii) or \cite[\S6.2(c)]{pw}, using the Hanson--Wright inequality (cf. \cite[Theorem 4.1]{ChatterjeeGalkowski}) and the quantitative convergence as in part (i) above.
First one establishes the inequality for some $\varepsilon_k\to0$,
\begin{align}\label{eqn:equi-prob}
\P\left[\max_{m\in\intbrr{1:D^k}}\left|\langle \psi^{(k,m)}|\operatorname{Op}_{k,\ell}(a)|\psi^{(k,m)}\rangle-\int_{\T^2}a(\x)\,d\x\right|>\varepsilon_k\|a\|_\infty\right]\le CD^{k}\exp\left[-c\left(\frac{D^k}{4k}\right)^{1/2}(1-o(1))\right],
\end{align}
using the method in the proof of Theorem~\ref{thm:rw}(ii) in Section~\ref{subsubsec:hw}, followed by a union bound (e.g. \cite[Lemma 6.3]{pw}). 

Then one can take a countable subset $(a_j)_j$ in $\operatorname{Lip}(\T^2)$ that is dense with respect to the $\|\cdot\|_\infty$ norm, and set
\begin{align}
\Gamma_k=\left\{(\psi^{(k,m)})_{m=1}^{D^k}\in\Omega_k:\forall m,j\in\intbrr{1:D^k},\left|\langle \psi^{(k,m)}|\operatorname{Op}_{k,\ell}(a_j)|\psi^{(k,m)}\rangle-\int_{\T^2}a_j(\x)\,d\x\right|\le \varepsilon_k\|a_j\|_\infty\right\},
\end{align}
so that
\begin{align}
\P[\Gamma_k^c] &\le C(D^k)^2\exp\left[-c\left(\frac{D^k}{4k}\right)^{1/2}(1-o(1))\right]\to0.
\end{align}
Denseness of $(a_j)_j$ combined with the bound $\|\operatorname{Op}_{k,\ell}(a-a_j)\|\le \|a-a_j\|_\infty$ shows that equidistribution for all $a\in\operatorname{Lip}(\T^2)$ holds as well.

\subsubsection{Parts (iii), (iv), (v)} These are essentially the same as in the proof of Theorem~\ref{thm:rw} given in Section~\ref{subsec:rw-proof}. We need the relevant projection matrix estimates, which are provided by Theorem~\ref{thm:walsh-proj} and the more specific estimate Eq.~\eqref{eqn:walsh-gset}. For (iii) and (iv), since we work with random {unit} vectors instead of the random Gaussian vector \eqref{eqn:rw}, we also need to use concentration of the norm $\|g\|_2$ near $\sqrt{n}$ for $g\sim N(0,I_n)$; see e.g. \cite[Theorem 3.1.1]{Vershynin}. This allows one to essentially replace a unit vector $u\sim\frac{g}{\|g\|_2}$ with the Gaussian vector $\frac{g}{\sqrt{n}}$ for large $n$. Part (v) involves sign changes so is unchanged whether one normalizes by $\|g\|_2$ or $\sqrt{\operatorname{dim}E_{(\jj)}}$.

\appendix

\section{Windowed quantum ergodicity}\label{sec:wqe}
In this section we provide the proof of  Theorem~\ref{thm:qv} on the windowed quantum variance decay. In \cite{DNW}, the (non-windowed) quantum variance decay was proved using exponential decay of correlations and trace properties. For ease of applying the windowed version of the generalized Weyl law, we follow the more usual proof method described for example in \cite{zworski,MOK}, using just ergodicity. In this method, one replaces use of the generalized Weyl law with its windowed version Theorem~\ref{thm:lweyl}.
Corollary~\ref{cor:qe} then follows from Theorem~\ref{thm:qv} by the standard Chebyshev--Markov and density argument, e.g. see \cite[Ch.15]{zworski} or \cite{MOK,DNW}. 

First, we collect several results we will need for the proof.
In what follows we drop the $N$ superscript on the eigenvalues $\theta^{(j,N)}$ and (orthonormal) eigenvectors $\varphi^{(j,N)}$.

\subsection{Preliminaries}
There is an explicit bound on the rate of convergence for $C^1$ functions for the baker map (using the equivalence with a two-sided Bernoulli shift) in the $L^2$ ergodic theorem, which follows from the following exponential decay of correlations (also used in \cite{DNW}),
\[
\left|\int_{\T^2}a(\x)b(B^{-n}(\x))\,d\x-\int_{\T^2}a(\x)\,d\x\int_{\T^2}b(\x)\,d\x\right| \le C\|a\|_{C^1}\|b\|_{C^1}e^{-\Gamma|n|}.
\]
By considering observables $f$ with $\int_{\T^2}f(\mathbf{y})\,d\mathbf{y}=0$ and changing variables under $B^r$ which is measure-preserving, one obtains,
\begin{lem}[ergodicity rate]\label{lem:erg-rate}
There is a constant $C>0$ so that for the classical baker map $B:\T^2\to\T^2$ and any smooth observable $f\in C^\infty(\T^2)$,
\begin{align}
\int_{\T^2}\left|\frac{1}{T}\sum_{t=0}^{T-1}f\circ B^{-t}(\x)-\int_{\T^2}f(\mathbf{y})\,d\mathbf{y}\right|^2\,d\x &\le \frac{C\|f\|_{C^1}^2}{T}.
\end{align}
\end{lem}
An explicit decay rate is not needed to prove quantum ergodicity, but it does allow one to obtain quantitative error estimates.

Next, we need the classical-quantum correspondence principle, which relates quantum evolution by $\hat{B}_N$ to classical evolution by $B$ for short times.
\begin{thm}[Egorov theorem, {\cite[Eq. (5.27)]{DNW}}]\label{thm:egorov}
Let $\delta\in(0,2^{-t-1}),\gamma\in(0,1/2)$. Then if $a=a_N\in C^\infty(\T^2)$ is supported in $\mathcal{G}_{t,\delta,\gamma,N}$, as $N\to\infty$,
\begin{multline}
\left\|\hat{B}_N^t\operatorname{Op}_N^\W(a)\hat{B}_N^{-t}-\operatorname{Op}_N^\W(a\circ B^{-t})\right\|\lesssim \|a\|_{C^0}N^{5/4}2^{t/4}e^{-\pi N\min(2^t\delta^2,\frac{\gamma^2}{2^t})}+\\
+\frac{2^t}{N}(\|a\|_{C^5}+\|a\circ B^{-t}\|_{C^5}),
\end{multline}
where the implied constant in the notation $\lesssim$ is uniform in $t,\delta,\gamma$.
\end{thm}

We define cut-off functions as in \cite[\S5.3.1]{DNW}. For $0<\beta<1/4$, let $\tilde{\Chi}_\beta\in C^\infty(\R/\Z)$ which is zero in $[-\beta,\beta]\mod \Z$ and is 1 for $q\in[2\beta,1-2\beta]\mod \Z$. Specifically, we can start with a smooth bump function $\eta$ supported in $[-1/2,1/2]$, define $\eta_{\beta}(q):=\beta^{-1}\eta(\beta^{-1}q)$, and then take $\tilde{\Chi}_\beta:=\oneb_{[3\beta/2,1-3\beta/2]}*\eta_\beta$. The derivatives of $\tilde{\Chi}_\beta$ then satisfy $\|\partial^j\tilde{\Chi}_\beta\|_\infty\le \beta^{-j}\|\partial^j\eta\|_{L^1}$.

Now define for any $n\ge0$, 
\begin{align*}
\Chi_{\beta,n}(\x)&:=\tilde{\Chi}_\beta(2^nq)\tilde{\Chi}_\beta(p),\qquad \text{for }\x=(q,p).
\end{align*}
For $a\in C^\infty(\T^2)$, split $a$ into a ``good part'' and a ``bad part'',
\begin{align*}
a_n(\x)&:=a(\x)\Chi_{\delta,n}(\x),\qquad a_n^\mathrm{bad}(\x)=a(\x)-a_n(\x).
\end{align*}
The good part $a_n$ is supported on $\mathcal{G}_{n,\beta/2^n,\beta,N}$, while the bad part $a_n^\mathrm{bad}$ is supported on a region of area $\mathcal{O}(\beta)$. Additionally, $a_n$ behaves nicely through $t$ iterations of the classical map $B$, for $t\le n$.

Since the error terms in our estimates will involve quantities like $\|f\|_{C^m}$, the following norm estimates will be useful.
\begin{lem}[norm estimates, \cite{DNW}]\label{lem:norm}
For $a\in C^\infty(\T^2)$ and $j\in\N_0$, we have the following estimates, for $t\le n$,
\begin{align}\label{eqn:norm-bound}
\max(\|a_n\|_{C^j},\|a_n^\mathrm{bad}\|_{C^j},\|a_n\circ B^{-t}\|_{C^j}) &\le C_j\|a\|_{C^j}2^{nj}\beta^{-j}.
\end{align}
\end{lem}
\begin{proof}

By the Leibniz product formula, for $\gamma=(\gamma_1,\gamma_2)\in\N_0^2$ and $|\gamma|=\gamma_1+\gamma_2$,
\begin{align*}
\|\partial^\gamma a_n\|_\infty &\le C_{|\gamma|}\|a\|_{C^{|\gamma|}}2^{n\gamma_1}\beta^{-|\gamma|}.
\end{align*}

Also,
\begin{align*}
|\partial^\gamma(a_n\circ B^{-t})(q,p)| &=\left|\partial_q^{\gamma_1}\partial_p^{\gamma_2}\left[a_n(q/2^t-\xi(p)/2^t,2^tp-\lfloor 2^tp\rfloor)\right]\right| \\
&=2^{-t\gamma_1}2^{t\gamma_2}\left|(\partial^{\gamma_1}_q\partial_p^{\gamma_2}a_n)(q/2^t+\xi/2^t,2^tp-\lfloor 2^tp\rfloor)\right|\\
&\le 2^{-t\gamma_1}2^{t\gamma_2}\|\partial^\gamma a_n\|_\infty \le C_{|\gamma|}\|a\|_{C^{|\gamma|}}2^{(n-t)\gamma_1}2^{t\gamma_2}\beta^{-|\gamma|},
\end{align*}
where we were able to take derivatives of $a_n$ since the argument $(q/2^t+\xi/2^t,2^tp-\floor{2^tp})$ is only discontinuous at points outside the support of $a_n$.
For $t\le n$, then $(n-t)\gamma_1+t\gamma_2\le n(\gamma_1+\gamma_2)=n|\gamma|$, and so we obtain the upper bound in \eqref{eqn:norm-bound}.

The estimate with $a_n^\mathrm{bad}$ follows from that for $\|a_n\|_{C^j}$ since $a_n^\mathrm{bad}=a-a_n$.
\end{proof}

\subsection{Proof of Theorem~\ref{thm:qv} quantum variance decay}\label{subsec:qe-proof}
We may assume $\int_{\T^2} a(q,p)\,dq\,dp=0$. In what follows, the parameters $n,T,\beta$ will depend on $N$. Splitting $a=a_n+a_n^\mathrm{bad}$, then
\begin{multline}\label{eqn:qv-split}
\frac{2\pi}{N|I(N)|}\sum_{j\in I(N)}\big|\big\langle\varphi^{(j)}|\operatorname{Op}_N^\W(a)|\varphi^{(j)}\big\rangle\big|^2 \\
\le \frac{4\pi}{N|I(N)|}\sum_{j\in I(N)}\left(|\big\langle\varphi^{(j)}|\operatorname{Op}_N^\W(a_n)|\varphi^{(j)}\big\rangle\big|^2+\big|\big\langle\varphi^{(j)}|\operatorname{Op}_N^\W(a_n^\mathrm{bad})|\varphi^{(j)}\big\rangle\big|^2\right).
\end{multline}
Our goal is to show each of the terms in the sum is $o(1)$. The ``good part'' is the term involving $a_n$ which is supported away from the discontinuities, and the ``bad part'' is the term involving $a_n^\text{bad}$.

\subsubsection{Good part} 
This follows the standard quantum variance decay argument, though we need to pay more careful attention to the error bounds.
By the Egorov theorem (Theorem~\ref{thm:egorov}) with symbol $a_n$ supported in $\mathcal{G}_{n,\beta/2^n,\beta,N}$, for any $t\le n$ and eigenstate $\varphi^{(j)}$ of $\hat{B}_N$,
\begin{align*}
\big\langle\varphi^{(j)}|\operatorname{Op}_N^\W(a_n)|\varphi^{(j)}\big\rangle &=\langle \varphi^{(j)}|\hat{B}_N^t\operatorname{Op}_N^\W(a_n)\hat{B}_N^{-t}|\varphi^{(j)}\big\rangle \\
\numberthis\begin{split}\label{eqn:op-iter}
&=\big\langle\varphi^{(j)}|\operatorname{Op}_N^\W(a_n\circ B^{-t})|\varphi^{(j)}\big\rangle+\\
&\hspace{2cm}+\mathcal{O}\bigg(\|a_n\|_{C^0}N^{5/4}2^{t/4}e^{-\pi N\beta^2/2^t}+\frac{2^t(\|a_n\|_{C^5}+\|a_n\circ B^{-t}\|_{C^5})}{N}\bigg).
\end{split}
\end{align*}
By the norm estimates in Lemma~\ref{lem:norm},
\begin{align*}
\|a_n\|_{C^5} &\le C \|a\|_{C^5}2^{5n}\beta^{-5},\quad\text{and}\quad\|a_n\circ B^{-t}\|_{C^5} \le C\|a\|_{C^5}2^{5n}\beta^{-5}.
\end{align*}

The next step is to average over $t$, up to some time $T$. Letting $[a_n]_T:=\frac{1}{T}\sum_{t=0}^{T-1}a_n\circ B^{-t}$, then by \eqref{eqn:op-iter},
\begin{align}\label{eqn:7.5}
\big\langle\varphi^{(j)}|\operatorname{Op}_N^\W(a_n)|\varphi^{(j)}\big\rangle &=\big\langle\varphi^{(j)}|\operatorname{Op}_N^\W([a_n]_T)|\varphi^{(j)}\big\rangle +\mathcal{O}\Big(\frac{1}{T}\Big(2^{T/4}N^{5/4}e^{-\pi N\beta^2/2^T}\|a\|_{C^0}+\frac{2^T2^{5n}\|a\|_{C^5}}{N\beta^5}\Big)\Big).
\end{align}
The terms in \eqref{eqn:qv-split} are squared, so applying Cauchy--Schwarz, we obtain that
\begin{align*}
\big|\big\langle\varphi^{(j)}|\operatorname{Op}_N^\W([a_n]_T)|\varphi^{(j)}\big\rangle\big|^2 &\le \big\langle\varphi^{(j)}|\operatorname{Op}_N^\W([a_n]_T)^*\operatorname{Op}_N^\W([a_n]_T)|\varphi^{(j)}\big\rangle\\
&= \big\langle\varphi^{(j)}|\operatorname{Op}_N^\W(|[a_n]_T|^2)|\varphi^{(j)}\big\rangle +\mathcal{O}\Big(\frac{\||[a_n]_T|^2\|_{C^4}^2}{N}\Big),\numberthis\label{eqn:cs-term}
\end{align*}
where we also used e.g. \cite[Lemma 3.1]{MOK},
\begin{align}
\|\operatorname{Op}_N^\W(a)\operatorname{Op}_N^\W(b)-\operatorname{Op}_N^\W(ab)\|&\le C\frac{\|a\|_{C^4}\|b\|_{C^4}}{N}.
\end{align}
For the error term in \eqref{eqn:cs-term}, by the Leibniz product formula,
\begin{align*}
\||[a_n]_T|^2\|_{C^4} \le C_4\|[a_n]_T\|_{C^4}^2
&\le C_4\left(\frac{1}{T}\sum_{t=0}^{T-1}\|a_n\circ B^{-t}\|_{C^4}\right)^2\\
&\le C_4\|a\|_{C^4}^22^{8n}\beta^{-8}.
\end{align*}
Finally, using \eqref{eqn:7.5} and \eqref{eqn:cs-term}, we can then bound the quantum variance contribution from the good part $a_n$ as
\begin{multline}\label{eqn:qv-bound}
\frac{4\pi}{N|I(N)|}\sum_{\theta^{(j)}\in I(N)}\big|\big\langle\varphi^{(j)}|\operatorname{Op}_N^\W(a_n)|\varphi^{(j)}\big\rangle\big|^2 \\
\le \frac{8\pi}{N|I(N)|}\sum_{\theta^{(j)}\in I(N)}\big\langle\varphi^{(j)}|\operatorname{Op}_N^\W(|[a_n]_T|^2)|\varphi^{(j)}\big\rangle +\hspace{4cm}
\\
+ \mathcal{O}\bigg(\frac{1}{T^2}\Big(2^{T/4}N^{5/4}e^{-\pi N\beta^2/2^T}\|a\|_{C^0}+\frac{2^T2^{5n}\|a\|_{C^5}}{N\beta^5}\Big)^2\bigg)
+ \mathcal{O}\bigg(\frac{1}{N}\|a\|_{C^4}^42^{16n}\beta^{-16}\bigg).
\end{multline}
By the windowed generalized Weyl law (Theorem~\ref{thm:lweyl}, \eqref{eqn:lweyl}) and the ergodicity convergence rate in Lemma~\ref{lem:erg-rate}, the main term $\frac{2\pi}{N|I(N)|}\sum_{\theta^{(j)}\in I(N)}\big\langle\varphi^{(j)}|\operatorname{Op}_N^\W(|[a_n]_T|^2)|\varphi^{(j)}\big\rangle$ above is
\begin{align*}
\int_{\T^2} \left|[a_n]_T(\x)\right|^2 \,d\x +\mathcal{O}\Big(&(1+\||[a_n]_T|^2\|_{C^3})\frac{1}{|I(N)|J}\Big)\\ &\le 
2\int_{\T^2}\left|[a]_T(\x)\right|^2\,d\x+2\int_{\T^2}\left|[a_n^\mathrm{bad}]_T(\x)\right|^2\,d\x+\mathcal{O}\Big((1+\||[a_n]_T|^2\|_{C^3})\frac{1}{|I(N)|J}\Big)\\
&\le \frac{C\|a\|_{C^1}^2}{T}+\|a\|_\infty^2\mathcal{O}(\beta)+\frac{C(1+\|a\|_{C^3}^22^{6n}\beta^{-6})}{|I(N)|J},\numberthis\label{eqn:qv-good}
\end{align*}
where the second term in the last expression occurs since $a_n^\mathrm{bad}$ is supported on a region of area $\mathcal{O}(\beta)$.
We take $n=n(N)\to\infty$ sufficiently slowly, and  $\beta=\beta(N)\to0$ sufficiently slowly,  so that $\frac{2^{6n}\beta^{-6}}{|I(N)|J}\to0$. 
For example, making no attempt to optimize, one can take $\beta=n^{-1}$ and $n=\frac{1}{100}\log_2(|I(N)|J)$.
Since we need $T\le n$, simply take $T=n$, and then \eqref{eqn:qv-good} is of order 
\[
\frac{\|a\|_{C^3}}{\log(|I(N)|J)}+\frac{\|a\|_\infty^2}{\log(|I(N)|J)}+\frac{1}{|I(N)|J}+\frac{\|a\|_{C^3}^2(\log(|I(N)|J))^6}{(|I(N)|J)^{94/100}}\to0.
\]
These choices also ensure the error term in \eqref{eqn:qv-bound} is $o(1)$.

If one has a specific rate $|I(N)|\to0$, then one may optimize to obtain a better decay estimate.

\subsubsection{Bad part}
For the bad part, we just use the estimate,
\begin{align*}
\big|\big\langle\varphi^{(j)}|\operatorname{Op}_N^\W(a_n^\mathrm{bad})|\varphi^{(j)}\big\rangle\big|^2 &\le \big\langle\varphi^{(j)}|\operatorname{Op}_N^\W(a_n^\mathrm{bad})^*\operatorname{Op}_N^\W(a_n^\mathrm{bad})|\varphi^{(j)}\big\rangle \\
&= \big\langle\varphi^{(j)}|\operatorname{Op}_N^\W(|a_n^\mathrm{bad}|^2)|\varphi^{(j)}\big\rangle +\mathcal{O}\Big(\frac{\||a_n^\mathrm{bad}|^2\|_{C^4}^2}{N}\Big).
\end{align*}
There is the same bound as for good part $a_n$,  $\||a_n^\mathrm{bad}|^2\|_{C^4}\le C\|a\|_{C^4}2^{4n}\beta^{-4}$. Applying the windowed generalized Weyl law,
\begin{align*}
\frac{2\pi}{N|I(N)|}\sum_{\theta^{(j)}\in I(N)}\big|\big\langle\varphi^{(j)}|\operatorname{Op}_N^\W(a_n^\mathrm{bad})|\varphi^{(j)}\big\rangle\big|^2 &\le \frac{2\pi}{N|I(N)|}\sum_{\theta^{(j)}\in I(N)}\big\langle\varphi^{(j)}|\operatorname{Op}_N^\W(|a_n^\mathrm{bad}|^2)|\varphi^{(j)}\big\rangle +\mathcal{O}\Big(\frac{\|a\|_{C^4}^22^{8n}}{N\beta^8}\Big)\\
&=\int_{\T^2}|a_n^\mathrm{bad}(\x)|^2\,d\x+\mathcal{O}\Big(\frac{\|a\|_{C^3}^22^{6n}\beta^{-6}}{|I(N)|J}\Big)+\mathcal{O}\Big(\frac{\|a\|_{C^4}^22^{8n}}{N\beta^8}\Big)\\
&=\mathcal{O}(\beta)\|a\|_\infty^2+o(1).\numberthis\label{eqn:qv-bad}
\end{align*}
Thus combining \eqref{eqn:qv-bound}, \eqref{eqn:qv-good}, and \eqref{eqn:qv-bad} in \eqref{eqn:qv-split}, we obtain
\begin{align}
\lim_{N\to\infty}\frac{2\pi}{N|I(N)|}\sum_{j\in I(N)}\big|\big\langle\varphi^{(j)}|\operatorname{Op}_N^\W(a)|\varphi^{(j)}\big\rangle\big|^2 &=0,
\end{align}
proving Theorem~\ref{thm:qv}.\qed

\section{Extension to $q_N$}\label{sec:qN}

\subsection{Spectral estimates}\label{subsec:qN}
In this section, we state and prove the extension of Theorem~\ref{thm:Pmat} to functions $q_N$.
\begin{thm}[windowed spectral functions]\label{thm:Qmat} 
Let $N\in2\N$, and let $(e^{i\theta^{(j,N)}},\varphi^{(j,N)})_{j}$ be  eigenvalue-eigenvector pairs corresponding to an orthonormal eigenbasis $(\varphi^{(j,N)})_j$ of $\hat{B}_N$. Suppose $I(N)$ is a sequence of intervals in $\R/(2\pi \Z)$ such that
$|I(N)|\log N\to\infty$ as $N\to\infty$.
Let $q_N:\R/(2\pi\Z)\to\C$ be a sequence of $C^2$ functions with $\|q_N''\|_\infty=o(\log N)$, and define the operator
\begin{align*}
Q_{N,I(N)}:=(q_N\Chi_{I(N)})(\hat{B}_N)=\sum_{\theta^{(j,N)}\in I(N)}q_N(\theta^{(j,N)})|\varphi^{(j,N)}\rangle\langle\varphi^{(j,N)}|,
\end{align*}
where $|\varphi^{(j,N)}\rangle\langle\varphi^{(j,N)}|$ is the orthogonal projection onto the eigenstate $\varphi^{(j,N)}$.
Then for at least $N(1-o(|I(N)|))$ coordinates $x\in\intbrr{0:N-1}$, we have the pointwise estimate
\begin{align}\label{eqn:diag-Q}
(Q_{N,I(N)})_{xx}&=\frac{|I(N)|}{2\pi}\left(\frac{1}{|I(N)|}\int_{I(N)}q_N(z)\,dz+ o(1)(1+\|q_N\|_\infty)\right),
\end{align}
and for at least $N^2(1-o(|I(N)|))$ pairs $(x,y)\in\intbrr{0:N-1}^2$, we have the bound,
\begin{align}\label{eqn:off-diag-Q}
(Q_{N,I(N)})_{xy}&=o(1)|I(N)|(1+\|q_N\|_\infty),\; x\ne y,
\end{align}
with asymptotic decay rates uniform over the allowable $x,y$, and the location of $I(N)$.
The points $x$ and pairs $(x,y)$ to avoid are the same as those in Theorem~\ref{thm:Pmat}; more precisely, in terms of sets  defined in Section~\ref{subsec:sets} and parameters defined in \eqref{eqn:param2}, Eq.~\eqref{eqn:diag-Q} holds for $x\not\in\da_{J,\delta,\gamma,N}^W$, and \eqref{eqn:off-diag-Q} holds for $(x,y)\not\in\tilde{A}_{J,\delta,\gamma,N}^W$.
\end{thm}
\begin{proof}
We will extend equations \eqref{eqn:P-diag} and \eqref{eqn:P-offdiag} for $P_{I(N)}$ to those for $Q_{N,I(N)}=(q_N\Chi_{I(N)})(\hat{B}_N)$.
In the definition of $J=(\log_2 N)\varepsilon(N)$ in equation \eqref{eqn:param2}, choose $\varepsilon(N)=\max\left(\big(\frac{\|q_N''\|_\infty}{\log N}\big)^{1/2},\frac{1}{(|I(N)|\log N)^{1/2}}\right)$. Then $\varepsilon(N)\to0$ and $\varepsilon(N)|I(N)|\log N\to\infty$ as required for Definition~\ref{def:param}. We obtain \eqref{eqn:P-diag} and \eqref{eqn:P-offdiag} via Theorem~\ref{prop:mpowers}. This choice of $\varepsilon(N)$ with $\|q_N''\|_\infty=o(\log N)$ implies $\frac{\|q_N''\|_\infty}{J}\to0$ as well, which we will use in the estimates below. In what follows, the error terms $o(1)$ may depend on this rate $\frac{\|q_N''\|_\infty}{J}\to0$.

Instead of $\Chi_{I(N)}$ as for $P_{I(N)}$, we have the product  $q_N\Chi_{I(N)}$. Approximate $q_N$ by its Fourier series partial sum $q_{N,J/2}$ up to degree $J/2$, and use the decomposition,
\begin{align}\label{eqn:q-decomp}
q_{N}\Chi_{I(N)} &= q_N(\Chi_{I(N)}-G_{I(N),J/2}^{(+)}) + (q_N-q_{N,J/2})G_{I(N),J/2}^{(+)} + q_{N,J/2}G_{I(N),J/2}^{(+)},
\end{align}
where $G_{I(N),J/2}^{(+)}$ is the Selberg polynomial defined in \eqref{eqn:G}.
In the last term $q_{N,J/2}G_{I(N),J/2}^{(+)}$, both factors are trigonometric polynomials of degree  at most $J/2$, and so we can compute its application on $\hat{B}_N$ easily. We will show the other two terms are small. For notational convenience let $f_N=q_{N,J/2}G_{I(N),J/2}^{(+)}$. By abuse of notation regarding functions on $\R/(2\pi\Z)$ vs functions on the unit circle in $\C$, then
\begin{align*}
q_{N,J/2}G_{I(N),J/2}^{(+)}(\hat{B}_N) &= \sum_{|k|\le J}\hat{f_N}(k)\hat{B}_N^k \\
&=\left(\frac{1}{2\pi}\int_0^{2\pi}q_{N,J/2}(x)G_{I(N),J/2}^{(+)}(x)\,dx\right)\operatorname{Id}+\sum_{1\le |k|\le J}\hat{f_N}(k)\hat{B}_N^k.\numberthis\label{eqn:q-fourier}
\end{align*}
The integral on the identity term is $\frac{1}{2\pi}\int_{I(N)}q_{N}(x)\,dx+\mathcal{O}\big(\frac{\|q_N\|_\infty}{J}+\frac{|I(N)|\|q_N''\|_\infty}{J}\big)$, using that $q_N$ is $C^2$ so has the Fourier coefficient decay in \eqref{eqn:fourier-bound}.

For the non-identity terms in \eqref{eqn:q-fourier}, the Fourier coefficients for $1\le |k|\le J$ are
\begin{align*}
|\hat{f_N}(k)| &\le \frac{1}{2\pi}\int_0^{2\pi}|q_{N,J/2}(x)||G_{I(N),J/2}^{(+)}(x)|\,dx \\
&\le \frac{\|q_{N,J/2}\|_\infty}{2\pi}\left(|I(N)|+\frac{4\pi}{J}\right) \le \frac{|I(N)|}{2\pi}\left(\|q_{N}\|_\infty+\frac{C\|q_N''\|_\infty}{J}\right)(1+o(1)).
\end{align*}
Restricting to $(x,y)\not\in\tilde{A}_{J,\delta,\gamma,N}^W$ with parameters \eqref{eqn:param2} and using $(\hat{B}_N^k)_{xy}\le 2^{-k}r(N)$ with $r(N)\to0$ (Theorem~\ref{prop:mpowers}) for such $(x,y)$, we obtain  $\left|\sum_{1\le |k|\le J}\hat{f_N}(k)(\hat{B}_N^k)_{xy}\right|=o(|I(N)|)\|q_{N,J/2}\|_\infty$. 
So for $(x,y)\not\in\tilde{A}_{J,\delta,\gamma,N}^W$ and using $\|q_N''\|_\infty J^{-1}\to0$ and $J^{-1}=o(|I(N)|)$,
\begin{align}\label{eqn:A3}
q_{N,J/2}G_{I(N),J/2}^{(+)}(\hat{B}_N)_{xy} &= \frac{\delta_{xy}}{2\pi}\int_{I(N)}q_{N}(t)\,dt +(1+ \|q_N\|_\infty) o(|I(N)|).
\end{align}

So it just remains to show the other terms in \eqref{eqn:q-decomp} applied to $\hat{B}_N$ produce small matrix entries. For the first term $q_N(\Chi_{I(N)}-G_{I(N),J/2}^{(+)})$, applying off-diagonal estimates like in Section~\ref{subsec:off-diag} shows,
\begin{align*}
\left|q_N(\Chi_{I(N)}-G_{I(N),J/2}^{(+)})(\hat{B}_N)_{xy}\right|&= \left|\sum_{j=1}^N q_N(\theta^{(j)}) (G_{I(N),J/2}^{(+)}(\theta^{(j)})-\Chi_{I(N)}(\theta^{(j)}))\langle x|\varphi^{(j)}\rangle\langle\varphi^{(j)}|y\rangle\right|\\
&\le\|q_N\|_\infty \left((G_{I(N),J/2}^{(+)}(\hat{B}_N))_{xx}-(P_{I(N)})_{xx}\right)^{1/2}\left((G_{I(N),J/2}^{(+)}(\hat{B}_N))_{yy}-(P_{I(N)})_{yy}\right)^{1/2} \\
&\le o(1)|I(N)|\|q_N\|_\infty,
\end{align*}
where in the last line we used that $x,y\not\in\da_{J,\delta,\gamma,N}^W$ in order to use the diagonal estimates like in Section~\ref{subsec:diag-power}.

For the second term $(q_N-q_{N,J/2})G_{I(N),J/2}^{(+)}$, use that $q_N$ is $C^2$ and so has good Fourier coefficient decay \eqref{eqn:fourier-bound}.
Then for $x,y\not\in\da_{J,\delta,\gamma,N}^W$,
\begin{align*}
\big|\langle x|(q_N&-q_{N,J/2})G_{I(N),J/2}^{(+)}(\hat{B}_N)|y\rangle\big| \\
&=\left|\sum_{j=1}^N\left(\sum_{|k|\ge \lfloor J/2\rfloor +1}\hat{q_N}(k)e^{2\pi ik\theta^{(j)}}\right)G_{I(N),J/2}^{(+)}(\theta^{(j)})\langle x|\varphi^{(j)}\rangle\langle\varphi^{(j)}|y\rangle\right|\\
&\le \sum_{|k|\ge J/2+1}|\hat{q_N}(k)|\left(\sum_{j=1}^N G_{I(N),J/2}^{(+)}(\theta^{(j)})|\langle x|\varphi^{(j)}\rangle|^2\right)^{1/2}\left(\sum_{j=1}^N G_{I(N),J/2}^{(+)}(\theta^{(j)})|\langle y|\varphi^{(j)}\rangle|^2\right)^{1/2}\\
&\le \frac{C\|q_N''\|_\infty}{J}\frac{|I(N)|}{2\pi}(1+o(1)).
\end{align*}
This is $o(|I(N)|)$ since $\frac{\|q_N''\|_\infty}{J}\to0$.

In conclusion, from \eqref{eqn:q-decomp}, \eqref{eqn:A3}, and the above estimates, we obtain for  $(x,y)\not\in\tilde{A}_{J,\delta,\gamma,N}^W$,
\begin{align*}
q_N\Chi_{I(N)}(\hat{B}_N)_{xy} &= \frac{\delta_{xy}}{2\pi}\int_{I(N)}q_{N}(t)\,dt +(1+ \|q_N\|_\infty) o(|I(N)|),
\end{align*}
as desired.
\end{proof}

\subsection{Generalized Weyl law}\label{appsubsec:weyl}
\begin{sloppypar}
The extension of Theorem~\ref{thm:lweyl} to $q_N$ in \eqref{eqn:lweyl-q} follows the proof in Section~\ref{sec:lweyl-proof}, using $|(Q_{N,I(N)})_{xy}|\le \|q_N\|_\infty$ in general,  and using Theorem~\ref{thm:Qmat}  to obtain $|(Q_{N,I(N)})_{xy}|\le o(1)|I(N)|(1+\|q_N\|_\infty)$ for $(x,y)\not\in\tilde{A}_{J,\delta,\gamma,N}^W$ with $x\ne y$, and that $(Q_{N,I(N)})_{xx}=\frac{|I(N)|}{2\pi}\left(\fint_{I(N)} q_N(z)\,dz+o(1)(1+\|q_N\|_\infty)\right)$ for $(x,x)\in S$.
\end{sloppypar}

\subsection{Eigenvalue counting extension}\label{appsubsec:ev}
For the proof of Corollary~\ref{cor:weyl} Eq.~\eqref{eqn:weyl-qn}, recall that with the choices in \eqref{eqn:param2}, 
\begin{align*}
	\#\da_{J,\delta,\gamma,N}^W &\le C(\gamma N+2^J\delta N+2^JW)=\mathcal{O}(N^{2/3})=o(N|I(N)|).
\end{align*}
Eq.~\eqref{eqn:weyl-qn} then follows from the diagonal entry estimate \eqref{eqn:diag-Q} of Theorem~\ref{thm:Qmat}, along with $|(Q_{N,I(N)})_{xx}|\le \|q_N\|_\infty$, yielding
\begin{align*}
	\operatorname{tr}Q_{N,I(N)} &= \sum_{x\not\in\da_{J,\delta,\gamma,N}^W}\frac{|I(N)|}{2\pi}\left(\fint_{I(N)}q_N(z)\,dz+o(1)(1+\|q_N\|_\infty)\right)+\sum_{x\in\da_{J,\delta,\gamma,N}^W}(Q_{N,I(N)})_{xx}\\
	&=\frac{N|I(N)}{2\pi}\left(\fint_{I(N)}q_N(z)\,dz+o(1)(1+\|q_N\|_\infty)\right).
\end{align*}

\subsection*{Acknowledgments} 
This work was supported by Simons Foundation grant 563916, SM. The author would like to thank Victor Galitski, Abu Musa Patoary, and Amit Vikram for discussions on baker's map quantizations, and the referees for their careful reading and helpful remarks which have improved the presentation and clarity of this paper.

\subsection*{Data availability statement}
We do not analyse or generate any datasets.

\bibliographystyle{abbrv}
\bibliography{baker_weyl.bib}

\begin{thebibliography}{10}

\bibitem{anantharaman}
N.~Anantharaman.
\newblock Quantum ergodicity on regular graphs.
\newblock {\em Comm. Math. Phys.}, 353(2):633--690, 2017.

\bibitem{alm}
N.~Anantharaman and E.~Le~Masson.
\newblock Quantum ergodicity on large regular graphs.
\newblock {\em Duke Math. J.}, 164(4):723--765, 2015.

\bibitem{AN}
N.~Anantharaman and S.~Nonnenmacher.
\newblock Entropy of semiclassical measures of the {W}alsh-quantized baker's
  map.
\newblock {\em Ann. Henri Poincar\'{e}}, 8(1):37--74, 2007.

\bibitem{as}
N.~Anantharaman and M.~Sabri.
\newblock Quantum ergodicity on graphs: from spectral to spatial
  delocalization.
\newblock {\em Ann. of Math. (2)}, 189(3):753--835, 2019.

\bibitem{BV}
N.~L. Balazs and A.~Voros.
\newblock The quantized baker's transformation.
\newblock {\em Ann. Physics}, 190(1):1--31, 1989.

\bibitem{Berard77}
P.~H. B\'erard.
\newblock On the wave equation on a compact {R}iemannian manifold without
  conjugate points.
\newblock {\em Math. Z.}, 155(3):249--276, 1977.

\bibitem{qgraphs}
G.~Berkolaiko, J.~P. Keating, and U.~Smilansky.
\newblock Quantum ergodicity for graphs related to interval maps.
\newblock {\em Comm. Math. Phys.}, 273(1):137--159, 2007.

\bibitem{berry}
M.~V. Berry.
\newblock Regular and irregular semiclassical wavefunctions.
\newblock {\em J. Phys. A}, 10(12):2083--2091, 1977.

\bibitem{bgs}
O.~Bohigas, M.-J. Giannoni, and C.~Schmit.
\newblock Characterization of chaotic quantum spectra and universality of level
  fluctuation laws.
\newblock {\em Phys. Rev. Lett.}, 52(1):1--4, 1984.

\bibitem{bdb}
A.~Bouzouina and S.~De~Bi\`evre.
\newblock Equipartition of the eigenfunctions of quantized ergodic maps on the
  torus.
\newblock {\em Comm. Math. Phys.}, 178(1):83--105, 1996.

\bibitem{exp2}
T.~A. Brun and R.~Schack.
\newblock Realizing the quantum baker's map on a {NMR} quantum computer.
\newblock {\em Phys. Rev. A}, 59:2649--2658, Apr 1999.

\bibitem{Canzani}
Y.~Canzani.
\newblock Monochromatic random waves for general {R}iemannian manifolds.
\newblock In {\em Frontiers in analysis and probability---in the spirit of the
  {S}trasbourg-{Z}\"{u}rich meetings}, pages 1--20. Springer, Cham, 2020.

\bibitem{ChatterjeeGalkowski}
S.~Chatterjee and J.~Galkowski.
\newblock Arbitrarily small perturbations of {D}irichlet {L}aplacians are
  quantum unique ergodic.
\newblock {\em J. Spectr. Theory}, 8(3):909--947, 2018.

\bibitem{ChatterjeeMeckes}
S.~Chatterjee and E.~Meckes.
\newblock Multivariate normal approximation using exchangeable pairs.
\newblock {\em ALEA Lat. Am. J. Probab. Math. Stat.}, 4:257--283, 2008.

\bibitem{deverdiere}
Y.~Colin~de Verdi\`ere.
\newblock Ergodicit\'{e} et fonctions propres du laplacien.
\newblock {\em Comm. Math. Phys.}, 102(3):497--502, 1985.

\bibitem{DBDE}
S.~De~Bi\`evre and M.~Degli~Esposti.
\newblock Egorov theorems and equidistribution of eigenfunctions for the
  quantized sawtooth and baker maps.
\newblock {\em Ann. Inst. H. Poincar\'{e} Phys. Th\'{e}or.}, 69(1):1--30, 1998.

\bibitem{qmaps}
M.~Degli~Esposti and S.~Graffi.
\newblock Mathematical aspects of quantum maps.
\newblock In {\em The mathematical aspects of quantum maps}, volume 618 of {\em
  Lecture Notes in Phys.}, pages 49--90. Springer, Berlin, 2003.

\bibitem{DNW}
M.~Degli~Esposti, S.~Nonnenmacher, and B.~Winn.
\newblock Quantum variance and ergodicity for the baker's map.
\newblock {\em Comm. Math. Phys.}, 263(2):325--352, 2006.

\bibitem{DiaconisFreedman}
P.~Diaconis and D.~Freedman.
\newblock Asymptotics of graphical projection pursuit.
\newblock {\em Ann. Statist.}, 12(3):793--815, 1984.

\bibitem{DyGu}
S.~Dyatlov and C.~Guillarmou.
\newblock Microlocal limits of plane waves and {E}isenstein functions.
\newblock {\em Ann. Sci. \'{E}c. Norm. Sup\'{e}r. (4)}, 47(2):371--448, 2014.

\bibitem{bakergap}
S.~Dyatlov and L.~Jin.
\newblock Resonances for open quantum maps and a fractal uncertainty principle.
\newblock {\em Comm. Math. Phys.}, 354(1):269--316, 2017.

\bibitem{scar-cat}
F.~Faure, S.~Nonnenmacher, and S.~De~Bi\`evre.
\newblock Scarred eigenstates for quantum cat maps of minimal periods.
\newblock {\em Comm. Math. Phys.}, 239(3):449--492, 2003.

\bibitem{phys}
S.~Gnutzmann, J.~P. Keating, and F.~Piotet.
\newblock Eigenfunction statistics on quantum graphs.
\newblock {\em Ann. Physics}, 325(12):2595--2640, 2010.

\bibitem{Gutkin}
B.~Gutkin.
\newblock Entropic bounds on semiclassical measures for quantized
  one-dimensional maps.
\newblock {\em Comm. Math. Phys.}, 294(2):303--342, 2010.

\bibitem{HejhalRackner}
D.~A. Hejhal and B.~N. Rackner.
\newblock On the topography of {M}aass waveforms for {${\rm PSL}(2,{\bf Z})$}.
\newblock {\em Experiment. Math.}, 1(4):275--305, 1992.

\bibitem{HMR}
B.~Helffer, A.~Martinez, and D.~Robert.
\newblock Ergodicit\'{e} et limite semi-classique.
\newblock {\em Comm. Math. Phys.}, 109(2):313--326, 1987.

\bibitem{KeatingMezzadriMonstra}
J.~P. Keating, F.~Mezzadri, and A.~G. Monastra.
\newblock Nodal domain distributions for quantum maps.
\newblock {\em J. Phys. A}, 36:L53--L59, 2003.

\bibitem{Keeler}
B.~Keeler.
\newblock A logarithmic improvement in the two-point {W}eyl law for manifolds
  without conjugate points.
\newblock {\em Ann. Inst. Fourier (Grenoble)}, 74(2):719--762, 2024.

\bibitem{KurlbergRudnick}
P.~Kurlberg and Z.~Rudnick.
\newblock Hecke theory and equidistribution for the quantization of linear maps
  of the torus.
\newblock {\em Duke Math. J.}, 103(1):47--77, 2000.

\bibitem{modmul}
A.~Lakshminarayan.
\newblock Modular multiplication operator and quantized baker's maps.
\newblock {\em Phys. Rev. A}, 76:042330, Oct 2007.

\bibitem{MOK}
J.~Marklof and S.~O'Keefe.
\newblock Weyl's law and quantum ergodicity for maps with divided phase space.
\newblock {\em Nonlinearity}, 18(1):277--304, 2005.
\newblock With an appendix ``Converse quantum ergodicity'' by Steve Zelditch.

\bibitem{Meckes}
E.~S. Meckes.
\newblock Quantitative asymptotics of graphical projection pursuit.
\newblock {\em Electron. Commun. Probab.}, 14:176--185, 2009.

\bibitem{ML2005}
N.~Meenakshisundaram and A.~Lakshminarayan.
\newblock Multifractal eigenstates of quantum chaos and the {T}hue-{M}orse
  sequence.
\newblock {\em Phys. Rev. E}, 71:065303, Jun 2005.

\bibitem{Montgomery2000}
H.~L. Montgomery.
\newblock Harmonic analysis as found in analytic number theory.
\newblock In {\em Twentieth century harmonic analysis---a celebration ({I}l
  {C}iocco, 2000)}, volume~33 of {\em NATO Sci. Ser. II Math. Phys. Chem.},
  pages 271--293. Kluwer Acad. Publ., Dordrecht, 2001.

\bibitem{NonnenmacherAnatomy}
S.~Nonnenmacher.
\newblock Anatomy of quantum chaotic eigenstates.
\newblock In {\em Chaos}, volume~66 of {\em Prog. Math. Phys.}, pages 193--238.
  Birkh\"{a}user/Springer, Basel, 2013.

\bibitem{NonnenmacherZworski}
S.~Nonnenmacher and M.~Zworski.
\newblock Distribution of resonances for open quantum maps.
\newblock {\em Comm. Math. Phys.}, 269(2):311--365, 2007.

\bibitem{pzk}
P.~Pako\'{n}ski, K.~\.{Z}yczkowski, and M.~Ku\'{s}.
\newblock Classical 1{D} maps, quantum graphs and ensembles of unitary
  matrices.
\newblock {\em J. Phys. A}, 34(43):9303--9317, 2001.

\bibitem{RubinSalwen}
R.~Rubin and N.~Salwen.
\newblock A canonical quantization of the baker's map.
\newblock {\em Ann. Physics}, 269(2):159--181, 1998.

\bibitem{RudelsonVershynin-HW}
M.~Rudelson and R.~Vershynin.
\newblock Hanson--{W}right inequality and sub-{G}aussian concentration.
\newblock {\em Electron. Commun. Probab.}, 18:no. 82, 9, 2013.

\bibitem{RudnickSarnak}
Z.~Rudnick and P.~Sarnak.
\newblock The behaviour of eigenstates of arithmetic hyperbolic manifolds.
\newblock {\em Comm. Math. Phys.}, 161(1):195--213, 1994.

\bibitem{Saraceno}
M.~Saraceno.
\newblock Classical structures in the quantized baker transformation.
\newblock {\em Ann. Physics}, 199(1):37--60, 1990.

\bibitem{SaracenoVoros}
M.~Saraceno and A.~Voros.
\newblock Towards a semiclassical theory of the quantum baker's map.
\newblock {\em Phys. D}, 79(2-4):206--268, 1994.

\bibitem{SchackCaves}
R.~Schack and C.~M. Caves.
\newblock Shifts on a finite qubit string: a class of quantum baker's maps.
\newblock {\em Appl. Algebra Engrg. Comm. Comput.}, 10(4-5):305--310, 2000.
\newblock Quantum computing (Schlo\ss Dagstuhl, 1998).

\bibitem{Schwartz-thesis}
N.~Schwartz.
\newblock {\em Statistical Properties of Quantized Toral Automorphisms}.
\newblock PhD thesis, Universit\'e Paris-Saclay, France, 2022.

\bibitem{ScottCaves}
A.~J. Scott and C.~M. Caves.
\newblock Entangling power of the quantum baker's map.
\newblock {\em J. Phys. A}, 36(36):9553--9576, 2003.

\bibitem{Selberg}
A.~Selberg.
\newblock {\em Collected papers. {V}ol. {II}}.
\newblock Springer-Verlag, Berlin, 1991.

\bibitem{pw}
L.~Shou.
\newblock Pointwise {W}eyl law for graphs from quantized interval maps.
\newblock {\em Ann. Henri Poincar\'{e}}, (24):2833–2875, 2023.

\bibitem{bakernumerics}
L.~Shou, A.~Vikram, and V.~Galitski.
\newblock {Spectral anomalies and broken symmetries in maximally chaotic
  quantum maps}.
\newblock {\em SciPost Phys.}, 18:151, 2025.

\bibitem{shnirelman}
A.~I. Snirel'man.
\newblock Ergodic properties of eigenfunctions.
\newblock {\em Uspehi Mat. Nauk}, 29(6(180)):181--182, 1974.

\bibitem{TracyScott}
M.~M. Tracy and A.~J. Scott.
\newblock The classical limit for a class of quantum baker's maps.
\newblock {\em J. Phys. A}, 35(39):8341--8360, 2002.

\bibitem{Vaaler}
J.~D. Vaaler.
\newblock Some extremal functions in {F}ourier analysis.
\newblock {\em Bull. Amer. Math. Soc. (N.S.)}, 12(2):183--216, 1985.

\bibitem{Vershynin}
R.~Vershynin.
\newblock {\em High-dimensional probability}, volume~47 of {\em Cambridge
  Series in Statistical and Probabilistic Mathematics}.
\newblock Cambridge University Press, Cambridge, 2018.

\bibitem{exp}
Y.~S. Weinstein, S.~Lloyd, J.~Emerson, and D.~G. Cory.
\newblock Experimental implementation of the quantum {B}aker's map.
\newblock {\em Phys. Rev. Lett.}, 89(15):157902, 4, 2002.

\bibitem{zelditch}
S.~Zelditch.
\newblock Uniform distribution of eigenfunctions on compact hyperbolic
  surfaces.
\newblock {\em Duke Math. J.}, 55(4):919--941, 1987.

\bibitem{Zelditch2}
S.~Zelditch.
\newblock Real and complex zeros of {R}iemannian random waves.
\newblock In {\em Spectral analysis in geometry and number theory}, volume 484
  of {\em Contemp. Math.}, pages 321--342. Amer. Math. Soc., Providence, RI,
  2009.

\bibitem{Zelditch-book}
S.~Zelditch.
\newblock {\em Eigenfunctions of the {L}aplacian on a {R}iemannian manifold},
  volume 125 of {\em CBMS Regional Conference Series in Mathematics}.
\newblock Published for the Conference Board of the Mathematical Sciences,
  Washington, DC; by the American Mathematical Society, Providence, RI, 2017.

\bibitem{zworski}
M.~Zworski.
\newblock {\em Semiclassical analysis}, volume 138 of {\em Graduate Studies in
  Mathematics}.
\newblock American Mathematical Society, Providence, RI, 2012.

\end{thebibliography}

\end{document}